\documentclass[a4paper,12pt]{article}

\usepackage[english]{babel}
\usepackage[utf8x]{inputenc}
\usepackage[T1]{fontenc}
\usepackage{enumerate}
\usepackage{multicol}
\usepackage[labelformat=simple]{subcaption}

\usepackage{natbib}
\usepackage[a4paper,top=3cm,bottom=3cm,left=2.3cm,right=2.3cm,marginparwidth=1.7cm, footnotesep=1cm]{geometry}

\usepackage{tikz}
\usetikzlibrary{shapes,decorations,arrows,calc,arrows.meta,fit,positioning,automata, positioning, quotes}

\usepackage{dsfont}
\usepackage{float}
\usepackage{amssymb}
\usepackage{amsthm}
\usepackage{thmtools, thm-restate}
\usepackage{amsmath}
\usepackage{graphicx}
\usepackage[colorlinks=true, allcolors=black]{hyperref}
\usepackage{mathtools}
\usepackage{booktabs}

\newcommand{\MiMit}{\ensuremath{{\mathbf{M}_i\mathbf{M}_i^T}}}

\newcommand{\indep}{\perp \!\!\! \perp}  
\newcommand{\dsep}{\ensuremath{{ \perp_{G}}}}
\newcommand{\notdsep}{\ensuremath{{ \not\! \perp_{G}}}}

\newcommand{\norm}[1]{\lVert #1 \rVert}
\newcommand{\doop}{{\text{do}}}

\newcommand{\Prob}{\ensuremath{{\mathbb P}}}
\newcommand{\R}{\ensuremath{{\mathbb R}}}
\newcommand{\E}{\ensuremath{{\mathbb E}}}

\newtheorem{definition}{Definition}[section]

\newtheorem{ass}{Assumption}
\newtheorem{example}[definition]{Example}

\title{A Graphical Approach to Treatment Effect Estimation with Observational Network Data}

\usepackage{authblk}
\author[1]{Meta-Lina Spohn\footnote{Authors with equal contribution.}}
\author[2]{Leonard Henckel$^*$}
\author[1]{Marloes H.~Maathuis}
\affil[1]{Seminar for Statistics, ETH Zurich, Switzerland}
\affil[2]{School of Mathematics and Statistics, University College Dublin, Ireland}
\date{}                   
\begin{document}




\date{}
\maketitle
\bigskip
\begin{abstract}
We propose an easy-to-use adjustment estimator for the effect of a treatment based on observational data from a single (social) network of units. The approach allows for interactions among units within the network, called interference, and for observed confounding. We define a simplified causal graph that does not differentiate between units, called generic graph. Using valid adjustment sets determined in the generic graph, we can identify the treatment effect and build a corresponding estimator. We establish the estimator's consistency and its convergence to a Gaussian limiting distribution at the parametric rate under certain regularity conditions that restrict the growth of dependencies among units. We empirically verify the theoretical properties of our estimator through a simulation study and apply it to estimate the effect of a strict facial-mask policy on the spread of COVID-19 in Switzerland.
\end{abstract}




\vspace{0.6cm}

\section{Introduction}

One common assumptions in causal inference is the stable unit-treatment value assumption (SUTVA) \citep{rubin1978bayesian}. Part of SUTVA is the no-interference assumption \citep{cox}, that is, the assumption that the treatment status of a unit may only influence the outcome of that unit and not the outcome of any other unit. In practical applications, however, interference is common as units can interact. For example the vaccination of a person against an infectious disease also helps protect the health of that person's social contacts \citep{perez}. Another example are students interacting in their class at school, such that a child's test score at the end of the year is not only affected by the student's math instruction type, but also the instruction type other students in the class received \citep{hong2008causal}.

Ignoring interference can lead to faulty conclusions \citep[e.g.][]{sofrygin}. It is therefore important to account for interference when estimating treatment effects in networks, but there are three major difficulties in doing so. First, in 
the classical i.i.d.~setting with a binary treatment and $N$ independent units, there is one counterfactual treatment for each of the $N$ units, namely the treatment that was not assigned to that unit. In the interference setting with $N$ dependent units, there are $2^N-1$ counterfactual treatments for each unit, namely one for every possible treatment assignment of the $N$ units except the observed one. As a result, it is less clear how to define causal effects such as the average treatment effect (ATE) \citep{rubin1977assignment}. One standard target effect in the literature is the difference between the average expected unit-specific outcome of two different hypothetical stochastic treatment interventions that assign treatments to units independently with a user-specified treatment probability \citep[c.f.~][]{munoz2012population}. We call this class of effects global treatment effects. A special case is the average total treatment effect \citep{imbens}, also called the global average treatment effect (GATE) \citep{chin}, which contrasts the hypothetical interventions of treating all units versus treating none.

Second, to account for interference, it is generally necessary to model it by making assumptions on the specific structure and pathways of the interference \citep{imbens}. A common assumption in the literature, called partial interference \citep{sobel}, is that interference takes place in arbitrary form but only within nonoverlapping groups of units and not across these groups \citep[e.g.,][]{tchetgen}. Another is to describe the dependencies among units via a known interaction network graph, in which the nodes represent the units and the edges indicate relations between units that facilitate interaction, such as geographical proximity. 
Given a network graph, it is possible to model interference by summarizing a unit's dependence on the treatment of other units through a finite set of functions that are common to all the units in the population and depend on the network graph. 
In the structural equation model (SEM) framework these functions are generally called interference features \citep{manski1993, chin}. 

Third, in many applications only observational data may be available. In such settings, it is important to account for confounding when estimating treatment effects in networks \citep[][]{tchetgen,sofrygin,Corinne2022}. This is a difficult problem, but one that has been extensively studied in the i.i.d.~setting. For example, given knowledge of the underlying causal structure in the form of a causal graph, the class of adjustment sets that correct for confounding has been fully characterized \citep{ema}. The members of this class are called valid adjustment sets. It is, however, unclear under what conditions we can apply these graphical results from the i.i.d.~setting to settings with interference.

In this paper we consider the estimation of treatment effects based on observational data from networks with interference and within-unit confounding, that is, confounding between a unit's treatment and its outcome. The target effects are global treatment effects and
we work in the framework of SEMs. Concretely, we assume a class of SEMs $S_e$ on explicit variables \citep{zhang2022causal}, that is, covariates $\mathbf{C}_i$, treatments $W_i$ and outcomes $Y_i$, for all units $i=1,\ldots,N$. With such explicit SEMs we can represent the simultaneous presence of within-unit confounding and interference. Based on the explicit  directed acyclic graph (DAG) $G_e$ corresponding to $S_e$, we define the \textit{generic graph} $\mathcal{G}$ on the variables $\boldsymbol{C}$, $W$ and $Y$ by stacking the subgraphs for each unit $i$ of $G_e$. 
While the generic graph is not as informative as the original explicit DAG, we show that for the class of explicit SEMs we consider, the generic graph can be used to identify a class of causal effects we call \textit{unit-specific effects}. Global treatment effects, however, do not belong to the class of unit-specific effects. To obtain an identification result for global treatment effects, we therefore adopt the approach of modelling interference via interference features, a finite set of known functions of the known interaction network graph and the treatment vector of the entire population. In addition, we assume a linear outcome model. Based on these two assumptions we show that we can rewrite the target global treatment effect as the weighted average of unit-specific effects, where the weights can be explicitly computed or approximated, and the unit-specific effects can be identified from the generic graph $\mathcal{G}$, using tools from causal graphical models. In particular, we will use graphical criteria for valid adjustment sets. Based on this identification result we then propose an adjustment estimator for global average treatment effects. Under some regularity conditions that limit the growth of dependencies between units, we prove that this estimator is consistent and converges at the parametric $N^{-1/2}$-rate to a Gaussian limiting distribution with finite variance that can be consistently estimated.

Methodologically, our work is most similar to the work of \cite{chin} and \cite{zhang2022causal}, with whom we share the assumption of a linear outcome model. \cite{chin}, however, does not allow for confounding and \cite{zhang2022causal} are interested in the bias of estimating the ATE if the units were isolated. Conceptually, our work is also related to the semi-parametric estimation of treatment effects in networks. This literature, however, either makes simplifying assumptions under which graphical identifiability results are trivial and/or estimate other treatment effects \citep{Laan2,Corinne2022}.
Finally, there exists literature on identifying treatment effects in networks using explicit DAGs \citep{ogburn2014}. However, the number of nodes in these graphs grows with the number of units $N$ and as a result 
these graphs become difficult to use for larger sample sizes. 

The paper is organized as follows. In Section \ref{sec:prelim} we introduce the set-up and the target effects. In Section \ref{sec:identification} we introduce the generic graph and interference features and discuss the identification of treatment effects using the generic graph. In Section \ref{sec:consistency} we showcase the use of the generic graph by proposing an adjustment estimator. In Section \ref{sec:empirical} we perform a simulation study to verify the properties of the adjustment estimator and apply our methods to estimate the effect of a strict facial-mask policy on the spread of COVID-19 in Switzerland. The code for the simulation study and the facial-mask policy analysis is available at 
\href{https://github.com/henckell/InterferenceCode}{github.com/henckell/InterferenceCode} 
and proofs are provided in the appendix.

\section{Preliminaries}\label{sec:prelim}

We consider a population of $N$ units. For each unit $i$ we observe a binary treatment $W_i \in \{0,1\}$, a possibly multivariate vector of covariates $\mathbf{C}_i \in \mathbb{R}^{D_C}$, and a continuous outcome $Y_i \in \mathbb{R}$. We aim to estimate a causal effect of the treatment on the outcome accounting for the presence of within-unit confounding
and interference. We illustrate the problem in Example \ref{example1}.

\begin{example}\label{example1}
We consider people interacting in their social network. Given a person $i$, the severity of a viral disease is the outcome $Y_i$ and the vaccination against the disease is the treatment $W_i$. Each person chooses whether to take the vaccination or not. This decision is governed by the variable $C_i$, representing the severity of previous infections with the disease. 
The variable $C_i$ also affects the outcome, that is, the severity of a new infection with the disease. Thus, $C_i$ constitutes within-unit confounding through the confounding path $W_i \leftarrow C_i \rightarrow Y_i$. In addition, the treatment status $W_j$ of person $j$ affects person $j$'s viral load. If person $i$ is in close contact with person $j$, person $j$'s viral load may in turn affect the severity of disease $Y_i$ of person $i$. The fact that the treatment of person $j$ affects the outcome of person $i$ constitutes interference.
\end{example}

Throughout the paper, we consider two types of random variables. Variables that distinguish between units, called \textit{explicit variables}, and variables that do not, called \textit{generic variables}. For example, we use $W_i$ to denote the explicit treatment variable for unit $i$ and $W$ to denote the generic treatment variable that does not distinguish between units. We use $\boldsymbol{\bar{W}}=(W_1,\dots,W_N)^T \in \R^N$ to denote the treatment vector for all units. To ease notation we use $\boldsymbol{\bar{W}}_{-i}=(W_1,\dots,W_{i-1},W_{i+1},\dots,W_{N})^T \in \R^{(N-1)}$ to denote the treatment vector for all units but $i$. We use the same notation for random vectors, e.g., $\boldsymbol{C}_i \in \R^{D_C}, \boldsymbol{C} \in \R^{D_C}  ,\boldsymbol{\bar{C}} \in \R^{N \times D_C}$ and $\boldsymbol{\bar{C}}_{-i} \in \R^{(N-1) \times D_C}$ are the explicit vector of covariates for unit $i$, the generic vector of covariates, the matrix of covariates for all units and the matrix of covariates for all units but $i$, respectively.

We now introduce our set-up and the treatment effects that are the targets of inference. Please refer to Appendix \ref{appendix:prelims} for the graphical notions used throughout, such as the definition of a DAG or the latent projection.

\subsection{Explicit Models with Confounding and Interference}

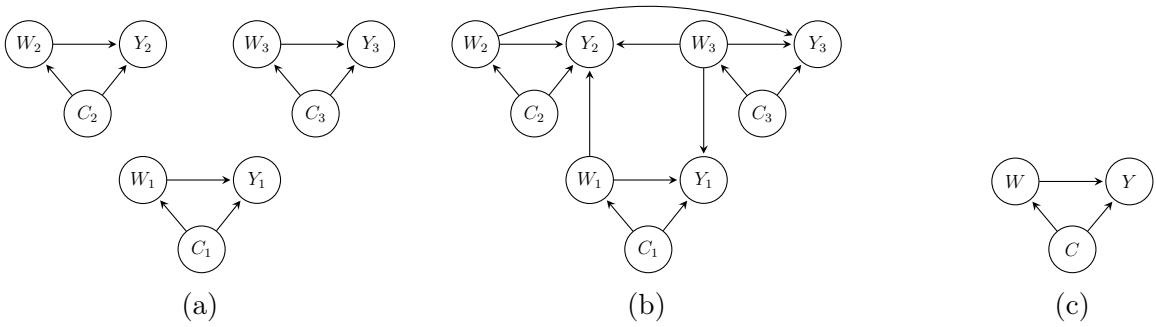
\begin{figure}[t]
\centering
\begin{subfigure}[t]{0.35\textwidth}
\centering
\begin{tikzpicture}
[>=stealth,
   shorten >=1pt,
   node distance=2cm,
   on grid,
   scale=0.6, 
   transform shape
    ]
	\node[state] (w1) at (0,0) {$W_1$};
	\node[state] (c1) at ($ (w1) + (-50:2)$) {$C_1$};
	\node[state] (y1) at  ($ (w1) + (0:2.5)$) {$Y_1$};
 \node[state] (w3) at (2.5,3) {$W_3$};
	\node[state] (c3) at ($ (w3) + (-50:2)$) {$C_3$};
	\node[state] (y3) at  ($ (w3) + (0:2.5)$) {$Y_3$};
\node[state] (w2) at (-2.5,3) {$W_2$};
	\node[state] (c2) at ($ (w2) + (-50:2)$) {$C_2$};	
	\node[state] (y2) at  ($ (w2) + (0:2.5)$) {$Y_2$};

	\path[->,draw]
      (w2) edge (y2)
      (c1) edge (y1)
      (c2) edge (w2)
      (c2) edge (y2)
      (c1) edge (w1)
      (c3) edge (w3)
      (c3) edge (y3)
      (w1) edge (y1)
      (w3) edge (y3);
\end{tikzpicture}
\caption{}
\label{subfig: explicit DAG isolated}
\end{subfigure}
  \begin{subfigure}[t]{0.35\textwidth}
  \centering
   \begin{tikzpicture}
   [>=stealth,
   shorten >=1pt,
   node distance=2cm,
   on grid,
   scale=0.6, 
   transform shape
    ]
	\node[state] (w1) at (0,0) {$W_1$};
	\node[state] (c1) at ($ (w1) + (-50:2)$) {$C_1$};
	\node[state] (y1) at  ($ (w1) + (0:2.5)$) {$Y_1$};
        \node[state] (w3) at (2.5,3) {$W_3$};
	\node[state] (c3) at ($ (w3) + (-50:2)$) {$C_3$};
	\node[state] (y3) at  ($ (w3) + (0:2.5)$) {$Y_3$};
        \node[state] (w2) at (-2.5,3) {$W_2$};
	\node[state] (c2) at ($ (w2) + (-50:2)$) {$C_2$};
	\node[state] (y2) at  ($ (w2) + (0:2.5)$) {$Y_2$};

	\path[->,draw]
      (w2) edge (y2)
      (c1) edge (y1)
      (c2) edge (w2)
      (c2) edge (y2)
      (c1) edge (w1)
      (c3) edge (w3)
      (c3) edge (y3)
      (w1) edge (y1)
      (w3) edge[] (y2)
      (w2) edge[bend left=20] (y3)
      (w3) edge[] (y1)
      (w1) edge[] (y2)
      (w3) edge (y3);
    \end{tikzpicture}
    \caption{}
    \label{subfig: explicit DAG interference}
\end{subfigure}
\hfill
\begin{subfigure}[t]{0.25\textwidth}
\centering
\begin{tikzpicture}
[>=stealth,
   shorten >=1pt,
   node distance=2cm,
   on grid,
   scale=0.6, 
   transform shape
    ]
    \node[state] (w) at (0,0)  {$W$};
    \node[state] (y) at (2.5,0) {$Y$};
    \node[state] (c) at (1.25,-1.5) {$C$};
    
	\path[->,draw]
      (w) edge (y)
      (c) edge (w)
      (c) edge (y);
      \end{tikzpicture}
      \caption{}
      \label{subfig: generic DAG}
      \end{subfigure}
\captionsetup{labelformat=default}
\caption{An explicit DAG $G$ without interference \subref{subfig: explicit DAG isolated}, an explicit DAG with interference \subref{subfig: explicit DAG interference} and the corresponding generic DAG for both \subref{subfig: generic DAG}.}
\label{fig: no features}
\end{figure}

In the classical setting where units do not interact with each other, it is common to write structural equations which do not specify or differentiate between units. This implicitly assumes $(1)$ that the structural equations and therefore the causal relationships between variables of a unit are the same for all units and $(2)$ that there are no causal effects between units. To make these assumptions explicit, we consider structural equations on the explicit variables $\mathbf{C}_i$, $W_i$ and $Y_i$, for $i=1,\ldots,N$. We define an \textit{explicit SEM} $S_e$ as a SEM on explicit variables, and call the DAG $G_e$ corresponding to $S_e$ an \textit{explicit DAG}. An example of an explicit DAG $G_e$ on $N=3$ units is shown in Figure \ref{subfig: explicit DAG isolated}. It represents the classical case with no interference between the three units. Explicit SEMs allow us to characterize settings where the assumptions $(1)$ and/or $(2)$ are violated. An example of an explicit DAG $G_e$ on $N=3$ units with interference between all three units is shown in Figure \ref{subfig: explicit DAG interference}.

We limit our considerations to a specific class of explicit SEMs $S_e$ with interference, defined in the following assumption. For simplicity we restrict ourselves to recursive SEMs, that is, we do not allow cycles.

\begin{ass} \label{def:sem0}
The explicit recursive SEM $S_e$ with within-unit confounding and interference is given by 
 \begin{align*}
     &\boldsymbol{C}_i  \leftarrow g_{\boldsymbol{C}}(\boldsymbol{C}_i, \boldsymbol{\epsilon}_{C_i}), \ W_i \leftarrow g_W(\boldsymbol{C}_i,\epsilon_{W_i}) \text{ and }
     Y_i \leftarrow g_{Y,i}(\boldsymbol{C}_i, W_i, \boldsymbol{\bar{W}}_{-i}, \epsilon_{Y_i}),
 \end{align*}
for each unit $i=1,\ldots,N$. We assume that $\boldsymbol{\epsilon}_{C_i},\epsilon_{W_i}$ and $\epsilon_{Y_i}$ are jointly independent error terms with expectation zero, and that their distribution does not depend on $i$.
\end{ass}

Under Assumption \ref{def:sem0}, a unit $i$ may depend on another unit $j$ solely through interference. In the explicit DAG, this means that we allow edges from $W_j$ to $Y_i$ for $j \neq i$, but no other between-unit edges. Furthermore, $g_{\boldsymbol{C}}(\cdot)$ and $g_W(\cdot)$ do not depend on $i$, that is, they are functions common to all units.

\subsection{Target Treatment Effects} \label{sec:estimands}

We consider hypothetical stochastic interventions or policies, where the treatments are assigned independently to each unit with some fixed probability $\pi\in [0,1]$ \citep[e.g.][]{munoz2012population,haneuse2013estimation,sofrygin}. We denote such a stochastic intervention with $\mathrm{do}(\boldsymbol{\bar{W}} \overset{\text{i.i.d.}}{\sim} \text{Bern}(\pi))$
using the do-notation by \citet{pearl2009causality}. Due to interference between the units, $\E[Y_i \mid \text{do}(\boldsymbol{\bar{W}} \overset{\text{i.i.d.}}{\sim} \text{Bern}(\pi))]$ may differ for $i=1,\ldots,N$, and we therefore consider their average. 
The causal effect of interest is the contrast under two different stochastic interventions:
\begin{equation*}
\tau_N(\pi,\eta):= \frac{1}{N}\sum_{i=1}^N\left( \E[Y_i \mid \text{do}(\boldsymbol{\bar{W}} \overset{\text{i.i.d.}}{\sim} \text{Bern}(\pi))]- \E[Y_i \mid \text{do}(\boldsymbol{\bar{W}} \overset{\text{i.i.d.}}{\sim} \text{Bern}(\eta))] \right).    
\end{equation*}
We call the effect $\tau_N(\pi,\eta)$ a \textit{global treatment effect}, as it considers a simultaneous intervention on all units. The GATE, $\tau_N(1, 0) $, is a special case.

\section{Identification of the Target Treatment Effects} \label{sec:identification}

While explicit DAGs can be used for causal inference, they become complex for even a moderate number of units $N$, since the number of nodes is increasing in the number of units.
In the classical setting, where there are no causal effects between different units, we overcome this difficulty by implicitly stacking the induced subgraphs for each unit in the explicit DAG $G_e$ to obtain the conventional DAG $G$ on variables that are not indexed by $i$. 

In this section, we first generalize this stacking approach to any explicit DAG $G_e$. We refer to the resulting graph as a \textit{generic graph $\mathcal{G}$}. While the generic graph is not as informative as the explicit DAG, we show that for the class of explicit SEMs satisfying Assumption \ref{def:sem0}, the generic graph can be used to identify a class of causal effects we call \textit{unit-specific effects}. However, the global treatment effect $\tau_N(\pi,\eta)$ does not belong to this class. We overcome this problem by modelling interference via interference features \citep{manski1993,chin} and showing that $\tau_N(\pi,\eta)$ can be decomposed into a weighted average of unit-specific effects. 
The weights in this decomposition only depend on our choice of interference features and can be explicitly computed or approximated. The unit-specific effects, on the other hand, can be identified with graphical criteria for valid adjustment sets \citep{ema}, applied to the generic graph. Thus, this approach allows us to identify the target treatment effect $\tau_N(\pi,\eta)$ in the presence of within-unit confounding and interference.

\subsection{Generic Graphs} \label{sec:ident_nofeat}

\begin{definition}[Generic graph] \label{def:genericDAG}
Consider an explicit DAG $G_e$ on explicit variables $\mathbf{V}_i$, $i=1,\ldots, N$. The corresponding generic graph $\mathcal{G}=(\boldsymbol{V},\boldsymbol{E})$ is defined as follows: if $A_i \rightarrow B_i$ for any $i$ in $G_e$, then add $ A \rightarrow B$ to the edge set $\boldsymbol{E}$.
\end{definition}

Definition \ref{def:genericDAG} is similar to the isolated interaction model from Definition 2 of \cite{zhang2022causal}, in that it only considers within-unit edges. For explicit SEMs satisfying Assumption \ref{def:sem0}, the generic graph is guaranteed to be a DAG, since the induced subgraph on $\mathbf{V}_i$ of the explicit DAG is the same for all units. For illustration, consider the two explicit DAGs $G_e$ in Figures \ref{subfig: explicit DAG isolated} and \ref{subfig: explicit DAG interference}. Both have the same generic graph $\mathcal{G}$, shown in Figure \ref{subfig: generic DAG}. 

For explicit SEMs satisfying Assumption \ref{def:sem0}, the generic graph $\mathcal{G}$ is also the latent projection as defined by \citet{richardson2003markov} of $G_e$ over $\boldsymbol{\bar{V}}_{-i}$. This implies that interventional distributions $f(b \mid \doop( \boldsymbol{A}=\boldsymbol{a}))$ for $\{B\} \cup \mathbf{A} \subseteq \mathbf{V}_i$, that is, belonging to the same unit $i$, factorize according to $\mathcal{G}$. In other words, $\mathcal{G}$ is a causal DAG for each $\boldsymbol{V}_i$ \citep{pearl2009causality,evans2016graphs}. We also provide an explicit proof of this fact in Proposition \ref{lemma:strucpreserv} of Appendix \ref{sec:ident_nofeat}. Thus, we can use the generic graph $\mathcal{G}$ corresponding to an explicit SEM satisfying Assumption \ref{def:sem0} to identify the following class of causal effects.

\begin{definition}[Unit-specific effects]
Consider an explicit SEM $S_e$ on explicit variables $\boldsymbol{\bar{V}}=\bigcup_{i=1}^N \boldsymbol{V}_i$. Let $\boldsymbol{A} \subset \boldsymbol{\bar{V}}$ and $B \in \boldsymbol{\bar{V}} \setminus \mathbf{A}$ and consider causal effects of $\boldsymbol{A}$ on $B$ of the form 
 $\frac{\partial}{\partial \boldsymbol{a}}\E[B \mid \text{do}(\mathbf{A}=\boldsymbol{a})]$ or $\E[B \mid \text{do}(\mathbf{A}=\boldsymbol{a})]- \E[B \mid \text{do}(\mathbf{A}=\boldsymbol{a'})]$ for some $\boldsymbol{a}\neq \boldsymbol{a'}$ in the support of $\boldsymbol{A}$. We say that the causal effect is unit-specific if $\mathbf{A} \cup \{B\} \subset \mathbf{V}_i$ for some unit $i$.
\end{definition}

An example of an average of unit-specific effects is the expected average treatment effect (EATE) \citep{savje}, given by 
\begin{equation*}
    \frac{1}{N}\sum_{i=1}^N \big(\E[Y_i\mid \doop(W_i=1)]- \allowbreak \E[Y_i\mid \doop(W_i=0)]\big).
\end{equation*}
It captures how, on average, the outcome of a unit is affected by its own treatment. We are, however, interested in the global treatment effect $\tau_N(\pi,\eta)$, which
involves interventions on $\boldsymbol{\bar{W}}$. Since $\tau_N(\pi,\eta)$ is not unit-specific, it may not be identifiable using $\mathcal{G}$. In the following section we show that we can overcome this problem if we impose additional structure on the interference mechanism by introducing interference features \citep{manski1993,chin}.


\subsection{Interference Features} \label{sec:ident_feat}

We refine Assumption \ref{def:sem0} on the explicit SEM $S_e$, by assuming that the outcome model takes the form $
    Y_i \leftarrow g_Y'(\boldsymbol{C}_i, W_i, \boldsymbol{X}_i,\epsilon_{Y_i})$,
where $\boldsymbol{X}_i$ are possibly multivariate interference features capturing the effect of $\boldsymbol{\bar{W}}_{-i}$ on $Y_i$, and the function $g_Y'(\cdot)$ does not depend on $i$.

Specifically, we assume that for each unit $i$ the effect of $\boldsymbol{\bar{W}}_{-i}$ on $Y_i$ is modulated by a known and nonrandom \textit{interaction network graph} $I^N$, with nodes $i=1, \ldots, N$ representing the units and edges representing interaction between the respective units, such as, for example, friendship ties in a social group or geographical adjacency between agricultural fields. We use the convention that all edges in $I^N$ are directed, with an edge $i \rightarrow j$ indicating that there is an interaction from $i$ to $j$. If the interaction is bi-directional, we add the edge $j \rightarrow i$ in $I^N$. 
We also use $I^N$ to denote the corresponding adjacency matrix $I^N\in \{0,1\}^{N\times N}$, where $I^N_{ij}=1$ if there is an edge $i \rightarrow j$. 

Similarly to \cite{manski1993} and \cite{chin}, we assume that for each unit $i=1,\ldots, N$, the interference features $\boldsymbol{X}_i=(X_{i1},\ldots, X_{iP})^T$ are functions of $I^N$ and the treatment vector $\boldsymbol{\bar{W}}_{-i}$, that is, $X_{ik} = h^k(\boldsymbol{\bar{W}}_{-i},I^N)$ for $k = 1,\ldots, P,$
where the functions $h^k(\cdot): \mathbb{R}^{(N-1) \times (N \times N)} \mapsto \mathbb{R}$ do not depend on $i$.

\begin{figure}[t]

\centering
\begin{subfigure}[t]{0.3\textwidth}
\centering
\begin{tikzpicture}
[>=stealth,
   shorten >=1pt,
   node distance=2cm,
   on grid,
   scale=0.7, 
   transform shape,
   align=center,minimum size=3em]
    \node[state] (1) at (0,0) {$1$};
    \node[state] (2) [right =5em of 1]{$2$};
    \node[state] (3) [above= 5em of 2]{$3$};

    \path[->,draw]
    (1) edge (2)
    (3) edge (1)
    (2) edge[bend right=10] (3)
    (3) edge[bend right=10] (2);
\end{tikzpicture}
\caption{}
\label{subfig: interaction graph small}
\end{subfigure}
\begin{subfigure}[t]{0.3\textwidth}
\centering
\begin{tikzpicture}
 [>=stealth,
   shorten >=1pt,
   node distance=2cm,
   on grid,
   auto,
   scale=0.7, 
   transform shape,
   align=center,minimum size=3em]
[xshift=-5cm]
	\node[state] (2) at (0,0) {$2$};
	\node[state] (1) [left =5em of 2] {$1$}; 
	\node[state] (3) [right = 5em of 2] {$3$};   
	\node[state] (4) [above = 5em of 3] {$4$};   
	\node[state] (5) [above = 5em of 2] {$5$};   
	\node[state] (6) [above = 5em of 1] {$6$};

	\path[->,draw]
    (5) edge (1)
    (2) edge[bend right=10] (5)
    (5) edge[bend right=10] (2)
    (6) edge (5)
    (2) edge (3);
    \end{tikzpicture}
    \caption{}
    \label{subfig: interaction graph big}
\end{subfigure}
\begin{subfigure}[t]{0.3\textwidth}
\centering
    \begin{tikzpicture}
     [>=stealth,
   shorten >=1pt,
   node distance=2cm,
   on grid,
   auto,
   scale=0.7, 
   transform shape,
   align=center,minimum size=3em]
	\node[state] (2) at (0,0) {$2$};
	\node[state] (1) [left =5em of 2] {$1$}; 
	\node[state] (3) [right = 5em of 2] {$3$};   
	\node[state] (4) [above = 5em of 3] {$4$};   
	\node[state] (5) [above = 5em of 2] {$5$};   
	\node[state] (6) [above = 5em of 1] {$6$};

	\path[-,draw]
    (1) edge (2)
    (1) edge (6)
    (2) edge (6)
    (5) edge (3);
	\end{tikzpicture}
 \caption{}
 \label{subfig: dependency graph}
\end{subfigure}
\caption{An interaction network graph on $N=3$ units \subref{subfig: interaction graph small}, an interaction network graph on $N=6$ units \subref{subfig: interaction graph big} and the interference dependency graph corresponding to the latter if we consider as interference feature the fraction of treated parents of parents \subref{subfig: dependency graph}.}
\label{fig:network}
\end{figure}
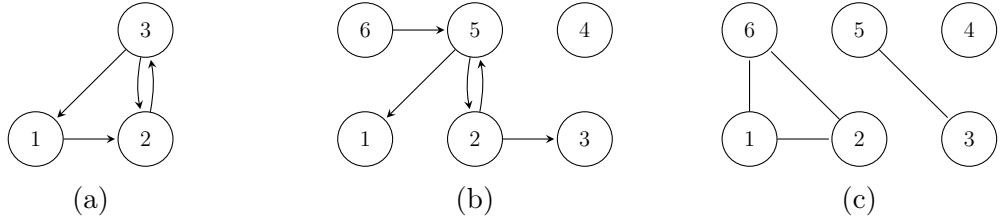

\begin{example}\label{example1_revisited}
A natural interference feature is the fraction of treated parents
\begin{align}
X_{i}^1 \coloneqq \frac{1}{\left|\mathcal{N}_i^{1}\right|}\sum_{j \in \mathcal{N}_i^{1}}W_j,\label{formula:feat1}
\end{align}
where $\mathcal{N}_i^{1} \coloneqq \left\{ j\in \{1,\ldots,N\}: I^N_{ji}=1\right\}$
denotes the parents of $i$ in $I^N$. Another possible interference feature is the indicator that at least $50\%$ of the parents of $i$ are treated.
To model interference beyond the parents in $I^N$, we may, for example, consider the fraction of treated parents of parents
\begin{align}
    X_{i}^2 &\coloneqq \frac{1}{\left| \mathcal{N}_i^{(2)}\right|}\sum_{j \in \mathcal{N}_i^{(2)}}W_j, \label{formula:feat3}
\end{align}
where
$\mathcal{N}_i^{(2)} \coloneqq \left\{ j \in \{1,\ldots,N\} \setminus \{i\}: \text{there exists l such that } I^N_{jl}I^N_{li}=1\right\}$. 
\end{example}

We further assume that the outcome equation is linear and may differ for treated units ($W_i=1$) and untreated units $(W_i=0)$, but is common to units within these two treatment groups. Specifically, we assume that
\begin{equation}
Y_i \leftarrow  (1-W_i) (1,\boldsymbol{X}_i^T) \boldsymbol{\beta}_{0} + W_i (1,\boldsymbol{X}_i^T) \boldsymbol{\beta}_{1} + \boldsymbol{C}_i^T\boldsymbol{\gamma} + \epsilon_{Y_i},
\label{eq:outcomewithbeta} 
\end{equation}
where $\boldsymbol{X}_i \coloneqq (X_{i1}, X_{i2}, \ldots, X_{iP})^T
 \in \mathbb{R}^{P}$, and $\boldsymbol{\beta}_{0}, \boldsymbol{\beta}_{1} \in \mathbb{R}^{P+1}$. We summarize our assumptions on the model in Assumption \ref{def:sem}, where we also reparameterize the outcome model with coefficients $\boldsymbol{\alpha}_{0} = \boldsymbol{\beta}_{0}$ and $\boldsymbol{\alpha}_{1} = \boldsymbol{\beta}_{1}-\boldsymbol{\beta}_{0}$ as these are the parameters we will estimate. 

\begin{ass}
\label{def:sem}
The explicit recursive SEM $S_e$ with interference features and linear outcome model is given by 
 \begin{align}
    &\boldsymbol{C}_i \leftarrow g_{\boldsymbol{C}}(\boldsymbol{C}_i, \boldsymbol{\epsilon}_{C_i}), \ W_i \leftarrow g_W(\boldsymbol{C}_i,\epsilon_{W_i}),\notag\\
    &\boldsymbol{X}_i \leftarrow h(\boldsymbol{\bar{W}}_{-i},I^N), \ \boldsymbol{O}_i \leftarrow W_i \boldsymbol{X}_i \text{ and} \notag\\
     &Y_i \leftarrow (1,\boldsymbol{X}_i^T) \boldsymbol{\alpha}_{0}  + (W_i ,\boldsymbol{O}^T_i)  \boldsymbol{\alpha}_{1} +  \boldsymbol{C}_i^T \boldsymbol{\gamma}+ \epsilon_{Y_i},\label{eq:outcomeX}
 \end{align}
for each unit $i=1,\ldots,N$. We assume that the interaction network graph $I^N$ and the functions $h(\cdot) := (h^1(\cdot), \ldots, h^P(\cdot))^T$ are known. Further, we assume that $\boldsymbol{\epsilon}_{C_i},\epsilon_{W_i}$ and $\epsilon_{Y_i}$ are jointly independent error terms with expectation zero, and that their distributions does not depend on $i$.
\end{ass}

\begin{figure}[t!]
\centering

\begin{subfigure}[t]{0.6\textwidth}
\begin{tikzpicture}
[>=stealth,
   shorten >=1pt,
   node distance=2cm,
   on grid,
   scale=0.5, 
   transform shape
    ]
    \node[state] (y1)  at (0,0)  {$Y_2$};
    \node[state] (o1) [above =6em of y1] {$\boldsymbol{O}_{2}$};
    \node[state] (w1)  [left =6em of o1]  {$W_2$};
    \node[state] (x1) [right =6em of o1] {$\boldsymbol{X}_{2}$}; 
    \node[state] (c11) [below =6em of x1]  {$C_{21}$};
    \node[state] (c12) [left =6em of y1]  {$C_{22}$};
    \node[state] (c13) [above =6em of w1]  {$C_{23}$};

    \node[state] (y2)  at (-7,0)  {$Y_1$};
    \node[state] (o2) [above =6em of y2] {$\boldsymbol{O}_{1}$};
    \node[state] (w2)  [left =6em of o2]  {$W_1$}; 
    \node[state] (x2) [right =6em of o2] {$\boldsymbol{X}_{1}$}; 
    \node[state] (c21) [below =6em of x2]  {$C_{11}$};
    \node[state] (c22) [left =6em of y2]  {$C_{12}$};
    \node[state] (c23) [above =6em of w2]  {$C_{13}$};

    \node[state] (y3)  at (7,0)  {$Y_3$};
    \node[state] (o3) [above =6em of y3] {$\boldsymbol{O}_{3}$};
    \node[state] (w3)  [left =6em of o3]  {$W_3$};
    \node[state] (x3) [right =6em of o3] {$\boldsymbol{X}_{3}$}; 
    \node[state] (c31) [below =6em of x3]  {$C_{31}$};
    \node[state] (c32) [left =6em of y3]  {$C_{32}$};
    \node[state] (c33) [above =6em of w3]  {$C_{33}$};
 
	 \path[->,draw]
       (w1) edge (y1)
       (c12) edge (w1)
       (x1) edge (y1)
      (o1) edge (y1)
       (x1) edge (o1)
       (w1) edge (o1)
       (c11) edge[bend left] (c12)
       (c11) edge (y1)
       (c13) edge (w1)
       
      (w2) edge (y2)
       (c22) edge (w2)
       (x2) edge (y2)
      (o2) edge (y2)
       (x2) edge (o2)
       (w2) edge (o2)
       (c21) edge[bend left] (c22)
       (c21) edge (y2)
       (c23) edge (w2)

      (w3) edge (y3)
       (c32) edge (w3)
       (x3) edge (y3)
      (o3) edge (y3)
       (x3) edge (o3)
       (w3) edge (o3)
       (c31) edge[bend left] (c32)
       (c31) edge (y3)
       (c33) edge (w3)

      (w3) edge[bend right] (x2)
      (w3) edge[] (x1)
      (w1) edge[bend right=20] (x3)
      (w2) edge[bend right=20] (x1);    
\end{tikzpicture}
\caption{}
\label{subfig: explicit features}
\end{subfigure}
\begin{subfigure}[t]{0.35\textwidth}
\centering
\begin{tikzpicture}
[>=stealth,
   shorten >=1pt,
   node distance=2cm,
   on grid,
   scale=0.5, 
   transform shape]
    \node[state] (y)  at (0,0)  {$Y$};
    \node[state] (o) [above =6em of y] {$\boldsymbol{O}$};
    \node[state] (w)  [left =6em of o]  {$W$};
    \node[state] (x) [right =6em of o] {$\boldsymbol{X}$};
    \node[state] (c1) [right =6em of y]  {$C_{1}$};
    \node[state] (c2) [left =6em of y]  {$C_{2}$};
    \node[state] (c3) [above =6em of w]  {$C_{3}$};

	\path[->,draw]
       (w) edge (y)
       (c2) edge (w)
       (x1) edge (y)
      (o) edge (y)
       (x) edge (o)
       (w) edge (o)
       (c1) edge[bend left] (c2)
       (c1) edge (y)
       (c3) edge (w);
\end{tikzpicture}
\caption{}
\label{subfig: generic features}
\end{subfigure}
\caption{An explicit DAG compatible with an explicit SEM satisfying Assumption \ref{def:sem} \subref{subfig: explicit features} and the corresponding generic graph \subref{subfig: generic features}.}
\label{fig: feature DAGs}
\end{figure}
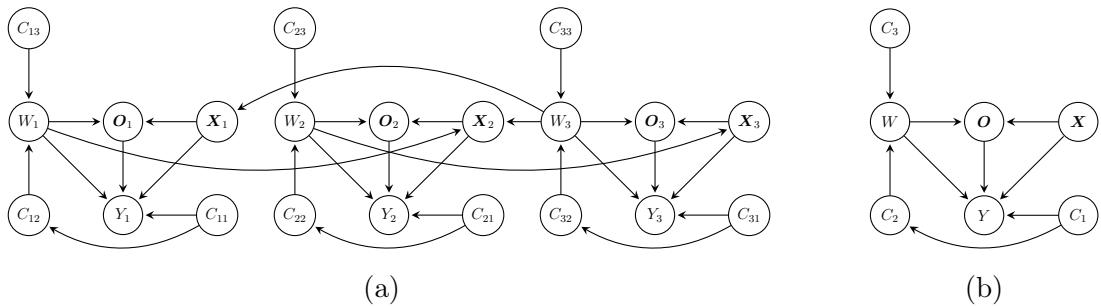

The interference features are a tool to model the interference mechanisms and are not unique. We also only need to know the features up to shift and scale (see Lemma \ref{lemma:invariance} in Appendix \ref{app:proofssecidentwithfeat}). The feature model is flexible, in that we allow for arbitrary features and arbitrary combinations of them, as long as the explicit SEM $S_e$ respects Assumption \ref{def:sem}. We do, however, impose further conditions on the asymptotic behaviour of the features in Section \ref{sec:consistency}. 

Given an explicit SEM respecting Assumption \ref{def:sem}, we consider a corresponding explicit DAG $G_e$ with possibly multivariate nodes for $\boldsymbol{X}_i$ and $\boldsymbol{O}_i$. We interpret the structural equation of a multivariate node in $G_e$ as the vector of structural equations of each of the variables in the node. 
An intervention $\doop(\boldsymbol{X}_i =\boldsymbol{x})$ on a multivariate node is given by simultaneously replacing all structural equations of the vector of structural equations by the vector $\boldsymbol{x}$. Treating $\boldsymbol{X}_i$ and $\boldsymbol{O}_i$ as single nodes in $G_e$ implies that the corresponding generic graph $\mathcal{G}$ also contains single nodes for $\boldsymbol{X}$ and  $\boldsymbol{O}$. The generic graph coincides with the latent projection of $G_e$ over $\boldsymbol{\bar{V}}_{-i} $ for a given unit $i$. Therefore, $\mathcal{G}$ can again be interpreted causally for each $\mathbf{V}_i$, $i=1,\ldots, N$ in the sense that interventional distributions $f(b \mid \doop( \boldsymbol{A}=\boldsymbol{a}))$ for $\{B\} \cup \mathbf{A} \subseteq \mathbf{V}_i$ for some $i$ factorize according to $\mathcal{G}$  for all units $i$. We also provide a proof of this fact in Proposition \ref{lemma:strucpreservwithfeat} in Appendix \ref{sec:ident_feat} and note that it does not hold if we treat each $X_{ik}$ as an individual node, that is, allow for interventions on proper subsets of $\mathbf{X}_i$.  

\begin{example} \label{example:explicitDAG}
Consider a model on three units and suppose that for each unit the explicit SEM takes the form
\begin{align*}
\begin{split}
&C_{i1}  \leftarrow -2 + \epsilon_{C_{i1}}, \ C_{i2}  \leftarrow 2C_{i1} + \epsilon_{C_{i2}} , \ C_{i3}  \leftarrow 0.5 + \epsilon_{C_{i3}}, \\
&W_i \sim \mathrm{Bern}\left( \frac{1}{1+ \exp{(-C_{i2}- 5 C_{i3})}} \right),  \\
 &\boldsymbol{X}_i \leftarrow h(\boldsymbol{\bar{W}}_{-i},I^N), \ \boldsymbol{O}_i \leftarrow W_i \boldsymbol{X}_i \text{ and}\\
&Y_i \leftarrow (1,\boldsymbol{X}_i^T)\boldsymbol{\alpha}_{0}  + (W_i ,\boldsymbol{O}^T_i)  \boldsymbol{\alpha}_{1} +  1.5 C_{i1}+ \epsilon_{Y_i}, 
\end{split}
\end{align*}
where $I^N$ is the interaction graph given in Figure \ref{subfig: interaction graph small} and $h(\boldsymbol{\bar{W}}_{-i},I^N)$ is the fraction of treated parents as defined in equation \eqref{formula:feat1}. Clearly, this explicit SEM satisfies Assumption \ref{def:sem}. The corresponding explicit and generic DAGs are shown shown in Figure \ref{fig: feature DAGs}.  
\end{example}

Based on Assumption \ref{def:sem} we can decompose the global treatment effect $\tau_N(\pi,\eta)$ into a weighted linear combination of unit-specific effects that we can identify using the interpretation of the generic graph $\mathcal{G}$ as a causal DAG. The decomposition is analogous to the decomposition result by \cite{chin} for the setting without confounding.

\begin{restatable}[Decomposition of global treatment effects]{prop}{lemmataureformulation}
\label{lemma:reformulation}
Let $S_e$ be an explicit SEM satisfying Assumption \ref{def:sem}. Then
 \begin{align}
 \label{eq: weighted}
\tau_N(\pi, \eta) &= \boldsymbol{\omega}^N_{0}(\pi,\eta)^T \boldsymbol{\alpha}_{0} + \boldsymbol{\omega}^N_{1}(\pi,\eta)^T(\boldsymbol{\alpha}_{0} +\boldsymbol{\alpha}_{1} ) ,
\end{align}
where
\begin{align*}
\begin{split}
 \boldsymbol{\omega}^N_{0}(\pi,\eta)^T = \frac{1}{N}\sum_{i=1}^N & \left( (1-\pi) \E[(1,\boldsymbol{X}^T_i)\mid \doop(\boldsymbol{\bar{W}}_{-i}  \overset{\text{i.i.d.}}{\sim} \text{Bern}(\pi))] \right. \\ 
 &\quad\quad \left. - (1-\eta) \E[(1,\boldsymbol{X}^T_i)\mid \doop(\boldsymbol{\bar{W}}_{-i} \overset{\text{i.i.d.}}{\sim} \text{Bern}(\eta))] \right) \text{ and} \\
  \boldsymbol{\omega}^N_{1}(\pi,\eta)^T = \frac{1}{N}\sum_{i=1}^N &\left( \pi \E[(1,\boldsymbol{X}^T_i)\mid \doop(\boldsymbol{\bar{W}}_{-i}  \overset{\text{i.i.d.}}{\sim} \text{Bern}(\pi))] \right. \\ 
  &\quad\quad \left.-  \eta \E[(1,\boldsymbol{X}^T_i) \mid \doop(\boldsymbol{\bar{W}}_{-i}  \overset{\text{i.i.d.}}{\sim} \text{Bern}(\eta))] \right).
\end{split}
\end{align*}
\end{restatable}

The weights $\boldsymbol{\omega}^N_{0}(\pi,\eta)$ and $\boldsymbol{\omega}^N_{1}(\pi,\eta)$ are functions of the expected value of the interference features $\mathbf{X}_{i}$ under the stochastic interventions on $\boldsymbol{\bar{W}}$ with probabilities ${\pi}$ and $\eta$, respectively. Even though the effect of $\boldsymbol{\bar{W}}$ on $\boldsymbol{X}_i$ is not unit-specific, we can exploit our knowledge of the interaction network graph $I^N$ and the interference function $h(\cdot)$ to either compute $\E[(1,\boldsymbol{X}^T_i) \mid \doop(\boldsymbol{\bar{W}}_{-i} \overset{\text{i.i.d.}}{\sim} \text{Bern}(\pi))]$ and $\E[(1,\boldsymbol{X}^T_i) \mid \text{do}(\boldsymbol{\bar{W}}_{-i} \overset{\text{i.i.d.}}{\sim}\text{Bern}(\eta))]$ in closed-form or approximate them with a simulation, holding $I^N$ and $h(\cdot)$ fixed and randomly drawing $\boldsymbol{\bar{W}}$ with probability $\pi$ or $\eta$, respectively. Effectively, we absorb the nonunit-specific part of our target effect $\tau_N(\pi,\eta)$ in the computable weights $\boldsymbol{\omega}^N_{1}(\pi,\eta)$ and $\boldsymbol{\omega}^N_{0}(\pi,\eta)$, and as a result we only need to estimate $\boldsymbol{\alpha}_{0}$ and $\boldsymbol{\alpha}_{1}$. We now show that $(\boldsymbol{\alpha}^T_{0},\boldsymbol{\alpha}^T_{1})$ is the unit-specific joint total effect of $(1,\boldsymbol{X}^T_i, W_i, \boldsymbol{O}^T_i)$ on $Y_i$ for all units $i=1,\ldots, N$. Here we treat the intercept term in equation \eqref{eq:outcomeX} as an additional nonrandom cause of $Y_i$ that we may intervene on. We do so for notational convenience, since the intercept's causal effect cancels in equation \eqref{eq: weighted} and is therefore irrelevant for computing $\tau_N(\pi, \eta)$.

\begin{restatable}[Total joint effect]{lemma}{lemmaisolatedeffectfeat}
\label{lemma:totaljointeffet} 
Let $S_e$ be an explicit SEM satisfying Assumption \ref{def:sem}. Then $(\boldsymbol{\alpha}^T_{0},\boldsymbol{\alpha}^T_{1})$ is the total joint effect of $(1,\boldsymbol{X}^T_i, W_i, \boldsymbol{O}^T_i)$ on $Y_i$.
\end{restatable}

Since $(\boldsymbol{\alpha}^T_{0},\boldsymbol{\alpha}^T_{1})$ is a unit-specific effect we can identify it using the generic graph $\mathcal{G}$ employing the graphical characterization of valid adjustment sets \citep{ema} from the causal graphical models literature. The following theorem summarizes our main identification result.

\begin{restatable}[Identification]{thm}{theoremident}
\label{theorem:identificantion} 
Let $S_e$ be an explicit SEM satisfying Assumption \ref{def:sem}. Then $\tau_N(\pi,\eta) = \boldsymbol{\omega}^N_{0}(\pi,\eta)^T \boldsymbol{\alpha}_{0}+ \boldsymbol{\omega}^N_{1}(\pi,\eta)^T(\boldsymbol{\alpha}_{0} +\boldsymbol{\alpha}_{1} )$,
where the weights $ \boldsymbol{\omega}^N_{0}(\pi,\eta)$ and $ \boldsymbol{\omega}^N_{1}(\pi,\eta)$ are computable, and $(\boldsymbol{\alpha}^T_{0},\boldsymbol{\alpha}^T_{1})$ is the total joint effect of $(1,\boldsymbol{X}^T_i, W_i, \boldsymbol{O}^T_i)$ on $Y_i$ in $S_e$ for all $i=1,\ldots, N$, and can be identified via adjustment in the generic graph $\mathcal{G}$.
\end{restatable}

The proof of Theorem \ref{theorem:identificantion} uses that we can interpret the generic graph causally, in the sense that the truncated factorization formula holds for unit-specific effects (see Proposition \ref{lemma:strucpreservwithfeat} in Appendix \ref{app:proofssecidentwithfeat}). This also implies that identification of effects is possible through other graphical tools for causal DAGs such as the frontdoor-criterion \citep{pearl1995} or instrumental variables \citep{brito2002new,LeoIV}. We focus on adjustment for simplicity and leave these alternatives for future research.

\section{Estimation of Target Treatment Effects}\label{sec:consistency}

Based on the identification result in Theorem \ref{theorem:identificantion}, we propose an adjustment estimator for the causal effect $\tau_N(\pi,\eta)$. In order to derive asymptotic properties for this estimator, we need to make restrictions on the behavior of the interaction network graph and the feature functions. As a tool to make these restrictions, we first introduce the interference dependency graph \citep{savje}.

\subsection{Interference Dependency Graph} \label{sec:depgraph}
As discussed before, we consider settings where the units exhibit interference via interference features $\boldsymbol{X}_i$ that are functions of the other units' treatment vector $\boldsymbol{\bar{W}}_{-i}$ and the interaction network graph $I^N$. Since we do not restrict the interference functions to be local in $I^N$, the absence of an edge $i \leftarrow j$ or $i \rightarrow j$ in $I^N$ does not necessarily indicate independence between any variable $V_i \in \mathbf{V}_i$ and any $V_j \in \boldsymbol{V}_j$. We use an additional undirected graph called the \textit{interference dependency graph} in which the absence of an edge $i - j$ does imply independence between $\mathbf{V}_i$ and $\boldsymbol{V}_j$. Dependency graphs are a standard approach to characterize dependencies between random variables \citep[e.g.][]{chen,baldiRinott}. We use a specific version, namely the interference dependency graph on networks as proposed by \cite{savje}, which is a function of the interaction network graph $I^N$ and the feature functions $h^k(\cdot)$, $k= 1,\ldots,P$. The following definition is written for general $U_{ik} \leftarrow h^k(\boldsymbol{\bar{W}}_{-i},I^N)$, but we mostly consider the case $U_{ik}= X_{ik}$. 

\begin{definition}[Interference Dependency Graph] \label{def:depgraph} 
Consider a treatment vector $\boldsymbol{\bar{W}}$ and an interaction network graph $I^N$. Given $P$ functions $h^1(\cdot),\ldots, h^P(\cdot)$, let $\boldsymbol{\bar{U}}$ be the matrix with entries $U_{ik} \leftarrow h^k(\boldsymbol{\bar{W}}_{-i},I^N)$ for $i=1 ,\ldots, N$ and $k =1,\ldots, P$, and let $\boldsymbol{U}_j$ denote the $j$th row of $\boldsymbol{\bar{U}}$.
We characterize the interference dependency graph by its adjacency matrix $D_{ij}(\boldsymbol{\bar{U}}, \boldsymbol{\bar{W}}) \in \{0,1\}^{N \times N}$, where for two units $i\neq j$ it holds that $D_{ij}(\boldsymbol{\bar{U}}, \boldsymbol{\bar{W}}) = D_{ji}(\boldsymbol{\bar{U}}, \boldsymbol{\bar{W}}) =1$, if one of the following conditions holds: (a) $W_i$ affects $\boldsymbol{U}_j$, (b) $W_j$ affects $\boldsymbol{U}_i$ or (c) $\boldsymbol{U}_i$ and $\boldsymbol{U}_j$ are affected by some $W_l$, $l \in \{1,\ldots, N\} \setminus \{i, j\}$.
\end{definition}

Here, affect means that $W_i$ appears in the generating equation of $\mathbf{U}_j$. By definition, the interference dependency graph $D(\boldsymbol{\bar{U}}, \boldsymbol{\bar{W}})$ is undirected, that is, it does not reflect whether the dependence between units $i$ and $j$ arises because $W_i$ causally affects $\boldsymbol{U}_j$ or $W_j$ causally affects $\boldsymbol{U}_i$. Consider $D(\boldsymbol{\bar{X}}, \boldsymbol{\bar{W}})$, the interference dependency graph for the interference features $\boldsymbol{\bar{X}}$. By Assumption \ref{def:sem}, interference between units may only occur via the features $\boldsymbol{X}_i$ and therefore the absence of an edge $i - j$ in $D(\boldsymbol{\bar{X}}, \boldsymbol{\bar{W}})$ implies that $\boldsymbol{V}_i$ and $\boldsymbol{V}_j$ are independent.

\begin{example}\label{ex:depgraph} 
Consider the interaction network graph $I^N$ in Figure~\ref{subfig: interaction graph big}. Suppose we choose as interference feature the fraction of treated parents of parents as defined in equation \eqref{formula:feat3}. The resulting dependency graph $D(\boldsymbol{\bar{X}}, \boldsymbol{\bar{W}})$ is given in Figure~\ref{subfig: dependency graph}. 
\end{example}

\subsection{Estimating Treatments Effects via Adjustment}\label{section:observed}

We propose an estimator of $\tau_N(\pi,\eta)$ based on an adjustment estimator for $\boldsymbol{\alpha}=(\boldsymbol{\alpha}^T_{0},\boldsymbol{\alpha}^T_{1})^T$. Let $\mathbf{Z}$ be a valid adjustment set relative to $(\{\boldsymbol{X}, W, \boldsymbol{O}\},Y)$ in the generic graph $\mathcal{G}$. For each unit $i$, let $
 \mathbf{M}_i:= 
 (1,\boldsymbol{X}^T_i, W_i, \boldsymbol{O}^T_i, \mathbf{Z}^T_i)^T$
and consider the OLS-estimator
\begin{equation}\label{alphahatols}
\boldsymbol{\hat{\alpha}}^{\text{full}}=(\mathbf{\bar{M}}^{T}\mathbf{\bar{M}})^{-1}\mathbf{\bar{M}}^{T}\boldsymbol{\bar{Y}}.
\end{equation}
We denote the components of $\boldsymbol{\hat{\alpha}}^{\text{full}}$ corresponding to $(1,\boldsymbol{X}^T)$ by $\boldsymbol{\hat{\alpha}}_{0}$ and those corresponding to $(W,\boldsymbol{O}^T)$ by $\boldsymbol{\hat{\alpha}}_{1}$. 
Given $\boldsymbol{\hat{\alpha}}_{0}$ and $\boldsymbol{\hat{\alpha}}_{1}$, we estimate $\tau_N(\pi,\eta)$ by
\begin{equation}\label{pluginest}
  \hat{\tau}_N(\pi, \eta) = \boldsymbol{\omega}^N_{0}(\pi,\eta)^T \boldsymbol{\hat\alpha}_{0} + \boldsymbol{\omega}^N_{1}(\pi, \eta)^T(\boldsymbol{\hat\alpha}_{0} + \boldsymbol{\hat\alpha}_{1}) ,
\end{equation}
where $\boldsymbol{\omega}^N_{0}(\pi,\eta)$ and $\boldsymbol{\omega}^N_{1}(\pi,\eta)$ can either be computed in closed-form or can be approximated through simulation. The following theorem shows that under mild assumptions on the interference features and their dependency graph, the estimator $\hat{\tau}_N(\pi, \eta)$ is consistent for $\tau_N(\pi,\eta)$.

\begin{restatable}[Consistency]{thm}{propconsistencyols}
\label{prop:consistency-ols}
Consider a sequence of explicit SEMs $S_e^N$ and corresponding interaction network graphs $I^N$, satisfying Assumption \ref{def:sem} such that the $S_e^N$ only differ in $I^N$ and $N$.
Let $G^N_e$ be the corresponding explicit DAGs, let $\mathbf{Z}$ be a valid adjustment set relative to $( \{\boldsymbol{X},W,\boldsymbol{O}\}, Y)$ in the generic graph $\mathcal{G}$ common to all $G^N_e$, let $\mathbf{M}_i= (1,\boldsymbol{X}^T_i, W_i, \boldsymbol{O}^T_i, \mathbf{Z}^T_i)^T$ and let $\hat{\tau}_N(\pi, \eta)$ be as defined in equation \eqref{pluginest}. Then,
$\hat{\tau}_N(\pi, \eta) -\tau_N(\pi,\eta)  \xrightarrow[]{P} 0$,
given that 
\begin{enumerate}
    \item[i)] the limits $\lim_{N\rightarrow \infty}\frac{1}{N}\sum_{i = 1}^N
  \E[\boldsymbol{X}_i \mid \doop(\boldsymbol{\bar{W}}_{-i} \overset{\text{i.i.d.}}{\sim} \text{Bern}(\theta))]$ for $\theta=\pi$ and $\theta=\eta$ exist, 
    \item[ii)] $d_{\text{max}}(N) \in o(N)$, where $d_{\text{max}}(N) := \max_{i \in \{1,\dots, N\}} \sum_{j=1}^ND_{ij}(\boldsymbol{\bar{X}}, \boldsymbol{\bar{W}})$ is the maximal degree in the interference dependency graph, holds 
\end{enumerate}
 and in addition the following regularity conditions hold:
\begin{enumerate}
    \item[iii)] $\E[Y_i^4] < \infty$ and $\E\left[\norm{\mathbf{M}_i}^4\right] < \infty$ for $i= 1,\ldots, N$, where $\norm{\cdot}$ denotes the Euclidean norm, 
     \item[iv)] $ \E\left[\mathbf{M}_i\mathbf{M}_i^T\right] < \infty$ is invertible for $i= 1,\ldots, N$,
    \item[v)] $ \lim_{N\rightarrow \infty}\frac{1}{N}\sum_{i=1}^N \E\left[\mathbf{M}_i\mathbf{M}_i^T\right] = \Sigma_{\mathbf{M} \mathbf{M}} < \infty$ elementwise, where $\Sigma_{\mathbf{M} \mathbf{M}}$ is invertible and
    \item[vi)] $\E[\mathbf{P_i} \mid \mathbf{Z}_i]=\delta^T \mathbf{Z}_i$  for $i= 1,\ldots, N$, and some matrix $\delta$, where $\mathbf{P}_i = \mathrm{pa}(Y_i,G_e) \setminus \{\boldsymbol{X}_i,W_i,\boldsymbol{O}_i\}$ .
\end{enumerate}
\end{restatable}

We require Condition i) to ensure that the limit of the target effect $\tau_N(\pi,\tau)$ exists. We require Condition ii) to ensure that a weak law of large number holds for the estimator $\boldsymbol{\hat{\alpha}}^{\text{full}}$. Both conditions are implicit restriction on the sequence of $I^N$ and the interference functions $h^k(\cdot)$ and allow us to avoid explicitly modelling them. For example, if the feature function is the number of treated parents, than Condition i) implies that the average number of parents in $I^N$ converges. We discuss Condition ii) more thoroughly in Example \ref{example: maximal degree local}. The other four conditions are more standard statistical regularity conditions. 

The next theorem shows that under a stricter set of assumptions the estimator $\hat{\tau}_N(\pi, \eta)$ is also asymptotically normal.

\begin{restatable}[Asymptotic Normality]{thm}{propnormalityols}
\label{prop:normality-ols}
Consider a sequence of explicit SEMs $S_e^N$ and corresponding interaction network graphs $I^N$, satisfying Assumption \ref{def:sem} such that the $S_e^N$ only differ in $I^N$ and $N$.
Let $G^N_e$ be the corresponding explicit DAGs, let $\mathbf{Z}$ be a valid adjustment set relative to $( \{\boldsymbol{X},W,\boldsymbol{O}\}, Y)$ in the generic graph $\mathcal{G}$ common to all $G^N_e$, let $\boldsymbol{M}=
\{\boldsymbol{X},W,\boldsymbol{O},\boldsymbol{Z}\}$ and let $\hat{\tau}_N(\pi, \eta)$ be as defined in equation \eqref{pluginest}. Then,
$\sqrt{N}\Big(\hat{\tau}_N(\pi, \eta) -\tau_N(\pi,\eta)\Big)\xrightarrow[]{d} \mathcal{N}(0, \sigma^2),$
given that the conditions from Theorem \ref{prop:consistency-ols} hold, 
\begin{enumerate}
    \item[i)] $d_{\text{max}}(N) \in o(N^{1/4})$,  where $d_{\text{max}}(N) := \max_{i \in \{1,\dots, N\}} \sum_{j=1}^ND_{ij}(\boldsymbol{\bar{X}}, \boldsymbol{\bar{W}})$ is the maximal degree in the interference dependency graph, holds
\end{enumerate}
and in addition the following regularity conditions hold:
\begin{enumerate}
    \item[ii)] $\E[Y_i^8] < \infty$ and $\E\left[\norm{\mathbf{M}_i}^8\right] < \infty$ for $i= 1,\ldots, N$ and
    \item[iii)] $\lim_{N\rightarrow \infty}\frac{1}{N}\sum_{i=1}^N\E\left[ \epsilon_i^2\mathbf{M}_i \mathbf{M}_i ^T\right] =\Sigma_{\epsilon^2\mathbf{M}
    \mathbf{M}} < \infty$, where $\epsilon_i :=Y_i - \mathbf{M}_i^T\boldsymbol{\alpha}^{\text{full}}$, with population level regression coefficients $\boldsymbol{\alpha}^{\text{full}}$ from the regression of $Y_i$ on $\mathbf{M}_i$.
\end{enumerate}
The asymptotic variance 
$\sigma^2$ is finite and given by 
\begin{equation*}
\sigma^2 = \begin{pmatrix}
\boldsymbol{\omega}_{0}(\pi,\eta) + \boldsymbol{\omega}_{1}(\pi,\eta ) \vspace{0.1cm} \\ 
\boldsymbol{\omega}_{1}(\pi,\eta ) \vspace{0.1cm}\\
\boldsymbol{0} \vspace{0.1cm} 
\end{pmatrix}^T\Sigma_{\mathbf{M} \mathbf{M}}^{-1}\Sigma_{\epsilon^2\mathbf{M} \mathbf{M}} \Sigma_{\mathbf{M}\mathbf{M}}^{-1}\begin{pmatrix}
\boldsymbol{\omega}_{0}(\pi,\eta) + \boldsymbol{\omega}_{1}(\pi,\eta ) \vspace{0.1cm} \\ 
\boldsymbol{\omega}_{1}(\pi,\eta ) \vspace{0.1cm}\\
\boldsymbol{0} \vspace{0.1cm} 
\end{pmatrix},
\end{equation*}
where $\boldsymbol{\omega}_{0}(\pi,\eta ) = \lim_{N \rightarrow \infty} \boldsymbol{\omega}^N_{0}(\pi,\eta )$,$\boldsymbol{\omega}_{1}(\pi,\eta )= \lim_{N \rightarrow \infty} \boldsymbol{\omega}^N_{1}(\pi,\eta )$, and
$ \boldsymbol{0}$ denotes a vector of zeros in $\R^{|\mathbf{Z}|}$.
\end{restatable}

We propose a plug-in estimator for the asymptotic variance $\sigma^2$ and show that it is consistent in the Appendix (Lemma \ref{lemma:varest}). 
The asymptotic normality and the consistent variance estimator, provide the asymptotically valid confidence interval 
\begin{equation}
    CI_{1-\alpha} := \hat{\tau}_N(\pi,\eta) \pm z_{1-\alpha/2}\sqrt{\frac{\hat{\sigma}_N^2}{N}}, \label{formula:CI} 
\end{equation}
where $z_{1-\alpha/2}$ is the $(1-\alpha/2)$-quantile of a standard normal distribution.

We now provide two examples. The first illustrates how the growth of the maximal degree of the interference dependency graph depends on both the interaction network graph $I^N$ and the features $\boldsymbol{X}_i$. The second illustrates how we can use the generic graph to find valid adjustment sets.

\begin{example}
Let $S_e$ be an explicit SEM satisfying Assumption \ref{def:sem}, with features $\boldsymbol{X}^1$ as per equation \eqref{formula:feat1} and $\boldsymbol{X}^2$ as per equation \eqref{formula:feat3} that depend on some given interaction network graph $I^N$. Here, the interference dependency graph (see Definition \ref{def:depgraph}) contains an edge between any two units $i\neq j$ in the following three cases: (i) $I^N_{ij} = 1$, (ii) $I^N_{ji} = 1$ and (iii) $P^N_{ij} \geq 1$ or $P^N_{ji} \geq 1$, where $P^N=(I^N)^T I^N$ is the Gram matrix of $I^N$.
The maximal degree $d_{\text{max}}(N)$ of the interference dependency graph $D(\boldsymbol{\bar{X}}, \boldsymbol{\bar{W}})$ is therefore given by \begin{equation*}
    d_{\text{max}}(N) =  \max_{i \in  \{1,\ldots N\}} \sum_{j\in \{1,\ldots N\}\setminus i}  \mathds{1} \left\{(I^N + P^N)_{ij}\geq 1\right\},
\end{equation*}
where $ \mathds{1} \{\cdot\}$ denotes the indicator function.
Whether $d_{\text{max}}(N)$ satisfies $d_{\text{max}}(N) \in o(N)$ or $d_{\text{max}}(N) \in o(N^{1/4})$ depends on the specific sequence of $I^N$. It will, for example, hold if the $I^N$ have bounded maximal degree.  

Suppose an interference feature is non-local in $I^N$, for example $\boldsymbol{\bar{X}}= \boldsymbol{T} \boldsymbol{\bar{W}}$, where $\boldsymbol{T}=(T_{ij})$ with $T_{ij} = \left||\mathcal{N}_i^{1}| - |\mathcal{N}_j^{1}| \right|$ the difference in in-degree centrality between nodes $i$ and $j$ in $I^N$. Then $d_{\text{max}}(N) \in o(N)$ will only hold for very specific sequences of $I^N$ such as the sequence of empty graphs.
\label{example: maximal degree local}
\end{example}



\begin{example}  \label{example:ols}
Consider the explicit SEM $S_e$ from Example \ref{example:explicitDAG} and the corresponding generic graph $\mathcal{G}$ given in Figure \ref{subfig: generic features}. By the adjustment criterion (see Appendix \ref{appendix:prelims}), the valid adjustment sets relative to $( \{\boldsymbol{X},W,O\}, Y)$ in the generic graph $\mathcal{G}$ are $\{C_1\}$, $\{C_2\}$, $\{C_1, C_2\}$, $\{C_1, C_3\}$, $\{C_2, C_3\}$, and $\{C_1, C_2, C_3\}$. Based on research from the i.i.d.~setting \citep{rotnitzky,leoefficiency} it is likely that using $\{C_1\}$ results in a smaller asymptotic variance estimator than using the alternative adjustment sets. 
\end{example}

\section{Empirical Validation} \label{sec:empirical}

In a simulation study we validate the performance and theoretical properties of our adjustment estimator and compare it to alternative estimators that do not control either for within-unit confounding and/or interference. In addition, we apply our adjustment estimator to a real data example, where we estimate the effect of a strict facial-mask policy on the spread of COVID-19 in the early phase of the pandemic in Switzerland.

\subsection{Simulation Study} \label{sec:simul}


We consider three different structures for the interaction network graphs $I^N$: First, \textit{Erd{\H{o}}s--R{\'e}nyi networks}~\citep{erdos} $I(N,p_N)$, where for each pair of units $i\neq j \in \{1,\ldots, N\}$, we either draw both edges $\{i \rightarrow j, i \leftarrow j\}$ or neither of them with probability $p_N$. Second, \textit{family networks} of disjoint families, where within a family all members are pairwise connected and the family sizes are randomly sampled between $1$ and $6$. Third, \textit{directed square $2$-dimensional lattices} with at most one edge between two units. 

\begin{table}[t!]
\centering
 \begin{tabular}{c|c|c|c } 
 & Erd{\H{o}}s--R{\'e}nyi & Family network & 2d-lattice \\ 
 \hline
 Features & $(X^1)$ & $(X^1,X^2)$ & $(X^1,X^2)$\\
 $\boldsymbol{\alpha}_0$ & $(2,1)^T$ & $(2,1)^T$ & $(2,1,0.5)^T$ \\
 $\boldsymbol{\alpha}_1$ & $(0.4,1.1)^T$  & $(0.4,1.1)^T$ & $(0.4,1.1,0.5)^T$ \\
 Target effect & $\tau_N(0.7, 0.2)$ & $\tau_N(1, 0)$ & $\tau_N(0.5, 0.1)$ \\
  Sample sizes 
 & \begin{tabular}{@{}c@{}} 300, 600, 1200, \\ 2400, 4800 \end{tabular}
 & \begin{tabular}{@{}c@{}} 300, 600, 1200, \\ 2400, 4800 \\ \end{tabular} 
 & \begin{tabular}{@{}c@{}} 289, 576, 1225, \\ 2401, 4761 \\ \end{tabular}
 \end{tabular}
 \caption{Parameters for different graph-types in the simulation study}
 \label{table: simulation settings}
\end{table}

Throughout we consider explicit SEMs of the form given in Example \ref{example:explicitDAG}, where the error terms $\epsilon_{C_{i1}}$, $\epsilon_{C_{i2}}$, $\epsilon_{C_{i3}}$, and $\epsilon_{Y_i}$  are mean zero Gaussian random variables with variance $1$, except $\epsilon_{Y_i}$ which is uniformly distributed. We assume that we do not observe $C_{i1}$ for some or all units $i$. For the Erd{\H{o}}s--R{\'e}nyi and family networks we choose $h(\boldsymbol{\bar{W}}_{-i},I^N) = X_i^1$ and for the $2$-d lattices we choose $h(\boldsymbol{\bar{W}}_{-i},I^N) = (X_i^1, X_{i}^2)$, where $X_i^{1}$ is the fraction of treated parents in $I^N$ as per equation \eqref{formula:feat1}, and $X_i^{3}$ is the fraction of treated parents of parents in $I^N$ as per equation \eqref{formula:feat3}. We summarize the features, $\boldsymbol{\alpha}$ vectors, sample sizes and target effects we consider in Table \ref{table: simulation settings}. 
For each graph-type and sample size we draw $\texttt{nrep.graph}=50$ interaction network graphs $I^N$. For each of the $\texttt{nrep.graph}$ network graphs $I^N$ we draw the data $\texttt{nrep.data}=100$ times according to the explicit SEM $S_e$ from Example \ref{example:explicitDAG}. We sample different $I^N$ to investigate how the interaction network graph affects the estimator's performance but emphasize that Theorems \ref{prop:consistency-ols} and \ref{prop:normality-ols} hold for a fixed sequence of $I^N$.

To estimate the target global treatment effects, we use the valid adjustment set $\{C_2\}$ determined graphically in the generic graph $\mathcal{G}$, shown in Figure \ref{subfig: generic features}. For comparison, we consider, in addition to the estimator we propose (called fully adjusted estimator in the following), three additional OLS-based estimators (called naive, confounding adjusted, and interference adjusted estimator) according to equation \eqref{alphahatols}, where we choose $\mathbf{M}_i$ as follows:
\begin{enumerate}
\itemindent=+180pt
    \item[Naive estimator:] 
    $ \mathbf{M}_i = (1,W_i)^T$,
    \item[Confounding adjusted estimator:] 
    $ \mathbf{M}_i = (1,W_i, C_{i2})^T$,
    \item[Interference adjusted estimator:] 
      $ \mathbf{M}_i = (1,W_i,\boldsymbol{X}_i^T,\boldsymbol{O}_{i}^T )^T$,
    \item[Fully adjusted estimator: ] 
    $ \mathbf{M}_i = (1,W_i,\boldsymbol{X}_i^T,\boldsymbol{O}_i^T, C_{i2})^T$.
\end{enumerate}
For each $I^N$ we use the four estimators to estimate the target effect across the $\texttt{nrep.data}$ data-sets and use the results to compute the root mean-squared-error (RMSE), the empirical bias and the logarithm of the empirical variance.

Before discussing the results, we first discuss the maximal degree $d_{\text{max}}(N)$ of the interference dependency graphs. For the family networks and the $2$-d lattices it is clear that the maximal degree of the interaction network graph $I^N$ does not increase with $N$, and therefore it naturally holds that $d_{\text{max}}(N) \in o(N^{1/4})$ for any sequence $I^N$ (see Example \ref{example: maximal degree local}). Thus Theorems \ref{prop:consistency-ols} and  \ref{prop:normality-ols} hold in these cases. For the Erd{\H{o}}s--R{\'e}nyi networks $I(N, p_N)$ we performed a simulation to observe how $d_{\text{max}}(N)$ grows with $N$ for three different regimes: $p_N=N/10$, $p_N= N^{-2/3}$ and $p_N= 0.2$. Specifically, we drew $100$ interaction network graphs for each $N$ and computed the average maximal degree of the corresponding interference dependency graphs. A plot of the logarithm of the average maximal degree, $\bar{d}_{\text{max}}(N)$, against the logarithm of $N$ is shown in Figure \ref{fig:slopes}. For $I(N,10/N)$ the slope is $0.17 < 0.25$, that is, $d_{\text{max}}(N) $ empirically satisfies $d_{\text{max}}(N) \in o(N^{1/4})$. Based on this, we expect Theorems \ref{prop:consistency-ols} and  \ref{prop:normality-ols} to hold. For $I(N,N^{-2/3})$ the slope is $0.64 > 0.25$, that is, $d_{\text{max}}(N) $ empirically satisfies $d_{\text{max}}(N) \in o(N)$ but not $d_{\text{max}}(N) \in o(N^{1/4})$. Based on this, we expect that Theorem \ref{prop:consistency-ols} holds. For $I(N,0.2)$ the slope is $1$, that is, $d_{\text{max}}(N) $ does not satisfy $d_{\text{max}}(N) \in o(N)$.  Therefore, neither Theorem \ref{prop:consistency-ols} nor \ref{prop:normality-ols} can be applied in this case.

\begin{figure}[t]

\centering
  \includegraphics[width=.6\linewidth]{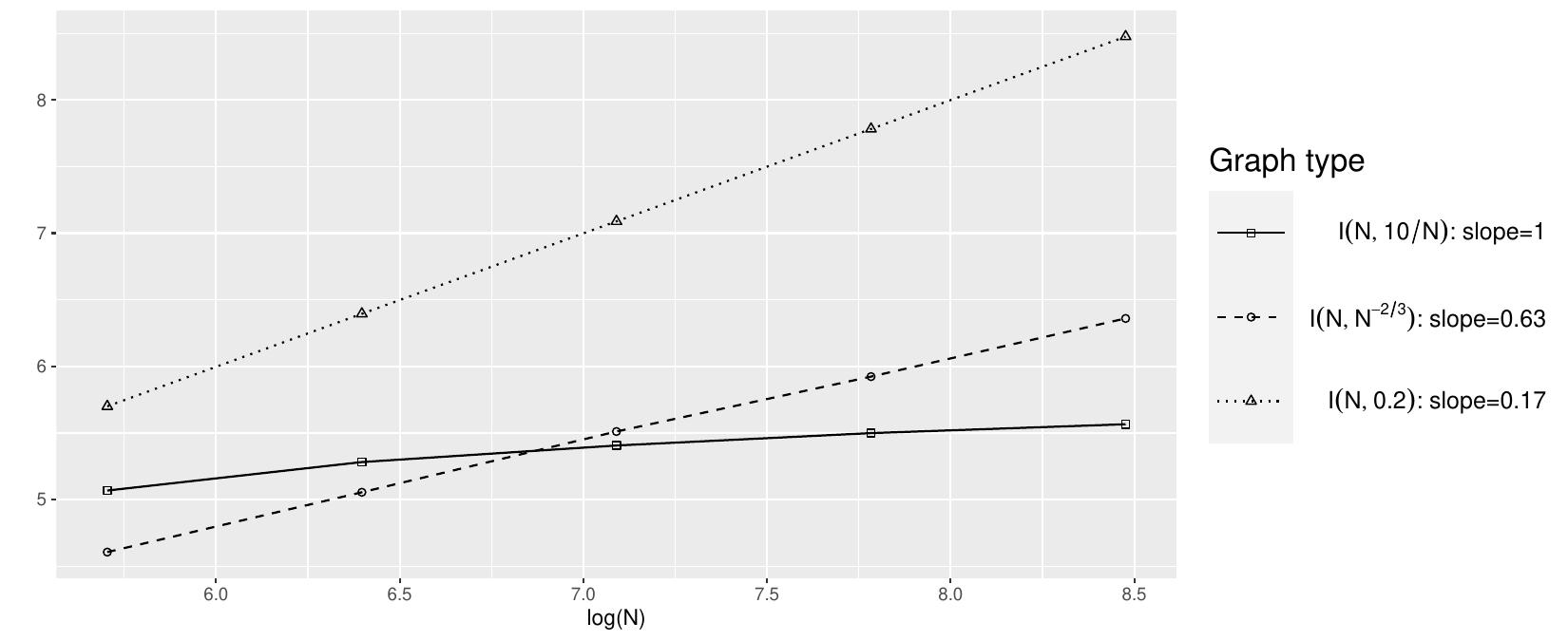}
  \caption{The logarithm of the average maximal degree of the dependency graph when considering the interference feature $X_i^{1}$ plotted against $\log(N)$.}
  \label{fig:slopes}
\end{figure}

We present the results with three plots, showing (i) the average root mean-squared-error (RMSE), (ii) the average empirical bias and (iii) the average logarithm of the empirical variance of $\hat{\tau}_N(\pi,\eta)$ against the logarithm of $N$ for these four estimators, with the average taken over the $\texttt{nrep.graph}=50$ network graphs $I^N$.
We assess the asymptotic normality and the consistency of the variance estimator (Lemma \ref{lemma:varest}) in Appendix \ref{app:normality}.

\begin{figure}

\centering
\begin{subfigure}[b]{0.9\textwidth}
\includegraphics[width=0.97\linewidth]{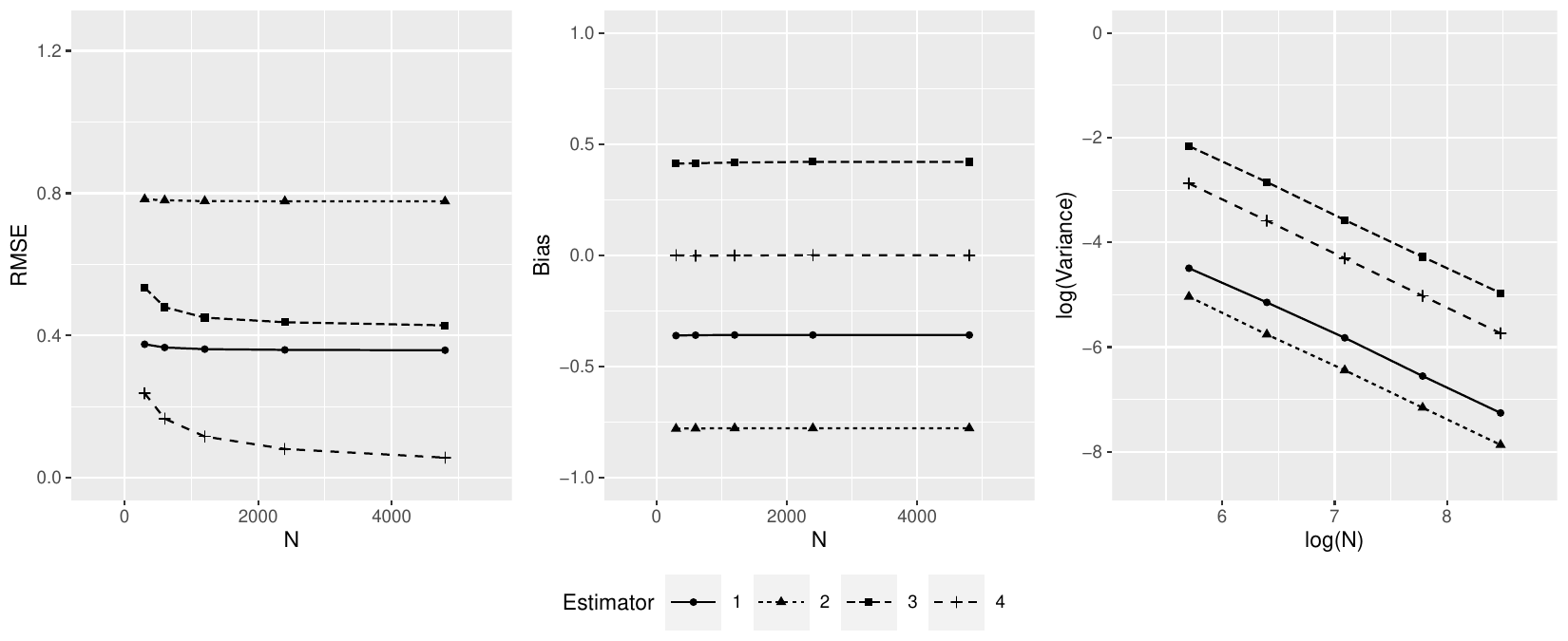}
\caption{}
\label{res:const10}
\end{subfigure}
\vfill
\begin{subfigure}[b]{0.9\textwidth}
\includegraphics[width=0.97\linewidth]{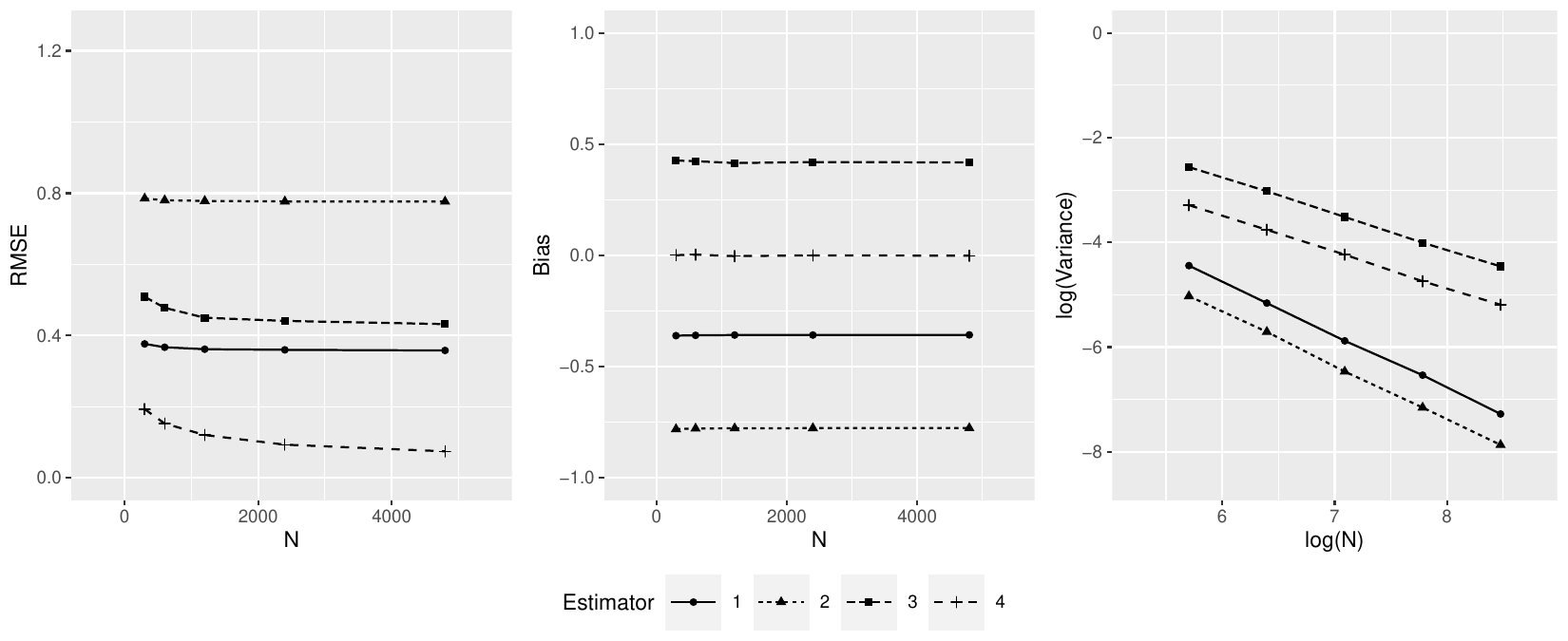}
\caption{}
\label{res:random23}
\end{subfigure}
\vfill
\begin{subfigure}[b]{0.9\textwidth}
\includegraphics[width=0.97\linewidth]{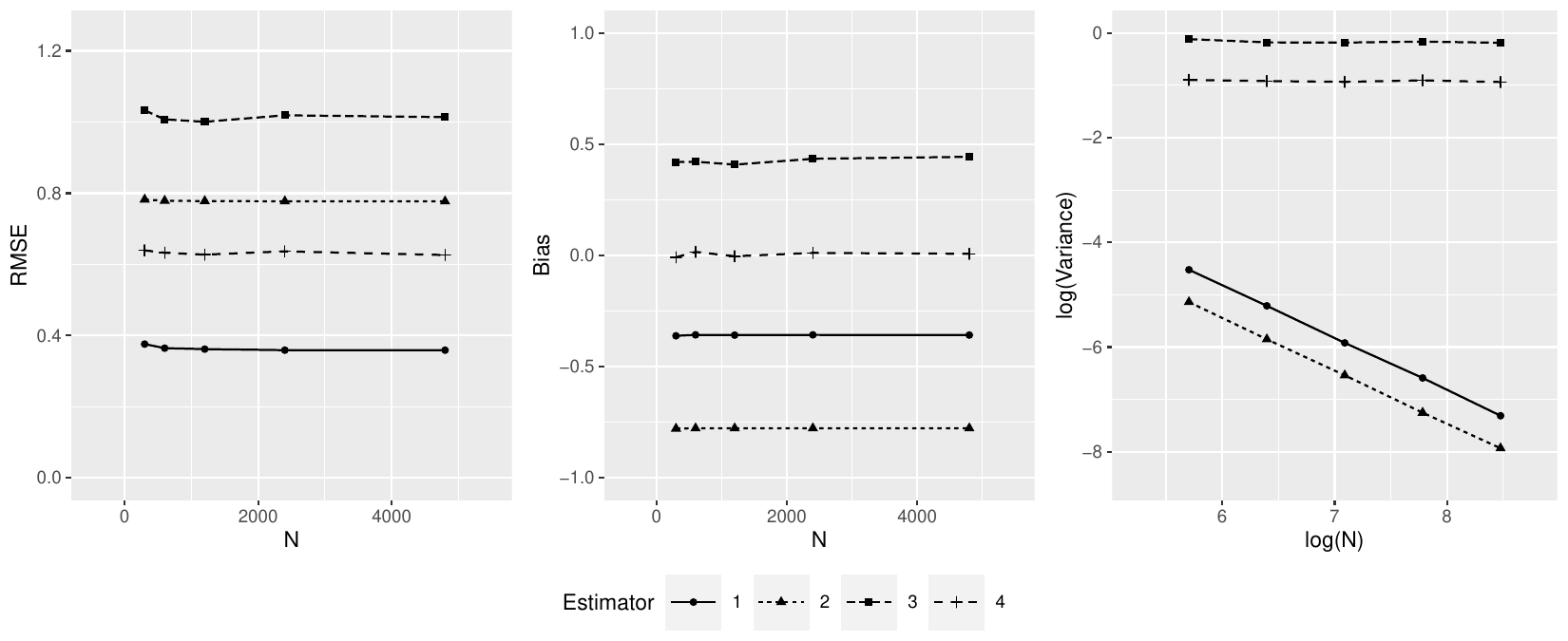}
\caption{}
\label{res:random02}
\end{subfigure}
\caption{RMSE, bias and log variance plots for the estimation of $\tau_N(0.7, 0.2)$ in \subref{res:const10} Erd{\H{o}}s--R{\'e}nyi networks $I(N,10/N)$, \subref{res:random23} Erd{\H{o}}s--R{\'e}nyi networks $I(N,N^{-2/3})$ and \subref{res:random02} Erd{\H{o}}s--R{\'e}nyi networks $I(N,0.2)$ using the naive (1), confounding adjusted (2), interference adjusted (3) and fully adjusted estimator (4), respectively.} 
\label{figure: simulation study}
\end{figure}

The results for the Erd{\H{o}}s--R{\'e}nyi networks are shown in Figure \ref{figure: simulation study}. 
The results for the family networks and for the $2$-d lattices are shown and discussed in Appendix \ref{app:furtherres}. 
The empirical bias plots show that the naive and the confounding adjusted estimator underestimate $\tau_N(\pi,\eta)$, while the interference adjusted estimator overestimates $\tau_N(\pi,\eta)$. In contrast the fully adjusted estimator appears to be close to unbiased even for small $N$. The variance plots also corroborate our results: for $I(N, 10/N)$, the only case where we expect Theorem \ref{prop:normality-ols} to hold, the variance of the fully adjusted estimator converges to zero with rate $N^{-1/2}$. We also verified that the fully adjusted estimator converges, when properly scaled, to a normal distribution (see Appendix \ref{app:normality}). For $I(N,N^{-2/3})$ we observe that, while the fully adjusted estimator still seems consistent, the convergence rate is slower than $N^{-1/2}$. For $I(N,0.2)$, the variance for the fully adjusted estimator does not appear to converge to zero, indicating inconsistency.

\subsection{Strict Facial-Mask Policy Data Analysis}

We now apply our estimator to study the effect of introducing a strict facial-mask policy on the spread of COVID-19 in Switzerland between July 2020 and December 2020. During several weeks in this early phase of the pandemic, the cantons of Switzerland could choose to adopt the government-determined facial-mask policy (mandatory facial-mask wearing on public transport) or a strict facial-mask policy (mandatory facial-mask wearing on public transport and in all public or shared spaces where social distancing was not possible).

This data set was gathered and analysed by \cite{nussli} and we closely follow their approach, including the causal assumptions. The key difference is that they estimate the causal effect of the strict facial-mask policy on the spread of COVID-19, without considering interference between neighboring cantons. Since people commute between neighboring cantons, the facial-mask policy of neighboring cantons might have had an effect on the spread of COVID-19 in a given canton. Here, we estimate the GATE $\tau_N(1,0)$, contrasting the hypothetical intervention of introducing the strict facial-mask policy nationally as compared to not introducing it in any canton.

We assume the following explicit SEM satisfying Assumption \ref{def:sem}, 
\begin{align}
     &\boldsymbol{C}_{i,t}  \leftarrow g_{\boldsymbol{C}}(\boldsymbol{C}_{i,t},\boldsymbol{\epsilon}_{C_{i,t}}), \ W_{i,t} \leftarrow g_W(\boldsymbol{C}_{i,t},\epsilon_{W_{i,t}}),\notag\\
    &X_{i,t} \leftarrow \frac{1}{\left|\mathcal{N}_i^{1}\right|}\sum_{j \in \mathcal{N}_i^{1}}W_{j,t}, \ O_{i,t} \leftarrow W_{i,t} X_{i,t} \text{ and}\notag\\
     &  Y_{i,t} =(1,X_{i,t}) \boldsymbol{\alpha}_{0} + ( W_{i,t}, O_{i,t}) \boldsymbol{\alpha}_{1}
+\boldsymbol{C}_{i,t}^T \boldsymbol{\gamma} +\epsilon_{Y_{i,t}},\label{eq:real}
 \end{align}
for each canton $i=1,\ldots,N=26$ and week $t=1,\ldots, T=24$. Here, a unit is given by a tuple $(i,t)$.
We assume that $(\boldsymbol{\epsilon}_{C_{i,t}},\epsilon_{W_{i,t}},\epsilon_{Y_{i,t}})$ are jointly independent error terms with expectation zero, and that their distributions do not depend on $i$ or $t$. Here, $\mathcal{N}_i^{1}$
denotes the neighbors of canton $i$ in $I^N \in \R^{N \times N}$, where $I^N$ is the geographical adjacency matrix.

\begin{figure}[t] 
\begin{subfigure}{0.5\textwidth}
\centering
\begin{tikzpicture}
 [>=stealth,
   shorten >=1pt,
   node distance=2cm,
   on grid,
   auto,
   scale=0.6, 
   transform shape,
   align=center,
   minimum size=3em
    ]
    \node[state,fill=lightgray] (m) at (0,0) {$W$};
    \node[state,fill=lightgray] (y) [right=6cm of m] {$Y$};
    \node[state] (ylag)[above =3cm of m] {$J$};
    \node[state] (meteo)[above =3cm of y ]{$\boldsymbol{M}$};
      \node[state] (p) [below= 3cm of m] {$\boldsymbol{P}$};
    \node[state] (w)[right= 3cm of ylag]{$\mathbf{D}$}; \node[state, dashed] (g)[left= 2.5cm of p]{$\boldsymbol{E}$};

    \node[state] (h)[below =3cm of y] {$H$};
        \node[state] (o)[below left=9cm and 1.5cm of w,diamond] {$O$};
    \node[state] (x)[right =3cm of o, diamond] {$X$};

\path[->,draw]
    (w) edge (m)
   (w) edge (p)
    (w) edge (ylag)
  (w) edge (y)
  (h) edge (y)
     (g) edge  (m)
     (g) edge (p)
     (ylag) edge (m) 
     (ylag) edge[bend right = 20] (p)
     (ylag) edge (y)
     (meteo) edge (y)
     (p) edge (y)
    (h) edge (m)
    (m) edge (o)
    (x) edge (o)
    (h) edge(p);
    \draw[->,line width=1.5pt] 
    (x) edge  (y)
    (m) edge (y)
    (o) edge  (y);
\end{tikzpicture}
\caption{}
\label{subfig: real data generic DAG}
\end{subfigure}
\begin{subfigure}{0.5\textwidth}
\centering
  \includegraphics[scale=0.5]{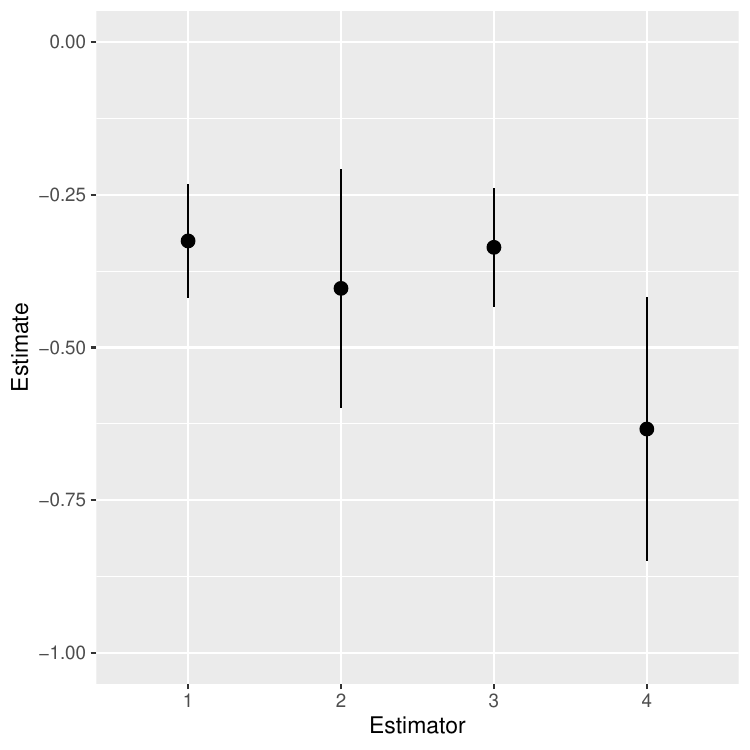}
  \caption{}
  \label{subfig: real data results}
\end{subfigure}
\caption{Assumed generic graph for the strict facial-mask policy analysis \subref{subfig: real data generic DAG} and estimates of $\tau_N(1,0)$ using the naive (1), confounding adjusted (2), interference adjusted (3) and fully adjusted estimator (4), respectively, with the corresponding $95\%$-confidence intervals \subref{subfig: real data results}.}
\label{fig:dagface}
\end{figure}
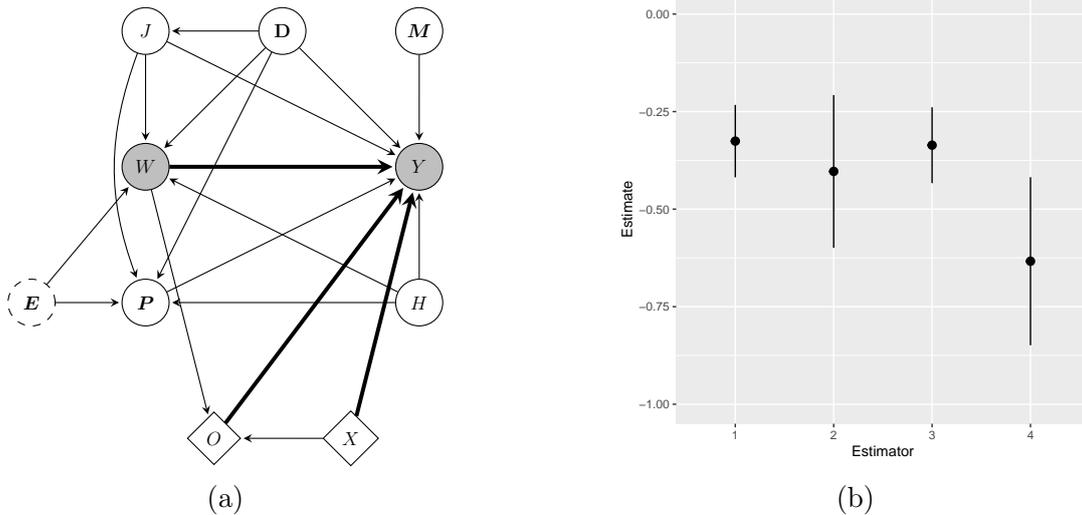

We now describe the response variable, the treatment variable and the covariates we consider.
\begin{itemize}
    \item[$Y_{i,t}$:] To specify the response variables,  let 
$G_{i,t} = \ln\left(A_{i,t}/A_{i,t-1}\right),$
where $A_{i,t}$ is the number of reported new cases in canton~$i$ in week~$t$. Due to the delay between the time of infection and the reporting of a new case, $G_{i,t}$ reflects the pandemic situation of a time period before $t$. Therefore, as response variable we use a future value of $G_{i,t}$. Specifically, $Y_{i,t}=G_{i, t+2}$.
    \item[$W_{i,t}$:] Treatment variable, given by the strict facial-mask policy indicator, where $0$ denotes the baseline government-determined policy and $1$ the strict facial-mask policy.
        \item[$\boldsymbol{P}_{i,t}$:]
     Indicators reflecting policies on the closing of workplaces, restrictions on gatherings and cancellations of public events.
    \item[$\boldsymbol{E}_{i,t}$:] Unobserved factors that determine the policy variables $W_{i,t}$ and $\boldsymbol{P}_{i,t}$.
    \item[$\mathbf{D}_{i}$:] Canton-specific demographic variables, given by population size, people of age $\geq 80$ years in $\%$, and people per $\textrm{km}^2$.
    \item[$H_{i,t}$:] Holiday indicator, where $1$ denotes public school holiday.
    \item[$\boldsymbol{M}_{i,t}$:] Meteorological variables, given by sunshine in minutes per day, air temperature in $^\circ\textrm{C}$, and mean relative humidity in $\%$.
      \item[$J_{i,t}$:] Information about the pandemic available to the public in week $t$, given by the lagged response variable $J_{i,t}= Y_{i, t-2}$.
    \item[$X_{i,t}$] and $O_{i,t}$: Interference feature and its product with the treatment $W_{i,t}$.
\end{itemize}
We use weekly data to remove weekly patterns and refer to \cite{nussli} for more details on the variables and the origin of the data.

The assumed generic graph is shown in Figure \ref{subfig: real data generic DAG}.
In our analysis, we adjust for $\{\mathbf{D}, H, \boldsymbol{M}, \boldsymbol{P} ,J\}$ which according to the generic graph is a valid adjustment set. Note that we cannot adjust for $\boldsymbol{E}$ as it is unobserved. In addition to the fully adjusted estimator we again consider the naive, confounding adjusted and interference adjusted estimators described in Section \ref{sec:simul}.

The results in Figure \ref{subfig: real data results} show the point estimates $\hat{\tau}_N(1,0)$ with their $95\%$-confidence intervals, computed using equation \eqref{formula:CI}. All four estimates are significantly negative, indicating that introducing the strict facial-mask policy nationally would have reduced the spread of COVID-19. The fully adjusted estimator provides the smallest estimate, indicating the presence of interference and illustrating the importance of taking it into account. As is always the case with observational data, the results need to be treated with care, as we assume, among other things, to know a valid adjustment set. 

\section{Discussion}


There are two natural avenues to generalize the results of this paper. First, in Section \ref{sec:identification} we show that for an explicit SEM following Assumption \ref{def:sem} the generic graph is a causal DAG. It is possible to derive similar results under weaker assumptions. For example, we do not allow for within-unit paths between $W_i$ on $Y_i$ that are mediated by some $C_i$, that is $W_i \rightarrow C_i \rightarrow Y_i$, but  the results generalize to explicit DAGs with such paths. For valid adjustment we can, however, assume that no such path exists without loss of generality \citep{witte2020efficient}. Second, by the identifiability results from Section \ref{sec:identification}, adjustment is only one possible strategy to estimate $\tau_N(\pi,\eta)$. Possible alternatives include the front-door criterion and instrumental variables. We restrict ourselves to models satisfying Assumptions \ref{def:sem} as well as adjustment to keep the presentation concise and focused on the crucial insight that we can adapt causal graphical model tools from the i.i.d.~setting to network effects.

There are three important caveats to our results. First, we require that the interference features be known. In practice, this will generally not be the case. There is, however, novel research on learning the interference mechanism \citep{belloni2022neighborhood}. Second, we assume a linear outcome model. This is needed for the important decomposition result in Proposition \ref{lemma:reformulation}. It may be possible to generalize our results to more flexible outcome models, as long as they admit a decomposition similar to Proposition \ref{lemma:reformulation}. A natural candidate is the class of partially linear models. Third, the constraints on the maximal degree of the interference dependency graph in Theorems \ref{prop:consistency-ols} and \ref{prop:normality-ols} are hard to formally verify, even in relatively simple examples. 

 \bibliography{myrefs}

 \newpage

\appendix

\section{Graphical Preliminaries} \label{appendix:prelims}

We now give an overview of the graphical terminology used throughout the paper.\\

\noindent\textbf{Graphs and Paths:}
A graph $G=(\boldsymbol{V}, \boldsymbol{E})$ is a tuple consisting of node-set $\boldsymbol{V}$ and edge-set $\boldsymbol{E}$. Edges may be directed $(\rightarrow)$, bi-directed $(\leftrightarrow)$, or undirected $(-)$. Two edges are \textit{adjacent} if they have a common node. A \textit{path} is a sequence of adjacent edges without repetition of a node. A path may consist of just a single edge. We call the first and the final node on a path the \textit{endpoint nodes} and all remaining nodes on the path \textit{nonendpoint nodes}.
\\

\noindent\textbf{DAGs:} A path from node $A$ to node $B$, where all edges on the path point towards $B$, together with an edge $B \rightarrow A$ forms a directed cycle. A directed graph without directed cycles is called a \textit{directed acyclic graph (DAG)}. \\

\noindent\textbf{Proper and Causal Paths:}
Let $G = (\boldsymbol{V}, \boldsymbol{E})$ be a DAG. A path from a set of nodes $\boldsymbol{A}$ to a set of nodes $\boldsymbol{B}$ in $G$ is a path from a node $V \in \boldsymbol{A}$ to a node $V' \in \boldsymbol{B}$. A path from $\boldsymbol{A}$ to $\boldsymbol{B}$ is called a \textit{proper path} if only the first node is in $\boldsymbol{A}$. A path from node $A$ to node $B$ in $G$ is called a \textit{causal path} if all edges on the path point towards $B$. Otherwise, we call the path \textit{noncausal}. \\

\noindent\textbf{Parents and Descendants:}
Let $G$ be a DAG.
We define the \textit{parents} of node $B$ in $G$ as all the nodes $A$ such that the edge $A \rightarrow B$ exists in $G$ and denote them $\text{pa}(B,G)$. We define the \textit{descendants} of $A$ in $G$ as all the nodes $B$, such that there exists a causal path from $A$ to $B$ in $G$ and denote them by $\text{de}(A,G)$. We use the convention that $A \in \text{de}(A,G)$. For a set $\boldsymbol{A}$, let $\text{de}(\boldsymbol{A},G)=\bigcup_{A\in \boldsymbol{A}} \text{de}(A,G)$.\\

\noindent \textbf{Colliders:} A nonendpoint node $V$ on a path $p$ in a DAG $G$ is a \textit{collider} if $p$ contains a subpath of the form $U \rightarrow V \leftarrow W$.
Otherwise, $V$ is called a \textit{noncollider} on $p$.\\

\noindent \textbf{Blocking and d-Separation:} (Definition $1.2.3$ in \cite{pearl2009causality} and Section $2.1$ in \cite{richardson2003markov}) Let $\boldsymbol{A}$ be a set of nodes in a DAG $G$. A path $p$ is blocked by $\boldsymbol{A}$ if $i)$ $p$ contains a noncollider that is in $\boldsymbol{A}$, or $ii)$ $p$ contains a collider $B$ such that no descendant of $B$ is in $\boldsymbol{A}$. If $\boldsymbol{A}$, $\boldsymbol{B}$ and $\mathbf{Z}$ are three pairwise disjoint sets of nodes in $G$, then $\mathbf{Z}$ \textit{d-separates} $\boldsymbol{A}$ from $\boldsymbol{B}$ if $\mathbf{Z}$ blocks every path between $\boldsymbol{A}$ and $\boldsymbol{B}$ in $G$. We then write $\boldsymbol{A} \ \dsep \ \boldsymbol{B} \mid \mathbf{Z}$. Otherwise, we write $\boldsymbol{A} \ \notdsep \ \boldsymbol{B} \mid \mathbf{Z}$. \\

\noindent\textbf{(Recursive) Structural Equation Model (SEM):} \citep{pearl2009causality}
Let $G = (\boldsymbol{V}, \boldsymbol{E})$ be a DAG. The random vector $\boldsymbol{V}= (V_1, \ldots, V_k)^T$ is generated from a \textit{structural equation model} (SEM) compatible with $G$ if each $V_j, j \in \{1,\ldots, k\}$, is generated by a structural equation,
$$ V_j \leftarrow f_j(\boldsymbol{V}_{\text{pa}(V_j, G)}, \epsilon_j),$$
where $f_j$ are functions and $\epsilon_j$ are independent error terms with expectation $0$. Each structural equation is interpreted as the generating mechanism, denoted by the assignment operator $\leftarrow$. Each structural equation is assumed to be invariant
to possible changes in the other structural equations. A SEM is called \textit{recursive} if there exists an ordering such that $f_{j}(\cdot, \epsilon_{j})$ only depends on variables $V_s$ with $s<j$ for all $j=1,\ldots, k$.\\

\noindent\textbf{do-Intervention:} A \textit{do-intervention} do$(V_j = A_j)$ in a SEM is modeled by replacing the structural equation 
\begin{align*}
  V_j \leftarrow h_j(\boldsymbol{V}_{\text{pa}(V_j, G)}, \epsilon_j) \ \ \text{by} \ \ V_j \leftarrow A_j,
\end{align*}
where $A_j$ may be deterministic or random. \\

\noindent\textbf{Total Joint Effect:} \citep{nandy2017estimating} The total joint effect of a set of random variables $\boldsymbol{A} = (A_1,\ldots, A_k)$ on a random variable $B$ is given by $
   \boldsymbol{\theta}_{b\boldsymbol{a}} := (\theta_{b a_1},\ldots,\theta_{b a_k})^T,$ 
where 
$$\theta_{b a_i}:= \frac{\partial}{\partial a_i}\E[B \mid \text{do}(\boldsymbol{A}= \boldsymbol{a}) ], \ \text{for} \ i = 1,\ldots, k.$$ \\

\noindent\textbf{Causal and Forbidden Nodes:} \citep{ema}
Let $G$ be a DAG.
We define the \textit{causal nodes} with respect to $(\boldsymbol{A},\boldsymbol{B})$ in $G$ as all nodes on proper causal paths from $\boldsymbol{A}$ to $\boldsymbol{B}$ excluding $\boldsymbol{A}$ and denote them by $\text{cn}(\boldsymbol{A}, \boldsymbol{B}, G)$. We define the \textit{forbidden nodes} relative to $(\boldsymbol{A},\boldsymbol{B})$ in $G$ as the descendants of the causal nodes as well as $\boldsymbol{A}$ and denote them by $\text{forb}(\boldsymbol{A},\boldsymbol{B},G)$.\\

\noindent\textbf{Valid Adjustment Sets:} \citep{ema} Consider disjoint node sets $\boldsymbol{A},\{B\}$ and $\mathbf{Z}$ in a DAG $G=(\boldsymbol{V},\boldsymbol{E})$ such that $\boldsymbol{V}$ is generated from a SEM compatible with $G$. We refer to $\mathbf{Z}$ as a \textit{valid adjustment set} relative to $(\boldsymbol{A}, B)$ in
$G$ if 
\begin{itemize}
    \item[i)] $\mathbf{Z} \cap \mathrm{forb}(\boldsymbol{A},B,G)  = \emptyset$, and
    \item[ii)] $\mathbf{Z}$ blocks all proper noncausal paths from $\boldsymbol{A}$ to $B$.
\end{itemize}
\vspace{0.4cm}

\noindent\textbf{Latent Projection:} \citep{Verma1990EquivalenceAS,shpitser2014introduction} Let $G$ be a DAG with node set $\boldsymbol{A} \cup \boldsymbol{B}$ where $\boldsymbol{A} \cap \boldsymbol{B}= \emptyset.$ The \textit{latent projection} of $G$ over $\boldsymbol{B}$ is a graph denoted $G^{\boldsymbol{B}}$ with node set $\boldsymbol{A}$ and edge-set defined as follows: For distinct nodes $A_i, A_j \in \boldsymbol{A}$, 
\begin{itemize}
    \item[i)] $G^{\boldsymbol{B}}$ contains a directed edge $A_i \rightarrow A_j$ if $G$ contains a directed path $A_i \rightarrow \cdots \rightarrow A_j$ on which all nonendpoint nodes are in $\boldsymbol{B}$,
    \item[ii)] $G^{\boldsymbol{B}}$ contains a bi-directed edge $A_i \leftrightarrow A_j$ if $G$ contains a path of the form $A_i \leftarrow \cdots \rightarrow A_j$ on which all nonendpoint nodes are noncolliders and in $\boldsymbol{B}$.
\end{itemize}

\section{Proofs for Section \ref{sec:identification}}
\subsection{Proofs for Section \ref{sec:ident_nofeat}}

The following definition formalizes what the generic graph $\mathcal{G}$ can be interpreted causally means.

\begin{definition}[Truncated factorization preserving generic graph] \label{def:structrpreserv}
Consider an explicit DAG $G_e$ with a compatible explicit SEM $S_e$ on explicit variables $\boldsymbol{V}_i$, $i=1,\ldots, N$. We say that the generic graph $\mathcal{G}=(\boldsymbol{V},\boldsymbol{E})$ is truncated factorization preserving for $G_e$   if it holds for all $i=1,\ldots, N$ and for all $\mathbf{A} \subset \mathbf{V}$ that
\begin{equation*}
    f(\boldsymbol{v}_i \setminus \boldsymbol{a}_i  \mid \doop (\boldsymbol{A}_i = \mathbf{a}_i))=
    \begin{cases}
	\prod_{V \in \mathbf{V}_i \setminus \mathbf{A}_i}f(v \mid \mathrm{pa}(V,\mathcal{G})), &  \text{if }\mathbf{A}_i =\mathbf{a}_i, \\
	0, & \text{otherwise,}
	\end{cases}
\end{equation*}
where for any node $N_i \in \boldsymbol{V}_i$ we define $\mathrm{pa}(N_i,\mathcal{G})=\mathrm{pa}(N,\mathcal{G})$, that is, the parent set of the node $N$ in $\mathcal{G}$ corresponding to $N_i$ according to Definition \ref{def:genericDAG}.
\end{definition}

\begin{restatable}[]{prop}{lemmastrucpreservnofeat}
\label{lemma:strucpreserv}
Let $S_e$ be an explicit SEM satisfying Assumption \ref{def:sem0} and let $G_e$ be the corresponding explicit DAG. Then the generic graph $\mathcal{G}$ of $G_e$ is truncated factorization preserving. 
\end{restatable}

\begin{proof}
Let $\boldsymbol{A} \subset \boldsymbol{V}$, where $\boldsymbol{V}$ is the node-set of the generic graph $\mathcal{G}=(\boldsymbol{V}, \boldsymbol{E})$.
Note first that since the explicit SEM $S_e$ is compatible with the explicit DAG $G_e$, the truncated factorization formula \citep{robinsgformula} holds with respect to $G_e$, that is, 
\begin{equation}
    f(\boldsymbol{\bar{v}}  \setminus \boldsymbol{a}_i \mid \text{do}(\mathbf{A}_i =\mathbf{a}_i))=
    \begin{cases}
	\prod_{V \in \boldsymbol{\bar{V}} \setminus \mathbf{A}_i}f(v \mid \mathrm{pa}(V,G_e)), &  \text{if }\mathbf{A}_i =\mathbf{a}_i, \\
	0, & \text{otherwise,}
	\end{cases}
	\label{eq11}
\end{equation}
where $\boldsymbol{\bar{V}} = \bigcup_{i=1}^N\boldsymbol{V}_i$. Further, let
$\boldsymbol{\bar{V}}_{-i} = \boldsymbol{\bar{V}} \setminus \mathbf{V}_i$ and
$\boldsymbol{\bar{Y}}= \bigcup_{i=1}^N Y_i$.

We distinguish two cases. The first case is $Y_i \in \boldsymbol{A}_i$. 
In the case that $\mathbf{A}_i =\mathbf{a}_i$, integrating out all variables in $\boldsymbol{\bar{V}}_{-i}$ we obtain
\begin{align}
    &f(\boldsymbol{v}_i \setminus \boldsymbol{a}_i \mid \doop(\mathbf{A}_i =\mathbf{a}_i)) \notag =\int_{\boldsymbol{\bar{v}}_{-i}}\prod_{V \in \boldsymbol{\bar{V}} \setminus \mathbf{A}_i}f(v \mid \mathrm{pa}(V,G_e))d\boldsymbol{\bar{v}}_{-i} \notag\\
&= \prod_{V \in \mathbf{V}_i \setminus \mathbf{A}_i } f(v \mid \mathrm{pa}(V,G_e)) \int_{\boldsymbol{\bar{v}}_{-i}} \prod_{V \in \boldsymbol{\bar{V}}_{-i}} f(y \mid \mathrm{pa}(Y,G_e)) d\boldsymbol{\bar{v}}_{-i} \\
&= \prod_{V \in \mathbf{V}_i \setminus \mathbf{A}_i } f(v \mid \mathrm{pa}(V,G_e)) = \prod_{V \in \mathbf{V}_i \setminus \mathbf{A}_i } f(v \mid \mathrm{pa}(V,\mathcal{G})), \notag
\end{align}
since the parents of any node $V \in \mathbf{V}_i \setminus \boldsymbol{A}_i $ are in $\mathbf{V}_i \setminus \boldsymbol{A}_i  $, as $Y_i \in \boldsymbol{A}_i$ and $Y_i$ is the only variable with parents indexed by other units $j$. Thus $\mathrm{pa}(V,G_e) = \mathrm{pa}(V,\mathcal{G})$, where $\mathrm{pa}(V,\mathcal{G})$ is defined in Definition \ref{def:structrpreserv}, for all nodes in $\mathbf{V}_i \setminus \mathbf{A}_i$, $i=1,\dots,N$. This concludes the proof of the first case.

The second case is $Y_i \notin \boldsymbol{A}_i$.
In the case that $\mathbf{A}_i =\mathbf{a}_i$, integrating out all variables in $\boldsymbol{\bar{V}}_{-i}$ we obtain that
\begin{align}
    &f(\boldsymbol{v}_i \setminus \boldsymbol{a}_i  \mid \doop(\mathbf{A}_i =\mathbf{a}_i)) \notag \\
    &=\int_{\boldsymbol{\bar{v}}_{-i}}\prod_{V \in \boldsymbol{\bar{V}} \setminus \mathbf{A}_i}f(v \mid \mathrm{pa}(V,G_e))d\boldsymbol{\bar{v}}_{-i} \notag\\
&= \prod_{V \in \mathbf{V}_i \setminus (\mathbf{A}_i \cup \{Y_i\}) } f(v \mid \mathrm{pa}(V,G_e)) \int_{\boldsymbol{\bar{v}}_{-i}} \prod_{Y \in \boldsymbol{\bar{Y}}} f(y \mid \mathrm{pa}(Y,G_e)) \notag \\
& \quad \quad \prod_{V \in \boldsymbol{\bar{V}}_{-i} \setminus  \boldsymbol{\bar{Y}}} f(v \mid \mathrm{pa}(V,G_e))d\boldsymbol{\bar{v}}_{-i}, \label{prop1:eq1}
\end{align}
where we use that all parents of nodes in $\mathbf{V}_i \setminus (\boldsymbol{A}_i \cup\{Y_i\}) $ are themselves in $\mathbf{V}_i \setminus  (\boldsymbol{A}_i \cup\{Y_i\})$.
Furthermore, considering the integral in equation \eqref{prop1:eq1}, we get

\begin{align*}
&\int_{\boldsymbol{\bar{v}}_{-i}} \prod_{Y \in \boldsymbol{\bar{Y}}} f(y \mid \mathrm{pa}(Y,G_e)) \prod_{V \in \boldsymbol{\bar{V}}_{-i} \setminus  \boldsymbol{\bar{Y}}} f(v \mid \mathrm{pa}(V,G_e))d\boldsymbol{\bar{v}}_{-i}\\
&=\int_{\boldsymbol{\bar{v}}_{-i}}  f(y_i \mid \mathrm{pa}(Y_i,G_e)) \prod_{j \neq i} f(y_i \mid \mathrm{pa}(Y_i,G_e)) f(\boldsymbol{\bar{v}}_{-i} \setminus \boldsymbol{\bar{y}}_{-i}) d\boldsymbol{\bar{v}}_{-i}\\
     &=
     \int_{\boldsymbol{\bar{v}}_{-i}} f(y_i\mid w_i,\boldsymbol{c}_i,\boldsymbol{\bar{w}}_{-i}) \prod_{j \neq i} f(y_j\mid w_j,\boldsymbol{c}_j,\boldsymbol{\bar{w}}_{-j}) f(\boldsymbol{\bar{v}}_{-i} \setminus \boldsymbol{\bar{y}}_{-i}) d\boldsymbol{\bar{v}}_{-i}\\
     &=
     \int_{\boldsymbol{\bar{v}}_{-i}} f(y_i\mid w_i,\boldsymbol{c}_i,\boldsymbol{\bar{w}}_{-i}) \prod_{j \neq i} f(y_j\mid w_i,\boldsymbol{c}_i, \boldsymbol{c}_j, \boldsymbol{\bar{w}}_{-i}) f(\boldsymbol{\bar{v}}_{-i} \setminus \boldsymbol{\bar{y}}_{-i} \mid w_i,\boldsymbol{c}_i)d\boldsymbol{\bar{v}}_{-i}\\
     &=
     \int_{\boldsymbol{\bar{v}}_{-i}} f(y_i\mid w_i,\boldsymbol{c}_i,\boldsymbol{\bar{v}}_{-i}) f(\boldsymbol{\bar{v}}_{-i} \mid w_i,\boldsymbol{c}_i)d\boldsymbol{\bar{v}}_{-i} = f(y_i \mid w_i,\boldsymbol{c}_i)= f(y_i \mid \mathrm{pa}(y_i,\mathcal{G})),
\end{align*}
where in the first equality we used that $\boldsymbol{\bar{V}}_{-i} \setminus \boldsymbol{\bar{Y}}_{-i}$ is an ancestral set, and in the third equality that $Y_j \indep \boldsymbol{C}_i \mid \boldsymbol{C}_j, \boldsymbol{\bar{W}}$ and $\boldsymbol{\bar{V}}_{-i} \setminus \boldsymbol{\bar{Y}}_{-i} \indep W_i,\boldsymbol{C}_i$, which follow from Assumption \ref{def:sem0} and the local Markov property, that is, for all $V \in \boldsymbol{\bar{V}}$ it holds that $V \indep \boldsymbol{\bar{V}} \setminus \{\text{de}(V, G_e) \cup \text{pa}(V, G_e) \} \mid  \text{pa}(V, G_e)$. Thus, combining the above we get 
\begin{align*}
    f(\boldsymbol{v}_i \setminus \boldsymbol{a}_i \mid \doop(\mathbf{A}_i 
    &=\mathbf{a}_i))= f(y_i \mid w_i,\boldsymbol{c}_i) \prod_{V \in \mathbf{V}_i \setminus (\mathbf{A}_i \cup \{Y_i\}) } f(v \mid \mathrm{pa}(V,G_e)) \\
    & =
	\prod_{V \in \mathbf{V}_i \setminus \mathbf{A}_i}f(v \mid \mathrm{pa}(V,\mathcal{G})), 
\end{align*}
since the parents of any node $V \in \mathbf{V}_i \setminus (\mathbf{A}_i \cup \{Y_i\})$ are in $\mathbf{V}_i$ and thus $\mathrm{pa}(V,G_e) = \mathrm{pa}(V,\mathcal{G})$, where $\mathrm{pa}(V,\mathcal{G})$ is defined in Definition \ref{def:structrpreserv}.
\end{proof}

\subsection{Proofs for Section \ref{sec:ident_feat}} \label{app:proofssecidentwithfeat}

\begin{restatable}[Invariance of $\tau_N(\pi,\eta)$ to linear transformations of features]{lemma}{lemmainvariance}
\label{lemma:invariance}
Consider an explicit SEM $S_e$ satisfying Assumption \ref{def:sem} with features $\boldsymbol{X}_i =(X_{i1}, X_{i2}, \ldots, X_{iP})^T$. Let $\tau_N(\pi,\eta)$ be the treatment effect obtained by using $\boldsymbol{X}_i$ as features and $ \tilde{\tau}_N(P_{\pi}, P_{ \eta})$ the treatment effect obtained by replacing $\boldsymbol{X}_i$ with $\boldsymbol{\tilde{X}}_i = (l_1(X_{i1}), l_2(X_{i2}), \ldots, l_P(X_{iP}))^T$, where $l_k(x) := a_kx +b_k$, for $k=1,\ldots, P$, with $a_k, b_k \in \mathbb{R} $. It then holds that
\begin{equation*}
 \tau_N(\pi,\eta)=  \tilde{\tau}_N(P_{\pi}, P_{ \eta}).
\end{equation*}
\end{restatable}

\begin{proof}
Let unit $i$ be fixed and let us focus on the outcome equation in \eqref{eq:outcomewithbeta}. Let us look at the case $W_i = 1$. We reformulate the generating equation of the outcome $Y_i$ in $S_e$ as
\begin{align*}
Y_i &= (1,\boldsymbol{X}_i^T)\boldsymbol{\beta}_{1} + \boldsymbol{C}_i^T\boldsymbol{\gamma} + \epsilon_{Y_i} \\
&=\beta^{1}_0  + \beta^{1}_1 X_{i1}  + \ldots + \beta^{1}_P X_{iP} + \boldsymbol{C}_i^T\boldsymbol{\gamma} + \epsilon_{Y_i} \\
&= \beta^{1}_0 - \frac{\beta^{1}_1}{a_1} b_1 - \ldots -\frac{\beta^{1}_P}{a_P} b_P + \frac{\beta^{1}_1}{a_1} (a_1X_{i1} +b_1) + \ldots + \frac{\beta^{1}_P}{a_P} (a_PX_{iP} +b_P) + \boldsymbol{C}_i^T\boldsymbol{\gamma}  + \epsilon_{Y_i} \\
&=   \tilde{\beta}^1_0 +  \tilde{\beta}^1_1 l_1(X_{i1}) + \ldots + \tilde{\beta}^1_P l_P(X_{iP})   +\boldsymbol{C}_i^T \boldsymbol{\gamma} + \epsilon_{Y_i},
\end{align*}
where
\begin{align}
    \tilde{\beta}^1_0 &:= \beta^{1}_0 - \frac{\beta^{1}_1}{a_1} b_1 - \ldots -\frac{\beta^{1}_P}{a_P} b_P, \\
    \tilde{\beta}^1_j &: =  \frac{\beta^{1}_j}{a_j} \text{ for  } j = 1,\ldots, P
\end{align}
and $\boldsymbol{\beta}_1 = (\beta^1_0,\dots,\beta^1_P)$.

In the following, we write $\boldsymbol{\omega}_{1}$ instead of $\boldsymbol{\omega}^N_{1}(\pi,\eta)$ and $\boldsymbol{\omega}_{0}$ instead of $\boldsymbol{\omega}^N_{0}(\pi,\eta)$ to ease notation. We also write $\tilde{\boldsymbol{\omega}}_{1}$ instead of  $\tilde{\boldsymbol{\omega}}^N_{1}(\pi, \eta)$ and $\tilde{\boldsymbol{\omega}}_{0}$ instead of  $\tilde{\boldsymbol{\omega}}^N_{0}(\pi, \eta)$, where $\tilde{\boldsymbol{\omega}}_{1} = (1,l_1(\omega^{1}_1), \ldots,l_P(\omega^{1}_P))$ are the weights obtained if the linearly transformed features $l_k(X_{ik})$ are used to compute the weights per the equations in Proposition \ref{lemma:reformulation}. It then follows that 
   $\boldsymbol{\omega}^{T}_1\boldsymbol{\beta}_{1} = \tilde{\boldsymbol{\omega}}^{T}_{1}\tilde{\boldsymbol{\beta}}_{1}$,
where $\tilde{\boldsymbol{\beta}}_{1} = (\tilde{\beta}^{1}_0 ,\ldots ,\tilde{\beta}^{1}_P) $, and by the same arguments 
$\boldsymbol{\omega}^{T}_{0}\boldsymbol{\beta}_{0} = 
\tilde{\boldsymbol{\omega}}^{T}_{0}\tilde{\boldsymbol{\beta}}_{0}$, 
 which proves the result.
 \end{proof}

\begin{restatable}[]{prop}{lemmastrucpreservwithfeat}
\label{lemma:strucpreservwithfeat}
Let $S_e$ be an explicit SEM satisfying Assumption \ref{def:sem} and let $G_e$ be the corresponding explicit DAG. Then the generic graph $\mathcal{G}$ of $G_e$ is truncated factorization preserving, if there is one multivariate node for the features $\boldsymbol{X}_i$ in $G_e$.
\end{restatable}

\begin{proof}
Let $\boldsymbol{A} \subset \boldsymbol{V}$, where $\boldsymbol{V}$ is the node-set of the generic graph $\mathcal{G}=(\boldsymbol{V}, \boldsymbol{E})$.
Note first that since the explicit SEM $S_e$ is compatible with the explicit DAG $G_e$, the truncated factorization formula holds with respect to $G_e$, that is, 
\begin{equation}
    f(\boldsymbol{\bar{v}}  \setminus \boldsymbol{a}_i \mid \text{do}(\mathbf{A}_i =\mathbf{a}_i))=
    \begin{cases}
	\prod_{V \in \boldsymbol{\bar{V}} \setminus \mathbf{A}_i}f(v \mid \mathrm{pa}(V,G_e)), &  \text{if }\mathbf{A}_i =\mathbf{a}_i, \\
	0, & \text{otherwise,}
	\end{cases}
\end{equation}
where $\boldsymbol{\bar{V}} = \bigcup_{i=1}^N \boldsymbol{V}_i$. Further, let
$\boldsymbol{\bar{V}}_{-i} = \boldsymbol{\bar{V}} \setminus \mathbf{V}_i$. 

We distinguish two cases. The first case is $\boldsymbol{X}_i \subseteq \boldsymbol{A}_i$. In the case that $\mathbf{A}_i =\mathbf{a}_i$, integrating out all variables in $\boldsymbol{\bar{V}}_{-i}$ we obtain that
\begin{align}
    &f(\boldsymbol{v}_i \setminus \boldsymbol{a}_i  \mid \doop(\mathbf{A}_i =\mathbf{a}_i)) \notag =\int_{\boldsymbol{\bar{v}}_{-i}}\prod_{V \in \boldsymbol{\bar{V}} \setminus \mathbf{A}_i}f(v \mid \mathrm{pa}(V,G_e))d\boldsymbol{\bar{v}}_{-i} \notag\\
&= \prod_{V \in \boldsymbol{V}_i \setminus \boldsymbol{A}_i} f(v \mid \mathrm{pa}(V,G_e)) \int_{\boldsymbol{\bar{v}}_{-i}} \prod_{V \in \boldsymbol{\bar{V}}_{-i}} f(v\mid \mathrm{pa}(V,G_e)) d\boldsymbol{\bar{v}}_{-i}, \label{prop:eq1}
\end{align}
since the parents of any node $ V \in \boldsymbol{V}_i \setminus \boldsymbol{A}_i$ are in $ \boldsymbol{V}_i \setminus \boldsymbol{A}_i$, since $\boldsymbol{X}_i \subseteq \boldsymbol{A}_i$ and $\boldsymbol{X}_i$ are the only variables with parents indexed by other units $j$.
In the following, we show that the integral in equation \eqref{prop:eq1} equals $1$. Consider the product of densities in equation \eqref{prop:eq1},
\begin{align}
&\prod_{V \in \boldsymbol{\bar{V}}_{-i}} f(v\mid \mathrm{pa}(V,G_e))  \notag \\ 
&= \prod_{j \neq i} f(\boldsymbol{c}_j) f(w_j \mid \boldsymbol{c}_j) f(\boldsymbol{o}_j \mid w_j, \boldsymbol{x}_j) f(\boldsymbol{x}_j \mid \boldsymbol{\bar{w}}_{-j}) f( y_j \mid w_j, \boldsymbol{c}_j, \boldsymbol{o}_j, \boldsymbol{x}_j) \notag \\ 
&= \prod_{j \neq i}  f(w_j,\boldsymbol{c}_j) f(\boldsymbol{o}_j \mid w_j, \boldsymbol{x}_j, \boldsymbol{\bar{w}}_{-j}) f(\boldsymbol{x}_j  \mid w_j, \boldsymbol{\bar{w}}_{-j}) f( y_j \mid w_j, \boldsymbol{c}_j, \boldsymbol{o}_j, \boldsymbol{x}_j)  \notag\\ 
&= \prod_{j \neq i}  f(w_j,\boldsymbol{c}_j) f(\boldsymbol{o}_j,\boldsymbol{x}_j  \mid w_j,  \boldsymbol{\bar{w}}_{-j}) f( y_j \mid w_j, \boldsymbol{c}_j, \boldsymbol{o}_j, \boldsymbol{x}_j)   \notag\\
&=\prod_{j \neq i}  f(w_j,\boldsymbol{c}_j \mid \boldsymbol{\bar{w}}_{-j}) f(\boldsymbol{o}_j,\boldsymbol{x}_j  \mid w_j,  \boldsymbol{\bar{w}}_{-j}, \boldsymbol{c}_j) f( y_j \mid w_j, \boldsymbol{c}_j, \boldsymbol{o}_j, \boldsymbol{x}_j , \boldsymbol{\bar{w}}_{-j}) \notag\\ 
&= \prod_{j \neq i}  f(w_j, \boldsymbol{c}_j,  \boldsymbol{o}_j,\boldsymbol{x}_j , y_j \mid \boldsymbol{\bar{w}}_{-j}) , \label{dens}
\end{align}
where in the second and fourth equality we used Assumption \ref{def:sem} and the local Markov property in $G_e$, that is, for all $V \in \boldsymbol{\bar{V}}$ it holds that $V \indep \boldsymbol{\bar{V}} \setminus \{\text{de}(V, G_e) \cup \text{pa}(V, G_e) \} \mid  \text{pa}(V, G_e)$. Especially, we used that $f(\boldsymbol{x}_j \mid \mathrm{pa}(\boldsymbol{X}_j,G_e)) = f(\boldsymbol{x}_j \mid \boldsymbol{\bar{w}}_{-j})$ by the local Markov property, even though it does not necessarily hold that $\mathrm{pa}(\boldsymbol{X}_j,G_e) = \boldsymbol{\bar{W}}_{-j}$, that is, not all $W_j$ for $j\neq i$ need to be in $\mathrm{pa}(\boldsymbol{X}_j,G_e)$.

We now consider the density in equation \eqref{dens} for a given $j\neq i$,
\begin{align*}
   f(w_j, \boldsymbol{c}_j,  \boldsymbol{o}_j,\boldsymbol{x}_j, y_j \mid \boldsymbol{\bar{w}}_{-j}) &=
   f(w_j \mid \boldsymbol{\bar{w}}_{-j}) f( \boldsymbol{c}_j,  \boldsymbol{o}_j,\boldsymbol{x}_j , y_j \mid w_j, \boldsymbol{\bar{w}}_{-j})\\
   &= f(w_j ) f( \boldsymbol{c}_j,  \boldsymbol{o}_j,\boldsymbol{x}_j, y_j \mid w_i, \boldsymbol{\bar{w}}_{-i})\\
    &= f(w_j) \frac{f( \boldsymbol{c}_j,  \boldsymbol{o}_j,\boldsymbol{x}_j, y_j  ,w_i\mid  \boldsymbol{\bar{w}}_{-i})}{f(w_i)},
\end{align*}
using that $W_i \indep W_j$ for $j\neq i$ by d-separation in $G_e$.
Using this reformulation of the density and considering the whole integral in equation \eqref{prop:eq1} leads to 
\begin{align*}
    &\int_{\boldsymbol{\bar{v}}_{-i}} \prod_{j \neq i}  f(w_j, \boldsymbol{c}_j,  \boldsymbol{o}_j,\boldsymbol{x}_j, y_j \mid \boldsymbol{\bar{w}}_{-j}) d\boldsymbol{\bar{v}}_{-i} \\
    &= \int_{\boldsymbol{\bar{v}}_{-i}} \prod_{j \neq i}  \frac{f(w_j)}{f(w_i)}f( \boldsymbol{c}_j,  \boldsymbol{o}_j,\boldsymbol{x}_j ,y_j  ,w_i\mid  \boldsymbol{\bar{w}}_{-i}) d\boldsymbol{\bar{v}}_{-i}.
\end{align*}
We now fix $j\neq i$. Using Fubini we integrate out all variables indexed by $j$ and obtain, 
\begin{align*}
& \int_{\boldsymbol{v}_j}\frac{f(w_j)}{f(w_i)} f( \boldsymbol{c}_j,  \boldsymbol{o}_j,\boldsymbol{x}_j, y_j  ,w_i\mid  \boldsymbol{\bar{w}}_{-i}) d\boldsymbol{v}_j\\
    &=\int_{w_j}\frac{f(w_j)}{f(w_i)} \left(\int_{\boldsymbol{c}_j} \int_{\boldsymbol{o}_j} \int_{\boldsymbol{x}_j}
    \int_{y_j}f( \boldsymbol{c}_j,  \boldsymbol{o}_j,\boldsymbol{x}_j, y_j  ,w_i\mid  \boldsymbol{\bar{w}}_{-i}) dy_j d\boldsymbol{x}_j d\boldsymbol{o}_j d\boldsymbol{c}_j\right) dw_j\\
    &=  \int_{w_j}\frac{f(w_j)}{f(w_i)} f(w_i \mid \boldsymbol{\bar{w}}_{-i}) dw_j = \int_{w_j}\frac{f(w_j)}{f(w_i)} f(w_i ) dw_j=  1,
\end{align*}
using in the third equality again that $W_i \indep W_j$ for $j\neq i$ by d-separation in $G_e$.
Thus, combining the above we get in the case that $\boldsymbol{A}_i = \boldsymbol{a}_i$ that
\begin{equation*}
    f(\boldsymbol{v}_i  \setminus \boldsymbol{a}_i  \mid \doop(\mathbf{A}_i = \mathbf{a}_i)) =\prod_{V \in \boldsymbol{V}_i \setminus \boldsymbol{A}_i} f(v \mid \mathrm{pa}(V,G_e))
    =
	\prod_{V \in \boldsymbol{V}_i \setminus \boldsymbol{A}_i} f(v \mid \mathrm{pa}(V,\mathcal{G})), 
\end{equation*}
since the parents of any node $V \in\ \boldsymbol{V}_i \setminus \boldsymbol{A}_i$ are in $\mathbf{V}_i$ and thus $\mathrm{pa}(V,G_e) = \mathrm{pa}(V,\mathcal{G})$, where $\mathrm{pa}(V,\mathcal{G})$ are defined in Definition \ref{def:structrpreserv}. This concludes the proof of the first case.

The second case is $\boldsymbol{X}_i \cap \boldsymbol{A}_i = \emptyset$.
In the case that $\mathbf{A}_i =\mathbf{a}_i$, integrating out all variables in $\boldsymbol{\bar{V}}_{-i}$ we obtain that
\begin{align}
    &f(\boldsymbol{v}_i \setminus \boldsymbol{a}_i  \mid \doop(\mathbf{A}_i =\mathbf{a}_i)) 
    =\int_{\boldsymbol{\bar{v}}_{-i}}\prod_{V \in \boldsymbol{\bar{V}} \setminus \boldsymbol{A}_i}f(v \mid \mathrm{pa}(V,G_e))d\boldsymbol{\bar{v}}_{-i} \notag\\
&= \prod_{V \in \boldsymbol{V}_i \setminus \{\mathbf{A}_i \cup \boldsymbol{X}_i\}} f(v \mid \mathrm{pa}(V,G_e)) \int_{\boldsymbol{\bar{v}}_{-i}} f(\boldsymbol{x}_i \mid \mathrm{pa}(\boldsymbol{X}_i,G_e)) \prod_{V \in \boldsymbol{\bar{V}}_{-i}} f(v\mid \mathrm{pa}(V,G_e)) d\boldsymbol{\bar{v}}_{-i}, \label{proppart2}
\end{align}
where we use again that the parents of any node $V \in \boldsymbol{V}_i \setminus \{\boldsymbol{A}_i \cup  \boldsymbol{X}_i\}$ are in $ \boldsymbol{V}_i \setminus \{\boldsymbol{A}_i \cup  \boldsymbol{X}_i\}$, since $\boldsymbol{X}_i \cap \boldsymbol{A}_i = \emptyset$ and $\boldsymbol{X}_i$ are the only variables with parents indexed by other units $j$.
We now consider the integral in equation \eqref{proppart2},
\begin{align*}
&\int_{\boldsymbol{\bar{v}}_{-i}} f(\boldsymbol{x}_i \mid \boldsymbol{\bar{w}}_{-i} ) \prod_{V \in \boldsymbol{\bar{V}}_{-i}} f(v\mid \mathrm{pa}(V,G_e)) d\boldsymbol{\bar{v}}_{-i} \notag \\
&=  \int_{\boldsymbol{\bar{v}}_{-i}} f(\boldsymbol{x}_i \mid \boldsymbol{\bar{w}}_{-i} )\prod_{j \neq i} f(w_j, \boldsymbol{c}_j,  \boldsymbol{o}_j,\boldsymbol{x}_j, y_j \mid \boldsymbol{\bar{w}}_{-j})  d\boldsymbol{\bar{v}}_{-i}\\
&=  \int_{\boldsymbol{\bar{v}}_{-i}} f(\boldsymbol{x}_i \mid \boldsymbol{\bar{w}}_{-i} )\prod_{j \neq i} f(w_j \mid \boldsymbol{\bar{w}}_{-j}) f(\boldsymbol{c}_j,  \boldsymbol{o}_j,\boldsymbol{x}_j, y_j \mid w_i, \boldsymbol{\bar{w}}_{-i}) d\boldsymbol{\bar{v}}_{-i} \\
&=  \int_{w_j, j \neq i} f(\boldsymbol{x}_i \mid \boldsymbol{\bar{w}}_{-i} ) \\
&\quad\quad \left(\prod_{j \neq i} f(w_j) \int_{\boldsymbol{c}_j} \int_{\boldsymbol{o}_j} \int_{\boldsymbol{x}_j} \int_{y_j}  f(\boldsymbol{c}_j,  \boldsymbol{o}_j,\boldsymbol{x}_j, y_j \mid w_i, \boldsymbol{\bar{w}}_{-i})  dy_j d\boldsymbol{x}_j  d\boldsymbol{o}_j  d\boldsymbol{c}_j \right) d \boldsymbol{\bar{w}}_{-i}\\
&=  \int_{w_j, j \neq i} f(\boldsymbol{x}_i \mid \boldsymbol{\bar{w}}_{-i} ) \left(\prod_{j \neq i} f(w_j) \right) d \boldsymbol{\bar{w}}_{-i}=  \int_{w_j, j \neq i} f(\boldsymbol{x}_i \mid \boldsymbol{\bar{w}}_{-i} )  f(\boldsymbol{\bar{w}}_{-i})  d \boldsymbol{\bar{w}}_{-i}\\
&=  \int_{w_j, j \neq i} f(\boldsymbol{x}_i ,\boldsymbol{\bar{w}}_{-i} )  d \boldsymbol{\bar{w}}_{-i}
= f(\boldsymbol{x}_i),
\end{align*} 
using in the second equality again equation \eqref{dens} and that $W_i \indep W_j$ for $j\neq i$ by d-separation in $G_e$.

Thus, combining the above we get in the case that $\boldsymbol{A}_i = \boldsymbol{a}_i$ that
\begin{equation*}
   f(\boldsymbol{v}_i \setminus \boldsymbol{a}_i \mid \doop(\mathbf{A}_i = \mathbf{a}_i)) =f(\boldsymbol{x}_i)\prod_{V \in \boldsymbol{V}_i \setminus \{\boldsymbol{A}_i \cup \boldsymbol{X}_i\}} f(v \mid \mathrm{pa}(V,G_e))
    =
	\prod_{V \in \boldsymbol{V}_i \setminus \boldsymbol{A}_i} f(v \mid \mathrm{pa}(V,\mathcal{G})), 
\end{equation*}
since the parents of any node $V \in  \boldsymbol{V}_i \setminus \{\boldsymbol{A}_i \cup  \boldsymbol{X}_i\}$ are in $ \boldsymbol{V}_i \setminus \{\boldsymbol{A}_i \cup  \boldsymbol{X}_i\}$, and thus $\mathrm{pa}(V,G_e) = \mathrm{pa}(V,\mathcal{G})$, where $\mathrm{pa}(V,\mathcal{G})$ are defined in Definition \ref{def:structrpreserv}. In addition, the parent set of $\boldsymbol{X}_i$ in $\mathcal{G}$ is the empty set. This concludes the proof of the second case.
\end{proof}

\lemmataureformulation*
\begin{proof} 
Let us consider first the term $\E[Y_i \mid \text{do}(\boldsymbol{\bar{W}} \overset{\text{i.i.d.}}{\sim} \text{Bern}(\pi))]$ for a fixed unit $i$. Plugging in the outcome equation \eqref{eq:outcomewithbeta}, we obtain
\begin{align*}
    &\E[Y_i \mid \text{do}(\boldsymbol{\bar{W}} \overset{\text{i.i.d.}}{\sim} \text{Bern}(\pi))] \\
    &\quad= (1-\E[W_i \mid \text{do}(\boldsymbol{\bar{W}} \overset{\text{i.i.d.}}{\sim} \text{Bern}(\pi))])\E[(1,\boldsymbol{X}_i^T) \mid \doop(\boldsymbol{\bar{W}}_{-i} \overset{\text{i.i.d.}}{\sim} \text{Bern}(\pi))] \boldsymbol{\beta}_{0}  \\
    &\quad\quad+\E[W_i \mid \text{do}(\boldsymbol{\bar{W}} \overset{\text{i.i.d.}}{\sim} \text{Bern}(\pi))]\E[(1,\boldsymbol{X}_i^T) \mid \doop(\boldsymbol{\bar{W}}_{-i} \overset{\text{i.i.d.}}{\sim} \text{Bern}(\pi))] \boldsymbol{\beta}_{1} \\
    &\quad\quad+  \E[\boldsymbol{C}_i^T \mid \doop(\boldsymbol{\bar{W}} \overset{\text{i.i.d.}}{\sim} \text{Bern}(\pi))] \boldsymbol{\gamma} + \E[\epsilon_{Y_i} \mid \text{do}(\boldsymbol{\bar{W}} \overset{\text{i.i.d.}}{\sim} \text{Bern}(\pi)) ]\\
    &\quad= (1-\pi)\E[(1,\boldsymbol{X}_i^T) \mid \doop(\boldsymbol{\bar{W}}_{-i} \overset{\text{i.i.d.}}{\sim} \text{Bern}(\pi))] \boldsymbol{\beta}_{0}  \\
    &\quad\quad+ \pi\E[(1,\boldsymbol{X}_i^T) \mid \doop(\boldsymbol{\bar{W}}_{-i} \overset{\text{i.i.d.}}{\sim} \text{Bern}(\pi))] \boldsymbol{\beta}_{1}
    + \E[\boldsymbol{C}_i^T] \boldsymbol{\gamma},
\end{align*}
where the first equality holds because
$W_i \indep \boldsymbol{X}_i$ by d-separation in $G_e$. 
The second equality holds because $W_i$ and $\boldsymbol{\bar{W}}_{-i}$ are d-separated in the explicit graph obtained by removing all incoming edges into the nodes in $\boldsymbol{\bar{W}}_{-i}$ (do-calculus Rule $1$ \citep{pearl1995}). Similarly, $\mathbf{C}_i$ and $\boldsymbol{\bar{W}}_{-i}$ are d-separated in the explicit graph obtained by removing all incoming edges into the nodes in $\boldsymbol{\bar{W}}_{-i}$. This yields
\begin{align*}
\tau_N(\pi,\eta) &:=  \frac{1}{N}\sum_{i=1}^N\left( \E[Y_i \mid \text{do}(\boldsymbol{\bar{W}} \overset{\text{i.i.d.}}{\sim} \text{Bern}(\pi))]- \E[Y_i \mid \text{do}(\boldsymbol{\bar{W}} \overset{\text{i.i.d.}}{\sim} \text{Bern}(\eta))] \right)\\
&= \boldsymbol{\omega}^N_{0}(\pi,\eta)^T\boldsymbol{\beta}_{0} + \boldsymbol{\omega}^N_{1}(\pi,\eta)^T \boldsymbol{\beta}_{1}\\
&= \boldsymbol{\omega}^N_{0}(\pi,\eta)^T \boldsymbol{\alpha}_{0} + \boldsymbol{\omega}^N_{1}(\pi,\eta)^T(\boldsymbol{\alpha}_{0} +\boldsymbol{\alpha}_{1} ) ,
\end{align*}
where $\boldsymbol{\beta}_{1} = \boldsymbol{\alpha}_{0}+ \boldsymbol{\alpha}_{1}$,
$\boldsymbol{\beta}_{0} =\boldsymbol{\alpha}_{0}$, and the weights $\boldsymbol{\omega}^N_{1}(\pi,\eta)$ and $\boldsymbol{\omega}^N_{0}(\pi,\eta)$ are as defined in the statement of Proposition \ref{lemma:reformulation}.

\end{proof}

\lemmaisolatedeffectfeat*
\begin{proof}
Recall the outcome equation \eqref{eq:outcomeX},
\begin{align*}
   Y_i \leftarrow (1,\boldsymbol{X}_i^T)\boldsymbol{\alpha}_{0}  + (W_i ,\boldsymbol{O}^T_i)  \boldsymbol{\alpha}_{1} +  \boldsymbol{C}_i^T \boldsymbol{\gamma}+ \epsilon_{Y_i},\ \ i = 1,\ldots, N
\end{align*}
with $\boldsymbol{X}_i \coloneqq (X_{i1}, X_{i2}, \ldots, X_{iP})^T \in \mathbb{R}^{P+1}$ and $\boldsymbol{O}_i \coloneqq (W_iX_{i1}, W_iX_{i2}, \ldots, W_iX_{iP})^T \in \mathbb{R}^{P}$. For any $i = 1,\ldots, N$, let $\boldsymbol{A}_i= (1,\boldsymbol{X}^T_i, W_i, \boldsymbol{O}^T_i)^T$ and let $\boldsymbol{a} = (x_0,\boldsymbol{x}, w, \boldsymbol{o})$ be a realization of $\boldsymbol{A}_i$, where $\boldsymbol{x}= (x_1, x_2, \ldots, x_{P})^T$ and $\boldsymbol{o}= (o_1, o_2, \ldots, o_{P})^T$.
We obtain 
\begin{align*}
    &\E[Y_i \mid \text{do}((1,\boldsymbol{X}^T_i, W_i, \boldsymbol{O}^T_i)^T=(x_0,\boldsymbol{x}, w, \boldsymbol{o})) ]\\
    &= \E\left[\alpha^{0}_0x_{0} + \ldots + \alpha^{0}_P x_{P} +  \alpha^{1}_0w + \alpha^{1}_1 o_1+ \ldots + \alpha^{1}_Po_P +  \boldsymbol{C}_i^T\boldsymbol{\gamma} + \epsilon_{Y_i} \mid \text{do}(\boldsymbol{A}_i= \boldsymbol{a})\right] \\
    &= \alpha^{0}_0x_{0} + \ldots + \alpha^{0}_P x_{P} +  \alpha^{1}_0w + \alpha^{1}_1 o_1+ \ldots + \alpha^{1}_Po_P +  \E[\boldsymbol{C}_i^T \mid \text{do}(\boldsymbol{A}_i= \boldsymbol{a}) ]\boldsymbol{\gamma},
\end{align*}
using $\E[\epsilon_{Y_i} \mid \doop(\boldsymbol{A}_i = \boldsymbol{a})]= \E[\epsilon_{Y_i} ] = 0$ and where $\boldsymbol{\alpha}_0 = (\alpha_0^0,\alpha_1^0,\dots,\alpha_P^0)$ and $\boldsymbol{\alpha}_1 = (\alpha_0^1,\alpha_1^1,\dots,\alpha_P^1)$. Recall that $\boldsymbol{C}_i$ contains no descendants of any variable in $\boldsymbol{A}_i= (1,\boldsymbol{X}^T_i, W_i, \boldsymbol{O}^T_i)^T$. Therefore, $\mathbf{C}_i$ and $\boldsymbol{A}_i$ are d-separated in the graph obtained from $G_e$ by removing all edges into $\boldsymbol{X}_i, W_i$, and $\boldsymbol{O}_i$, and therefore 
$\E[\boldsymbol{C}_i^T \mid \text{do}(\boldsymbol{A}_i= \boldsymbol{a}) ] = \E[\boldsymbol{C}_i^T ].$

We now compute the partial derivatives of $\E[Y_i \mid \text{do}(\boldsymbol{A}_i= \boldsymbol{a}) ]$ with respect to $x_j$, $j=0,\ldots, P$, and with respect to $o_k$, $k=1,\ldots, P$:
\begin{align*}
    &\theta_{y x_j}
    = \frac{\partial}{\partial x_j}\left(\alpha^{0}_0x_{0} + \ldots + \alpha^{0}_P x_{P} +  \alpha^{1}_0w + \alpha^{1}_1 o_1+ \ldots + \alpha^{1}_Po_P + \E[\boldsymbol{C}_i^T ]\boldsymbol{\gamma}  \right) 
    = \alpha_j^{0},\\
    &\theta_{y o_k}
    = \frac{\partial}{\partial o_i}\left(\alpha^{0}_0x_{0} + \ldots + \alpha^{0}_P x_{P} +  \alpha^{1}_0w + \alpha^{1}_1 o_1+ \ldots + \alpha^{1}_Po_P +  \E[\boldsymbol{C}_i^T ] \boldsymbol{\gamma}\right) 
    =  \alpha^{1}_k.
\end{align*}
In addition it holds that 
\begin{align*}
        \theta_{y w} &= \E[Y_i \mid \text{do}((1,\boldsymbol{X}^T_i, W_i, \boldsymbol{O}^T_i)^T=(x_0, \boldsymbol{x}, w=1, \boldsymbol{o})) ] \\
        &\quad\quad- \E[Y_i \mid \text{do}((1,\boldsymbol{X}^T_i, W_i, \boldsymbol{O}^T_i)^T=(x_0,\boldsymbol{x}, w=0, \boldsymbol{o}))]\\
    &= \left(\alpha^{0}_0x_{0} + \ldots + \alpha^{0}_P x_{P} +  \alpha^{1}_0 + \alpha^{1}_1 o_1+ \ldots + \alpha^{1}_Po_P +  \E[\boldsymbol{C}_i^T]\boldsymbol{\gamma} \right) \\
    &\quad\quad- \left(\alpha^{0}_0x_{0} + \ldots + \alpha^{0}_P x_{P} + \alpha^{1}_1 o_1+ \ldots + \alpha^{1}_Po_P +  \E[\boldsymbol{C}_i^T]\boldsymbol{\gamma} \right)\\
    &= \alpha^{1}_0,
\end{align*}
which implies that $(\boldsymbol{\alpha}^T_{0},\boldsymbol{\alpha}^T_{1})^T$ is the total joint effect of $(1,\boldsymbol{X}^T_i, W_i, \boldsymbol{O}^T_i)^T$ on $Y_i$ for all $i=1,\ldots, N$.
\end{proof}

\theoremident*

\begin{proof}

Proposition \ref{lemma:reformulation} and Lemma \ref{lemma:totaljointeffet} allow us to reduce the problem of identifying $\tau_N(\pi,  \eta)$ to the problem of identifying $(\boldsymbol{\alpha}^T_{0},\boldsymbol{\alpha}^T_{1})^T$, the total joint effect of $(1,\boldsymbol{X}^T_i, W_i, \boldsymbol{O}^T_i)^T$ on $Y_i$ for all $i=1,\ldots, N$.
Furthermore, the truncated factorization formula with respect to the explicit DAG $G_e$, given in equation \eqref{eq11}, implies the adjustment formula (Definition 3.6 in \cite{maathuis2015generalized}), that is, for each $i=1,\ldots,N$,
\begin{align*}
    f(\boldsymbol{b}_i \mid \doop(\boldsymbol{A}_i=\boldsymbol{a}_i)) = \int_{d_i} f(\boldsymbol{b}_i \mid \mathbf{z}_i, \boldsymbol{a}_i) f(\mathbf{z}_i) d\mathbf{z}_i 
\end{align*}
for pairwise disjoint node sets $\boldsymbol{A}_i, \boldsymbol{B}_i, \mathbf{Z}_i \subset \boldsymbol{V}_i$, if $\mathbf{Z}_i$ is a valid adjustment set in the explicit DAG $D_e$ corresponding to $S_e$.
See e.g.~Chapter 6.6 in \cite{peters} for a proof. Since by Proposition \ref{lemma:strucpreservwithfeat} the truncated factorization formula holds also with respect to the generic graph $\mathcal{G}$ of $G_e$, we can thus identify valid adjustment sets $\mathbf{Z}$ relative to $( \{\boldsymbol{X},W,\boldsymbol{O}\}, Y)$ in the generic graph $\mathcal{G}$.
\end{proof}

\section{Existing \& Preparatory Results for Section \ref{sec:consistency} Proofs}

\begin{restatable}[Weak Law of Large Numbers]{lemma}{wlln}
\label{WLLN}
\noindent Consider a treatment vector $\boldsymbol{\bar{W}}$ and an interaction network graph $I^N$. Given $P$ functions $h^1(\cdot),\ldots, h^P(\cdot)$, let $\boldsymbol{\bar{U}}$ be the matrix with entries $U_{ik} = h^k(\boldsymbol{\bar{W}}_{-i},I^N)$ for $i=1 ,\ldots, N$ and $k =1,\ldots, P$, and let $\boldsymbol{U}_j$ denote the $j$th row of $\boldsymbol{\bar{U}}$. Let $D(\boldsymbol{\bar{U}}, \boldsymbol{\bar{W}})$ be the dependency graph with respect to $\boldsymbol{\bar{U}}$ and $\boldsymbol{\bar{W}}$. Let
$d_{\text{max}}(N) := \max_{i \in \{1,\dots, N\}} \sum_{j=1}^ND_{ij}(\boldsymbol{\bar{U}}, \boldsymbol{\bar{W}})$ be the maximal degree of the dependency graph and let $\mu_{ij}= \E[\boldsymbol{U}_{ij}]$. If
\begin{enumerate}[i)]
    \item $\max_{\substack{i=1,\ldots, N , j=1,\ldots, P}} \mathrm{Var}(U_{ij}) \leq c \leq \infty$,
    \item $\frac{1}{N}\sum_{i=1}^N \boldsymbol{\mu}_i \rightarrow \boldsymbol{\mu^0} < \infty$, for some constant vector $\boldsymbol{\mu^0} $, and
    \item  $d_{\text{max}}(N) \in o(N)$,
\end{enumerate} 
then 
$$\frac{1}{N}\sum_{i=1}^N \boldsymbol{U}_i  \xrightarrow[]{P} \boldsymbol{\mu}^0.$$ 
\end{restatable}

\begin{proof}
We show that for each $j=1,\ldots, P$, the mean $S_{j}^N/N $, where $S_{j}^N = \sum_{i=1}^N U_{ij}$, converges in probability to its respective entry $\mu_j^0$. 
Let $\epsilon > 0$. Then
\begin{align}
    \Prob \left[ \left\rvert \frac{S_{j}^N}{N} - \mu_j^0 \right\rvert > \epsilon\right] &=
    \Prob \left[  \left\rvert\left( \frac{S_j^N}{N} - \frac{1}{N}\sum_{i=1}^N \mu_{ij}\right) + \left(\frac{1}{N}\sum_{i=1}^N \mu_{ij} - \mu_j^0 \right) \right\rvert > \epsilon\right] \notag\\
    & \leq   \Prob \left[  \left\rvert \frac{1}{N}\left(S_j^N - \sum_{i=1}^N \mu_{ij}\right) \right\rvert > \frac{\epsilon}{2} \right] +  \Prob \left[\left\rvert\frac{1}{N}\sum_{i=1}^N \mu_{ij} - \mu_j^0 \right\rvert > \frac{\epsilon}{2}\right], \label{rhs}
\end{align}
where we use the triangle inequality. 

We consider the first term on the RHS of \eqref{rhs}. Let $Y_j^N := \frac{1}{N}\sum_{i=1}^N\left(U_{ij} - \mu_{ij}\right)$. By Chebychev's inequality we get
\begin{align*}
    \Prob \left[ \left\rvert \frac{1}{N}\sum_{i=1}^N\left(U_{ij} - \mu_{ij}\right) \right\rvert > \frac{\epsilon}{2}\right] \leq \frac{4\mathrm{Var}(Y_j^N)}{\epsilon^2}.
\end{align*}
The variance of $Y_j^N$ is given by
\begin{align*}
    \mathrm{Var}(Y_j^N) &= \frac{1}{N^2} \left(\sum_{i=1}^N \mathrm{Var}(U_{ij}-\mu_{ij}) + \sum_{i=1}^N\sum_{k=1,k\neq i}^N \mathrm{Cov} (U_{ij}-\mu_{ij}, U_{kj} - \mu_{kj})\right)
     \\
     &= \frac{1}{N^2} \left(\sum_{i=1}^N \mathrm{Var}(U_{ij}) + \sum_{i=1}^N\sum_{k=1,k\neq i}^N \mathrm{Cov}(U_{ij},U_{kj})  \right)\\
     &\leq \frac{1}{N^2} \left( Nc + \sum_{i=1}^N\sum_{k=1, k\neq i}^N \mathrm{Cov}(U_{ij},U_{kj}) \right).
\end{align*}

Given a fixed $i$, we define the two sets 
\begin{align*}
    &\mathcal{C}_i = \left\{ k\in \{1,\dots,N\}\setminus i: D_{ik}(\boldsymbol{\bar{U}},\boldsymbol{\bar{W}}) =1\right\} \ \text{and} \ \\
    &\mathcal{C}_i^c = \left\{ k\in \{1,\dots,N\}\setminus i: D_{ik}(\boldsymbol{\bar{U}},\boldsymbol{\bar{W}}) =0\right\}.
\end{align*}
We now decompose
\begin{align*}
    \sum_{i=1}^N\sum_{k=1, k \neq i}^N \mathrm{Cov}(U_{ij}, U_{kj}) &= \sum_{i=1}^N \left( \sum_{k \in \mathcal{C}_i} \mathrm{Cov}(U_{ij},U_{kj})+\sum_{k \in \mathcal{C}_i^c} \mathrm{Cov}(U_{ij},U_{kj}) \right) \\
    &\leq \sum_{i=1}^N \left( c \sum_{k=1}^N D_{ik}(\boldsymbol{\bar{U}},\boldsymbol{\bar{W}}) +0\right)\\
    & \leq Nc d_{\text{max}}(N),
\end{align*}
using Cauchy-Schwarz to bound $\mathrm{Cov}(U_{ij}, U_{kj}) \leq c$ for all $i, k$. Combining all the above leads to
\begin{align*}
    \frac{4\mathrm{Var}(Y_j^N)}{\epsilon^2} \leq \frac{4cN(1 + d_{\text{max}}(N))}{\epsilon^2 N^2} =  \frac{4c}{\epsilon^2 N} +  \frac{4cd_{\text{max}}(N))}{\epsilon^2 N} \xrightarrow[]{N \rightarrow \infty} 0,
\end{align*} since by the assumption $iii)$, $d_{\text{max}}(N) \in o(N)$.

We now consider the second term on the RHS \eqref{rhs}. By assumption $ii)$ we know that $\lim_{N\rightarrow \infty} \sum_{i=1}^N\E[U_{ij}]/ N = \mu_j^0 $. 
Therefore, combining that both terms on the RHS \eqref{rhs} converge to zero implies
\begin{align*}
     \Prob \left[ \left\rvert \frac{S_j^N}{N} - \mu_j^0 \right\rvert > \epsilon\right]  \xrightarrow[]{N \rightarrow \infty} 0, 
\end{align*}
and therefore $\frac{1}{N}\sum_{i=1}^N U_{ij}  \xrightarrow[]{P} \mu_j^0$.
\end{proof}

\begin{restatable}{lemma}{lemmadepgraph}
\label{lemma:depgraph}
Let $S_e$ be an explicit SEM satisfying Assumption \ref{def:sem} with explicit DAG $G_e$. 
Let $\mathbf{Z}$ be a valid adjustment set relative to $( \{\boldsymbol{X},W,\boldsymbol{O}\}, Y)$ in the generic graph $\mathcal{G}$ of $G_e$. Suppose the population level OLS-estimator $(\boldsymbol{\gamma}_{A}, \boldsymbol{\gamma}_{Z})=\E[\mathbf{M}_i \mathbf{M}_i^T]^{-1} \E[\mathbf{M}_i^T Y_i]$ exist, where $\mathbf{M}_i = (\boldsymbol{A}_i, \mathbf{Z}_i)$ with $\boldsymbol{A}_i= (1,\boldsymbol{X}_i, W_i,\boldsymbol{O}_i)$. Let $\epsilon_i =Y_i - \mathbf{M}_i^T(\boldsymbol{\gamma}_{A}, \boldsymbol{\gamma}_{Z})$ and $\boldsymbol{\bar{\epsilon}} = (\epsilon_1, \ldots, \epsilon_N)^T$.
Then it holds that
\begin{align*}
D(\boldsymbol{\bar{X}}, \boldsymbol{\bar{W}}) 
&= D(\boldsymbol{\bar{M}}, \boldsymbol{\bar{W}}) \\
&= D(\boldsymbol{\bar{M}}^T\boldsymbol{\bar{M}}, \boldsymbol{\bar{W}}) \\
&= D(\boldsymbol{\bar{M}}\boldsymbol{\bar{\epsilon}}, \boldsymbol{\bar{W}}).
\end{align*}
\end{restatable}

\begin{proof} 

To prove equality of the four dependency graphs, we need to show that for $i \neq j \in \{1,\ldots, N\}$,
\begin{align}
    D_{ij}(\boldsymbol{\bar{X}}, \boldsymbol{\bar{W}}) =1 &\iff  D_{ij}(\boldsymbol{\bar{M}}, \boldsymbol{\bar{W}}) = 1 \text{ and } \label{lemma_auss1.0}\\
    D_{ij}(\boldsymbol{\bar{M}}, \boldsymbol{\bar{W}}) =1 &\iff  D_{ij}(\boldsymbol{\bar{M}}^T\boldsymbol{\bar{M}}, \boldsymbol{\bar{W}}) = 1 \text{ and } \label{lemma_auss1.1}\\
    D_{ij}(\boldsymbol{\bar{M}}^T\boldsymbol{\bar{M}}, \boldsymbol{\bar{W}}) = 1 &\iff D_{ij}(\boldsymbol{\bar{M}}\boldsymbol{\bar{\epsilon}}, \boldsymbol{\bar{W}}) = 1. \label{lemma_auss1.2}
\end{align}
Let us show Equivalence \eqref{lemma_auss1.0}. Let $D_{ij}(\boldsymbol{\bar{X}}, \boldsymbol{\bar{W}}) =1$. Thus, there exists $l\in \{1,\ldots,P\}$ such that either $W_j$ affects $X_{il}$ and/or $W_i$ affects $X_{jl}$ and/or $X_{il}$ and $X_{jl}$ are affected by some $W_k$, $k \in \{1,\ldots, N\} \setminus \{i, j\}$. Since $\mathbf{M}_i= (\boldsymbol{X}_i, W_i, \boldsymbol{O}_i, \mathbf{Z}_i)$ contains $X_{il}$ as well, it holds $D_{ij}(\boldsymbol{\bar{M}}, \boldsymbol{\bar{W}}) =1$. For the other direction, let $D_{ij}(\boldsymbol{\bar{M}}, \boldsymbol{\bar{W}}) =1$. Recall that $\boldsymbol{O}_i= \boldsymbol{X}_{i} W_i$. Thus, the dependency between $i$ and $j$ has to be due to the existence of $l\in \{1,\ldots,P\}$ such that either $W_j$ affects $X_{il}$ and/or $W_i$ affects $X_{jl}$ and/or $X_{il}$ and $X_{jl}$ are affected by some $W_k$, $k \in \{1,\ldots, N\} \setminus \{i, j\}$. Therefore, $D_{ij}(\boldsymbol{\bar{X}}, \boldsymbol{\bar{W}}) =1$. The proofs of equivalences \eqref{lemma_auss1.1} and \eqref{lemma_auss1.2} follow by a similar argument.
\end{proof}

We now give a lemma on how we can use the graphical notion of valid adjustment sets to recover the total joint effect $\boldsymbol{\theta}_{b\boldsymbol{a}}$ of a random vector $\boldsymbol{A}$ on a random variable $B$ with the ordinary least squares estimator. It it an adaptation of Example 1 in \cite{ema} and included for completeness.

\begin{restatable}{lemma}{lemmaadjustmentset}
\label{lemma:adjust}
Consider disjoint node sets $\boldsymbol{A},\{B\}$ and $\mathbf{Z}$ in a DAG $G=(\boldsymbol{V},\boldsymbol{E})$.
Assume that $\mathbf{Z}$ is a valid adjustment set relative to $(\boldsymbol{A},B)$ in $G$. Suppose that the conditional expectation of $B$ given $\boldsymbol{A}$ and $\mathbf{Z}$ is linear, that is, $\E[B\mid\boldsymbol{A},\mathbf{Z}]=\gamma + \boldsymbol{A}^T\boldsymbol{\gamma}_{A} + \mathbf{Z}^T\boldsymbol{\gamma}_{Z} $. Then $\boldsymbol{\gamma}_{A}=\boldsymbol{\theta}_{b\boldsymbol{a}}$, the total effect of $\boldsymbol{B}$ on $A$.
\end{restatable}

\begin{proof}
\begin{align*}
\boldsymbol{\theta}_{b\boldsymbol{a}} = \frac{\partial}{\partial \boldsymbol{a}}\E[B \mid do(\boldsymbol{A}= \boldsymbol{a}) ]
&= \frac{\partial}{\partial \boldsymbol{a}}\int_b b f(b\mid do(\boldsymbol{a})) db \\
&= \frac{\partial}{\partial \boldsymbol{a}}\int_b b \int_{\mathbf{z}} f(b \mid \boldsymbol{a},\mathbf{z}) f(\mathbf{z})d\mathbf{z} db \\
&= \frac{\partial}{\partial \boldsymbol{a}}\int_{\mathbf{z}} E[B \mid \boldsymbol{a},\mathbf{z}] f(\mathbf{z}) d\mathbf{z} \\
&=\frac{\partial}{\partial \boldsymbol{a}}\int_{\mathbf{z}} (\gamma + \boldsymbol{\gamma}_{A}^T \boldsymbol{a} + \boldsymbol{\gamma}_{Z}^T \mathbf{z}) f(\mathbf{z}) d\mathbf{z} \\
&= \frac{\partial}{\partial \boldsymbol{a}}(\gamma + \boldsymbol{\gamma}_{A}^T \boldsymbol{a} + \boldsymbol{\gamma}_{Z}^T \E[\mathbf{Z}])=\boldsymbol{\gamma}_{A},
\end{align*}
where the second equality follows by the definition of a valid adjustment set by \citet{ema}.
\end{proof}

\color{black}

We will use the following version of Stein's Lemma \citep[Theorem 3.6,][]{Ross2011} in our asymptotic normality proof.

\begin{restatable}{lemma}{ross}
\label{ross}
Let $\boldsymbol{\bar{A}} = (A_1, \ldots, A_N)^T$ be a collection of random variables such that for all $i = 1,\ldots,N$ it holds that $\E[A_i^4]<\infty$ and $\E[A_i]=0$. Let $S_N:= \sum_{i=1}^N A_i$ and $\sigma^2 = \lim_{N \rightarrow \infty} \mathrm{Var}(S_N) < \infty$. Let $\boldsymbol{\bar{W}} = (W_1,\dots, W_N)^T$ be the treatment vector and $D(\boldsymbol{\bar{A}}, \boldsymbol{\bar{W}})$ be the dependency graph with respect to $\boldsymbol{\bar{A}}$ and $\boldsymbol{\bar{W}}$. Let
$d_{\text{max}}(N) := \max_{i \in \{1,\dots, N\}} \sum_{j=1}^ND_{ij}(\boldsymbol{\bar{A}}, \boldsymbol{\bar{W}})$
be the maximal degree of $D(\boldsymbol{\bar{A}}, \boldsymbol{\bar{W}})$. Then for constants $C_1$ and $C_2$ which do not depend on $N$, $d_{\text{max}}(N)$ or $\sigma^2$, 
\begin{equation*}
    d_{\mathcal{W}}\left(\frac{S_N}{\sigma}\right)  \leq C_1 \frac{d_{\text{max}}(N)^{3/2}}{\sigma^2}\left( \sum_{i=1}^N E[A_i^4]\right)^{1/2} + C_2 \frac{d_{\text{max}}(N)^2}{\sigma^3}\sum_{i=1}^N E|A_i|^3,
    \end{equation*} where $ d_{\mathcal{W}}(\cdot)$ is the Wasserstein-distance to a standard Gaussian distribution.
\end{restatable}

\section{Proofs for Section \ref{sec:consistency}}
\label{appendix: section 4 proofs}

\propconsistencyols*

\begin{proof}

Recall that the OLS-estimator of $\boldsymbol{\alpha}^{\text{full}}$
is given by
\begin{equation*}
\boldsymbol{\hat{\alpha}}^{\text{full}}=(\boldsymbol{\bar{M}}^T\boldsymbol{\bar{M}})^{-1}\boldsymbol{\bar{M}}^T\boldsymbol{\bar{Y}},
\end{equation*}
where $\boldsymbol{\bar{M}}\in \R^{N \times (|\boldsymbol{A}_i|+|\mathbf{Z}_i|)}$ is the data matrix corresponding to $\mathbf{M}_i^T$ of all units $i = 1,\ldots, N$, and similarly, $\boldsymbol{\bar{Y}} \in \R^N$ is the vector of outcomes $Y_i$. Here, $\boldsymbol{A}_i=(1,\boldsymbol{X}^T_i,W_i,\boldsymbol{O}^T_i)^T$. We denote the first $|\boldsymbol{A}_i|$ components of $\boldsymbol{\hat{\alpha}}^{\text{full}}$ with $\boldsymbol{\hat{\alpha}}$, which is an estimator of $\boldsymbol{\alpha}$. 

First, we show that $\boldsymbol{\hat{\alpha}}$ converges in probability to $\boldsymbol{\alpha}$. By Assumption \ref{def:sem} on the explicit SEM $S_e$ and Condition $iv)$ of the current theorem, the population OLS-estimator $(\boldsymbol{\gamma}_{A}, \boldsymbol{\gamma}_{Z})=\E[\mathbf{M}_i \mathbf{M}_i^T]^{-1} \E[\mathbf{M}_i^T Y_i]$ exists and is constant for each $i=1,\ldots, N$. As a result, $\E[\mathbf{M}_i \epsilon_i]=0$, where $\epsilon_i =Y_i - \mathbf{M}_i^T(\boldsymbol{\gamma}_{A}, \boldsymbol{\gamma}_{Z})$  for $i=1,\ldots, N$. Therefore, it also holds that
\begin{equation} \label{limit:xeps}
    \lim_{N\rightarrow \infty}\frac{1}{N}\sum_{i=1}^N \E\left[\mathbf{M}_i \epsilon_i\right] = 0.
\end{equation}
We will use this property to apply the Weak Law of Large Numbers (Lemma \ref{WLLN}). Let $\boldsymbol{\bar{\epsilon}} = (\epsilon_1, \ldots, \epsilon_N)^T$. By Lemma \ref{lemma:depgraph} it holds that 
$D(\boldsymbol{\bar{X}}, \boldsymbol{\bar{W}}) =  D(\boldsymbol{\bar{M}}^T\boldsymbol{\bar{M}}, \boldsymbol{\bar{W}}) =  D(\boldsymbol{\bar{M}} \boldsymbol{\bar{\epsilon}}, \boldsymbol{\bar{W}})$. Thus, we can apply Lemma \ref{WLLN} to $\mathbf{M}_i\mathbf{M}_i^T$ by Conditions $ii)$, $iv)$, and $v)$. We can also apply it to $\mathbf{M}_i\epsilon_i$ by Conditions $ii)$, $iii)$, and $v)$ and equation \eqref{limit:xeps}. Therefore, we obtain
\begin{align}
 \boldsymbol{\hat{\alpha}}^{\text{full}}- (\boldsymbol{\gamma}_{A}, \boldsymbol{\gamma}_{Z}) &= 
   \left( \left( \frac{1}{N}\sum_{i=1}^N\MiMit\right)^{-1}\left( \frac{1}{N}\sum_{i=1}^N \mathbf{M}_i Y_i\right) -(\boldsymbol{\gamma}_{A}, \boldsymbol{\gamma}_{Z}) \right) \notag\\
    &=  \left( \left( \frac{1}{N}\sum_{i=1}^N \MiMit\right)^{-1}\left(\frac{1}{N} \sum_{i=1}^N \mathbf{M}_i(\mathbf{M}_i^T(\boldsymbol{\gamma}_{A}, \boldsymbol{\gamma}_{Z}) + \epsilon_i )\right) - (\boldsymbol{\gamma}_{A}, \boldsymbol{\gamma}_{Z})\right) \notag\\
    &\xrightarrow[]{P} 
    \E[\mathbf{M}_i \mathbf{M}_i^T]^{-1}
     \left(\lim_{N\rightarrow \infty}\frac{1}{N}\sum_{i=1}^N \E\left[\mathbf{M}_i \epsilon_i\right]\right), \label{rhs:consistencyols}
\end{align}
where the convergence in probability is due to Lemma \ref{WLLN} and the continuous mapping theorem. By equation \eqref{limit:xeps}, we therefore conclude that the RHS of \eqref{rhs:consistencyols} is zero and therefore $\boldsymbol{\hat{\alpha}}^{\text{full}}$ converges in probability to $(\boldsymbol{\gamma}_{A}, \boldsymbol{\gamma}_{Z})$. 

We now show that $\boldsymbol{\gamma}_{A}=\boldsymbol{\alpha}$ by applying Lemma \ref{lemma:adjust}. We first show that the conditions for Lemma \ref{lemma:adjust} hold. Let $\mathbf{P'} = \mathbf{P} \setminus \mathbf{Z}$ and $\mathbf{Z'} = \mathbf{Z} \setminus \mathbf{P}$, with $\mathbf{P}$ and $\mathbf{Z}$ denoting the generic set corresponding to $\mathbf{P}_i$ and $\mathbf{Z}_i$. Since $\mathbf{Z}$ is a valid adjustment relative to $(\{\boldsymbol{X},W,\boldsymbol{O}\})$ in $\mathcal{G}$ it holds that $\mathbf{P'} \perp_{\mathcal{G}} (\boldsymbol{X}, \{W\}, \boldsymbol{O}) \mid \mathbf{Z}$ and 
$\mathbf{Z'} \perp_{\mathcal{G}} Y \mid \boldsymbol{X},\{W\}, \boldsymbol{O},\mathbf{P}$, where $\mathbf{P'} = \mathbf{P} \setminus \mathbf{Z}$ and $\mathbf{Z'}= \mathbf{Z} \setminus \mathbf{P}$. Here we use that $\mathbf{P}$ is a valid adjustment set since there are no mediators between $\{\boldsymbol{X},W,\boldsymbol{P}\}$ and $Y$, that is, $\mathrm{cn}(\{\boldsymbol{X},W,\boldsymbol{P}\},Y,G)=\{Y\}$. Since the distribution of $\mathbf{V}_i$ is Markov to $\mathcal{G}$ for all $i$ by Proposition \ref{lemma:strucpreservwithfeat} it follows that $\mathbf{P'}_i \indep \boldsymbol{X}_i, \{W_i\}, \boldsymbol{O}_i) \mid \mathbf{Z}_i$ and 
$\mathbf{Z'}_i \indep Y_i \mid \boldsymbol{X}_i,\{W_i\}, \boldsymbol{O}_i,\mathbf{P}_i$. By Assumption \ref{def:sem} and Condition $vi)$ of the current theorem it therefore follows that
\begin{align*}
    \E[Y_i\mid \boldsymbol{X}_i,W_i,\boldsymbol{O}_i,\mathbf{Z}_i]
    &= \E[\E[Y_i\mid \boldsymbol{X}_i,W_i,\boldsymbol{O}_i,\mathbf{Z}_i,\mathbf{P'}_i]\mid \boldsymbol{X}_i,W_i,\boldsymbol{O}_i,\mathbf{Z}_i]\\
    &=\E[\E[Y_i\mid \boldsymbol{X}_i,W_i,\boldsymbol{O}_i,\mathbf{P}_i]\mid \boldsymbol{X}_i,W_i,\boldsymbol{O}_i,\mathbf{Z}_i]\\ 
    &= \E[ (1,\boldsymbol{X}_i^T) \boldsymbol{\alpha}_{0}  + (W_i ,\boldsymbol{O}^T_i)  \boldsymbol{\alpha}_{1} +  \boldsymbol{P}_i^T \boldsymbol{\gamma}_P\mid \boldsymbol{X}_i,W_i,\boldsymbol{O}_i,\mathbf{Z}_i]\\
    &= (1,\boldsymbol{X}_i^T) \boldsymbol{\alpha}_{0}  + (W_i ,\boldsymbol{O}^T_i)  \boldsymbol{\alpha}_{1} + \E[\mathbf{P}^T_i \mid \mathbf{Z}_i] \gamma_P \\
    &= (1,\boldsymbol{X}_i^T) \boldsymbol{\alpha}_{0}  + (W_i ,\boldsymbol{O}^T_i) \boldsymbol{\alpha}_{1} +  \mathbf{Z}^T_i \delta \gamma_P,
\end{align*}
where $\gamma_P$ is the vector of nonzero entries of $\gamma$. 
We can therefore apply Lemma \ref{lemma:adjust} and conclude that $\boldsymbol{\gamma}_{A}=\boldsymbol{\theta}_{y\boldsymbol{a}}$. Furthermore, we have shown that $\boldsymbol{\theta}_{y\boldsymbol{a}} = \boldsymbol{\alpha}$, that is, the joint total causal effects equal the coefficients $\boldsymbol{\alpha}$. Therefore, the components $\hat{\boldsymbol{\gamma}}_A$ of the estimator $\boldsymbol{\hat{\alpha}}^{\text{full}} = (\hat{\boldsymbol{\gamma}}_A, \hat{\boldsymbol{\gamma}}_D)$ converge in probability to the coefficients $\boldsymbol{\alpha}$.

Finally, we apply Slutsky's theorem and Condition i) to obtain that $\hat{\tau}_N(\pi, \eta) - \tau_N(\pi,\eta) = \boldsymbol{\omega}^N_{0}(\pi,\eta)^T ( \boldsymbol{\hat{\alpha}}_{0}-\boldsymbol{\alpha}_{0}) + \boldsymbol{\omega}^N_{1}(\pi,\eta)^T(\boldsymbol{\hat{\alpha}}_{0}- \boldsymbol{\alpha}_{0} +\boldsymbol{\hat{\alpha}}_{1} - \boldsymbol{\alpha}_{1} )  \xrightarrow[]{P} 0.$
\end{proof}

\propnormalityols*
\begin{proof}
Recall the OLS-estimator of $\boldsymbol{\alpha}^{\text{full}}$ is given by
\begin{equation*}
\boldsymbol{\hat{\alpha}}^{\text{full}}=(\boldsymbol{\bar{M}}^T\boldsymbol{\bar{M}})^{-1}\boldsymbol{\bar{M}}^T\boldsymbol{\bar{Y}},
\end{equation*}
where $\boldsymbol{\bar{M}}\in \R^{N \times (|\boldsymbol{A}_i|+|\mathbf{Z}_i|)}$ is the data matrix corresponding to $\mathbf{M}_i^T$ of all units $i = 1,\ldots, N$, and similarly, $\boldsymbol{\bar{Y}} \in \R^N$ is the vector of outcomes $Y_i$. We denote the first $|\boldsymbol{A}_i|$ components of $\boldsymbol{\hat{\alpha}}^{\text{full}}$ with $\boldsymbol{\hat{\alpha}}$, which is an estimator of $\boldsymbol{\alpha}$.  First, we show that the properly scaled components of the estimator $\boldsymbol{\hat{\alpha}}^{\text{full}}$ corresponding to $\boldsymbol{A}_i = (1,\boldsymbol{X}^T_i, W_i, \boldsymbol{O}^T_i)^T$ converge in distribution to a multivariate Gaussian distribution. 

By Assumption \ref{def:sem} on the explicit SEM and Condition $iv)$ of Theorem \ref{prop:consistency-ols}, the population OLS-estimator $(\boldsymbol{\gamma}_{A}, \boldsymbol{\gamma}_{Z})=\E[\mathbf{M}_i \mathbf{M}_i^T]^{-1} \E[\mathbf{M}_i^T Y_i]$ exists and is constant for each $i=1,\ldots, N$. As a result, $\E[\mathbf{M}_i \epsilon_i]=0$, where $\epsilon_i =Y_i - \mathbf{M}_i^T(\boldsymbol{\gamma}_{A}, \boldsymbol{\gamma}_{Z})$ for $i=1,\ldots, N$.
By the same argument as in the proof of Theorem \ref{prop:consistency-ols}, we obtain that
\begin{align}
\sqrt{N}\left( \boldsymbol{\hat{\alpha}}^{\text{full}}- (\boldsymbol{\gamma}_{A}, \boldsymbol{\gamma}_{Z}) \right) 
    &= \left( \frac{1}{N} \sum_{i=1}^N \MiMit\right)^{-1}\left( \frac{1}{\sqrt{N}} \sum_{i=1}^N\mathbf{M}_i \epsilon_i \right). \label{rhs_clt}
\end{align}
By Theorem \ref{theorem:identificantion}, $\boldsymbol{\gamma}_A = \boldsymbol{\alpha}$. By Lemma \ref{lemma:depgraph},
$D(\boldsymbol{\bar{X}}, \boldsymbol{\bar{W}}) =  D(\boldsymbol{\bar{M}}^T\boldsymbol{\bar{M}}, \boldsymbol W)$. Thus, we can apply Lemma \ref{WLLN} to $\mathbf{M}_i\mathbf{M}_i^T$ and obtain for the first term on the RHS of \eqref{rhs_clt} that
\begin{equation*}
\left(\frac{1}{N} \sum_{i=1}^N \MiMit\right)^{-1}\xrightarrow[]{P} \Sigma_{\mathbf{M} \mathbf{M}} ^{-1},    
\end{equation*}
for some finite matrix $ \Sigma_{\mathbf{M} \mathbf{M}}$, using the continuous mapping theorem.

We will use the Cramér-Wold device to show multivariate asymptotic normality of the second term on the RHS \eqref{rhs_clt},
\begin{align}
    \frac{1}{\sqrt[]{N}} \sum_{i=1}^N\mathbf{M}_i \epsilon_i &=
    \left(  \frac{1}{\sqrt[]{N}} \sum_{i=1}^N \ M_{i1}\epsilon_i,\ldots, \frac{1}{\sqrt[]{N}} \sum_{i=1}^N  M_{iP}\epsilon_i  \right)^T. \label{stepols}
\end{align}
Let $\boldsymbol{a} \in \R^P$ be a vector of scalars such that $\boldsymbol{a}^T\boldsymbol{a} = \boldsymbol{1}$, where $\boldsymbol{1}$ denotes the vector of ones of length $P$. We now apply a version of Stein's Lemma, Lemma \ref{ross}, to $A_i := \frac{\epsilon_i}{\sqrt{N}}\sum_{j=1}^P a_j M_{ij}$. By Condition $ii)$ the fourth moment of $A_i$ is bounded. We now show that the variance of $S_N :=\sum_{i=1}^N A_i $ converges. Using that $\E[\boldsymbol{M}_i\epsilon_i]=0$ it follows that
\begin{align*}
 \mathrm{Var} \left(\frac{1}{\sqrt[]{N}} \sum_{i=1}^N\mathbf{M}_i \epsilon_i\right) &=
 \frac{1}{N}\sum_{i=1}^N\E\left[ \epsilon_i^2 \textbf{M}_i \textbf{M}_i ^T\right],
\end{align*}
which, by Condition $iii)$, converges for $N\rightarrow \infty$ to a finite matrix $\Sigma_{\epsilon^2\mathbf{M} \mathbf{M}} < \infty$. Therefore, the variance of $S_N=\boldsymbol{a}^T \frac{1}{\sqrt[]{N}} \sum_{i=1}^N\mathbf{M}_i \epsilon_i $ is given by $\boldsymbol{a}^T \Sigma_{\epsilon^2\mathbf{M} \mathbf{M}}\boldsymbol{a}$, which we denote by $\sigma^2$. Since $\E[M_i\epsilon_i]=0$ it also holds that $\E[A_i]=\E\left[ \frac{\epsilon_i}{\sqrt{N}}\sum_{j=1}^P a_j M_{ij}\right] = 0$. Thus, all assumptions on $A_i$ of Lemma \ref{ross} are met.

We now show that that $S_N$ converges to a Gaussian distribution, by applying Lemma \ref{ross}. The dependency graph $D(\boldsymbol{\bar{U}}, \boldsymbol{\bar{W}})$ on $\boldsymbol{A}=(A_1,\ldots, A_N)$ equals $D(\boldsymbol{\bar{X}}, \boldsymbol{\bar{W}})$ by Lemma \ref{lemma:depgraph}. Thus, let $$d_{\text{max}}(N) := \max_{i \in \{1,\dots, N\}} \sum_{j=1}^ND_{ij}(\boldsymbol{\bar{X}}, \boldsymbol{\bar{W}})$$ be the maximal degree of the dependency graph $D(\boldsymbol{\bar{X}}, \boldsymbol{\bar{W}})$. By Lemma \ref{ross} we get,
\begin{align*}
    &d_{\mathcal{W}}\left(\frac{S_N}{\sigma}\right)  \leq C_1 \frac{d_{\text{max}}(N)^{3/2}}{\sigma^2}\left( \sum_{i=1}^N \E[A_i^4]\right)^{1/2} + C_2 \frac{d_{\text{max}}(N)^2}{\sigma^3}\sum_{i=1}^N \E|A_i|^3\\
    &= C_1 \frac{d_{\text{max}}(N)^{3/2}}{\sigma^2}\left( \frac{1}{N^{2}}\sum_{i=1}^N \E\left[\left( \epsilon_i\sum_{j=1}^P a_j M_{ij} \right)^4\right]\right)^{1/2} 
    \\
     & \hspace{3cm} + C_2 \frac{d_{\text{max}}(N)^2}{\sigma^3}\frac{1}{N^{3/2}}\sum_{i=1}^N \E\left|\epsilon_i\sum_{j=1}^P a_j M_{ij} \right|^3\\
     &= C_1 \frac{d_{\text{max}}(N)^{3/2}}{\sigma^2}\frac{1}{\sqrt{N}}\left(\frac{1}{N} \sum_{i=1}^N \E\left[\left( \epsilon_i\sum_{j=1}^P a_j M_{ij} \right)^4\right]\right)^{1/2} \\
     & \hspace{3cm} + C_2 \frac{d_{\text{max}}(N)^2}{\sigma^3}\frac{1}{\sqrt{N}}\left(\frac{1}{N}\sum_{i=1}^N \E\left|\epsilon_i\sum_{j=1}^P a_j M_{ij} \right|^3 \right).\\
\end{align*}
The term $ \E\left[\left( \epsilon_i\sum_{j=1}^P a_j M_{ij} \right)^4\right]$ is bounded by Condition $ii)$. The term $\E\left|\epsilon_i\sum_{j=1}^P a_j M_{ij} \right|^3$ is also bounded by Condition $ii)$ since
\begin{align*}
\left(\E\left|\epsilon_i\sum_{j=1}^P a_j M_{ij} \right|^3 \right)^2
    \leq \E \left[ \left(\epsilon_i\sum_{j=1}^P M_{ij}\right) ^6 \right],
\end{align*}
by the property $\boldsymbol{a}^T\boldsymbol{a} = \boldsymbol{1}$, Jensen's inequality and the convexity of the function $x \mapsto x^2$. Therefore
\begin{align*}
    &
    \E \left| \epsilon_i\sum_{j=1}^P M_{ij} \right|^3 
    \leq \sqrt{\E \left[ \left(\epsilon_i\sum_{j=1}^P M_{ij}\right) ^6 \right]}.
\end{align*}

Thus, $  d_{\mathcal{W}}\left(\frac{S_N}{\sigma}\right) \rightarrow 0$ for $N \rightarrow \infty$ if

\begin{align*}
    &\frac{d_{\text{max}}(N)^{3/2}}{\sqrt{N}} \rightarrow 0 \implies \frac{d_{\text{max}}(N)^{3}}{N} \rightarrow 0 \implies \frac{d_{\text{max}}(N)}{N^{1/3}} \rightarrow 0 \implies d_{\text{max}}(N) \in o(N^{1/3}),\\
    &\frac{d_{\text{max}}(N)^{2}}{\sqrt{N}} \rightarrow 0 \implies \frac{d_{\text{max}}(N)^{4}}{N} \rightarrow 0 \implies \frac{d_{\text{max}}(N)}{N^{1/4}} \rightarrow 0 \implies d_{\text{max}}(N) \in o(N^{1/4}),\\
\end{align*}
which is the case by Condition $i)$. We obtain that 
\begin{align*}
    &\boldsymbol{a}^T \frac{1}{\sqrt[]{N}} \sum_{i=1}^N\mathbf{M}_i \epsilon_i \xrightarrow[]{d} \mathcal{N}(0, \boldsymbol{a}^T \Sigma_{\epsilon^2\mathbf{M} \mathbf{M}} \boldsymbol{a})\\
    & \implies \frac{1}{\sqrt[]{N}} \sum_{i=1}^N\mathbf{M}_i \epsilon_i \xrightarrow[]{d}{} \mathcal{N}_P(0,  \Sigma_{\epsilon^2\mathbf{M} \mathbf{M}} )\\
    & \implies  \sqrt[]{N} (\boldsymbol{\hat{\alpha}}^{\text{full}}- \boldsymbol{\alpha}^{\text{full}} )= 
  \left( \frac{1}{N} \sum_{i=1}^N \MiMit\right)^{-1}\left( \frac{1}{\sqrt[]{N}} \sum_{i=1}^N\mathbf{M}_i \epsilon_i\right) \xrightarrow[]{d}
  \mathcal{N}_P(0, \Sigma_{\mathbf{M} \mathbf{M}}^{-1}\Sigma_{\epsilon^2\mathbf{M} \mathbf{M}} \Sigma_{\mathbf{M}\mathbf{M}}^{-1}),
\end{align*}
where $\boldsymbol{\alpha}^{\text{full}}= (\boldsymbol{\alpha}, \boldsymbol{\gamma}_{Z})$, the second implication is by Cramér-Wold device and the convergence in distribution follows by Slutsky's theorem.

Finally, we apply the delta method to see that the properly scaled $\hat{\tau}_N(\pi, \eta)$ is also asymptotically multivariate normal distributed:
\begin{align*}
&\sqrt{N}\left(\hat{\tau}_N(\pi, \eta) - \tau_N(\pi,\eta) \right) =\\ &\sqrt{N}\left( \boldsymbol{\omega}^N_{0}(\pi,\eta)^T ( \boldsymbol{\hat{\alpha}}_{0}-\boldsymbol{\alpha}_{0}) + \boldsymbol{\omega}^N_{1}(\pi,\eta)^T(\boldsymbol{\hat{\alpha}}_{0}- \boldsymbol{\alpha}_{0} +\boldsymbol{\hat{\alpha}}_{1}-\boldsymbol{\alpha}_{1} ) \right) 
\xrightarrow[]{d} \mathcal{N}(0, \sigma^2),    
\end{align*}
 using Condition i) from Theorem \ref{prop:consistency-ols}, where 
 \begin{align*}
       \sigma^2 = \begin{pmatrix}
\boldsymbol{\omega}_{0}(\pi,\eta) + \boldsymbol{\omega}_{1}(\pi,\eta ) \vspace{0.1cm} \\ 
\boldsymbol{\omega}_{1}(\pi,\eta ) \vspace{0.1cm}\\
\boldsymbol{0} \vspace{0.1cm} 
\end{pmatrix}^T\Sigma_{\mathbf{M} \mathbf{M}}^{-1}\Sigma_{\epsilon^2\mathbf{M} \mathbf{M}} \Sigma_{\mathbf{M}\mathbf{M}}^{-1}\begin{pmatrix}
\boldsymbol{\omega}_{0}(\pi,\eta) + \boldsymbol{\omega}_{1}(\pi,\eta ) \vspace{0.1cm} \\ 
\boldsymbol{\omega}_{1}(\pi,\eta ) \vspace{0.1cm}\\
\boldsymbol{0} \vspace{0.1cm} \end{pmatrix}
 \end{align*}
 by the delta method.
\end{proof}

\begin{restatable}[Variance Estimation]{lemma}{lemmavarest}
\label{lemma:varest} 
Under the assumptions of Theorem \ref{prop:normality-ols}, the variance $\sigma^2$ can be consistently estimated by 
\begin{align*}
   \hat{\sigma}_N^2 =  \begin{pmatrix}
\boldsymbol{\omega}^N_{0}(\pi,\eta) + \boldsymbol{\omega}^N_{1}(\pi,\eta ) \vspace{0.1cm} \\ 
\boldsymbol{\omega}^N_{1}(\pi,\eta ) \vspace{0.1cm}\\
\boldsymbol{0} \vspace{0.1cm} 
\end{pmatrix}^T&\left(\frac{1}{N}\boldsymbol{\bar{M}}^T\boldsymbol{\bar{M}}\right)^{-1}\left(\frac{1}{N}\boldsymbol{\bar{M}}^T \Delta \boldsymbol{\bar{M}}\right)\\
   &\left(\frac{1}{N}\boldsymbol{\bar{M}}^T \boldsymbol{\bar{M}}\right)^{-1}
\begin{pmatrix}
\boldsymbol{\omega}^N_{0}(\pi,\eta) + \boldsymbol{\omega}^N_{1}(\pi,\eta ) \vspace{0.1cm} \\ 
\boldsymbol{\omega}^N_{1}(\pi,\eta ) \vspace{0.1cm}\\
\boldsymbol{0} \vspace{0.1cm} 
\end{pmatrix},
\end{align*}
where $\Delta = diag(\hat{\epsilon}_1^2,\ldots, \hat{\epsilon}_N^2)$ is a diagonal matrix with squared residuals on the diagonal, that is, $\hat{\epsilon}_i:=Y_i - \mathbf{M}_i^T\boldsymbol{\hat{\alpha}}^{full}$ for $i=1,\ldots, N$, and $\boldsymbol{\hat{\alpha}}^{full}$ is given in equation \eqref{alphahatols}.
\end{restatable}

\begin{proof}
    By Condition i) of Theorem \ref{prop:consistency-ols} the weights $\boldsymbol{\omega}^N_{0}(\pi,\eta)$ and $\boldsymbol{\omega}^N_{1}(\pi,\eta)$ converge and therefore we only need to show that 
    \begin{align*}
         \left(\frac{1}{N}\boldsymbol{\bar{M}}^T\boldsymbol{\bar{M}}\right)^{-1}\left(\frac{1}{N}\boldsymbol{\bar{M}}^T \Delta Z\right)\left(\frac{1}{N}\boldsymbol{\bar{M}}^T \boldsymbol{\bar{M}}\right)^{-1} \xrightarrow[]{P} \Sigma_{\mathbf{M} \mathbf{M}}^{-1}\Sigma_{\epsilon^2\mathbf{M} \mathbf{M}} \Sigma_{\mathbf{M}\mathbf{M}}^{-1}.
    \end{align*}
    This is implied if we show that
    \begin{align*}
        \frac{1}{N}\sum_{i=1}^N \mathbf{M}_i\mathbf{M}_i^T \xrightarrow[]{P} \Sigma_{\mathbf{M}\mathbf{M}},
    \end{align*}  where  $ \Sigma_{\mathbf{M} \mathbf{M}}= \lim_{N\rightarrow \infty}\frac{1}{N}\sum_{i=1}^N \E\left[\mathbf{M}_i\mathbf{M}_i^T\right]$, which follows immediately from Condition $ii)$ of Theorem \ref{prop:normality-ols} and Lemma \ref{WLLN}, and
      \begin{align}
        \frac{1}{N}\sum_{i=1}^N \hat{\epsilon}_i^2\mathbf{M}_i\mathbf{M}_i^T \xrightarrow[]{P} \Sigma_{\epsilon^2\mathbf{M} \mathbf{M}}, \label{formsecond}
    \end{align}
where $\Sigma_{\epsilon^2\mathbf{M}
    \mathbf{M}} = \lim_{N\rightarrow \infty}\frac{1}{N}\sum_{i=1}^N\E\left[ \epsilon_i^2\mathbf{M}_i \mathbf{M}_i ^T\right] $, with $\epsilon_i =Y_i - \mathbf{M}_i^T\boldsymbol{\alpha}^{\text{full}}$ and $\hat{\epsilon}_i =Y_i - \mathbf{M}_i^T\boldsymbol{\hat{\alpha}}^{\text{full}}$.
    To show \eqref{formsecond}, we start with  
\begin{align*}
    \hat{\epsilon_i}^2 &= \left(Y_i - \mathbf{M}_i^T\boldsymbol{\hat{\alpha}}^{full}\right)^2 \\
    &=
    \left(Y_i - \mathbf{M}_i^T\boldsymbol{\hat{\alpha}}^{full} - \mathbf{M}_i^T\boldsymbol{\alpha}^{\text{full}} + \mathbf{M}_i^T\boldsymbol{\alpha}^{\text{full}} \right)^2\\
    &= \left(\epsilon_i + \mathbf{M}_i^T\left(\boldsymbol{\alpha}^{\text{full}} -  \boldsymbol{\hat{\alpha}}^{full} \right) \right)^2\\
    &= \epsilon_i^2 + 2\epsilon_i \mathbf{M}_i^T\left(\boldsymbol{\alpha}^{\text{full}} -  \boldsymbol{\hat{\alpha}}^{full} \right) + \left( \mathbf{M}_i^T\left(\boldsymbol{\alpha}^{\text{full}} -  \boldsymbol{\hat{\alpha}}^{full} \right)\right)^2. 
\end{align*} We now use the Cramér-Wold device to show \eqref{formsecond}. We thus assume w.l.o.g. that $\mathbf{M}_i \in \R$. We now consider
\begin{align}
       &\frac{1}{N}\sum_{i=1}^N \hat{\epsilon}_i^2\mathbf{M}_i\mathbf{M}_i^T \notag  \\     &=\frac{1}{N}\sum_{i=1}^N  \epsilon_i^2 \mathbf{M}_i\mathbf{M}_i^T +   \frac{2}{N}\sum_{i=1}^N \epsilon_i \mathbf{M}_i^T\left(\boldsymbol{\alpha}^{\text{full}} -  \boldsymbol{\hat{\alpha}}^{full}  \right) \mathbf{M}_i\mathbf{M}_i^T \notag \\
       &\quad\quad+    \frac{1}{N}\sum_{i=1}^N \left( \mathbf{M}_i^T\left(\boldsymbol{\alpha}^{\text{full}} -  \boldsymbol{\hat{\alpha}}^{full} \right) \right)^2\mathbf{M}_i\mathbf{M}_i^T \notag \\
       &=\frac{1}{N}\sum_{i=1}^N  \epsilon_i^2 \mathbf{M}_i^2 +   \left(\boldsymbol{\alpha}^{\text{full}} -  \boldsymbol{\hat{\alpha}}^{full}  \right)\frac{2}{N}\sum_{i=1}^N \epsilon_i \mathbf{M}_i^3 +  \left(\boldsymbol{\alpha}^{\text{full}} -  \boldsymbol{\hat{\alpha}}^{full} \right)^2  \frac{1}{N}\sum_{i=1}^N \mathbf{M}_i^4, \label{usethis}
\end{align}
where the first term in equation \eqref{usethis} converges in probability to $\Sigma_{\epsilon^2\mathbf{M} \mathbf{M}}$ by Condition $iii)$ of Theorem \ref{prop:normality-ols} and Lemma \ref{WLLN}, and the second and third terms in equation \eqref{usethis} converge in probability to $0$, due to the consistency of $\boldsymbol{\hat{\alpha}}^{full}$, which is implied by Theorem \ref{prop:normality-ols}, and the regularity conditions in Condition $ii)$ of Theorem \ref{prop:normality-ols}. Thus, by Cramér-Wold device, we have shown equation \eqref{formsecond} which concludes the proof.
\end{proof}

\begin{figure}[t]

\centering
\begin{subfigure}[b]{0.9\textwidth}
\includegraphics[width=0.97\linewidth]{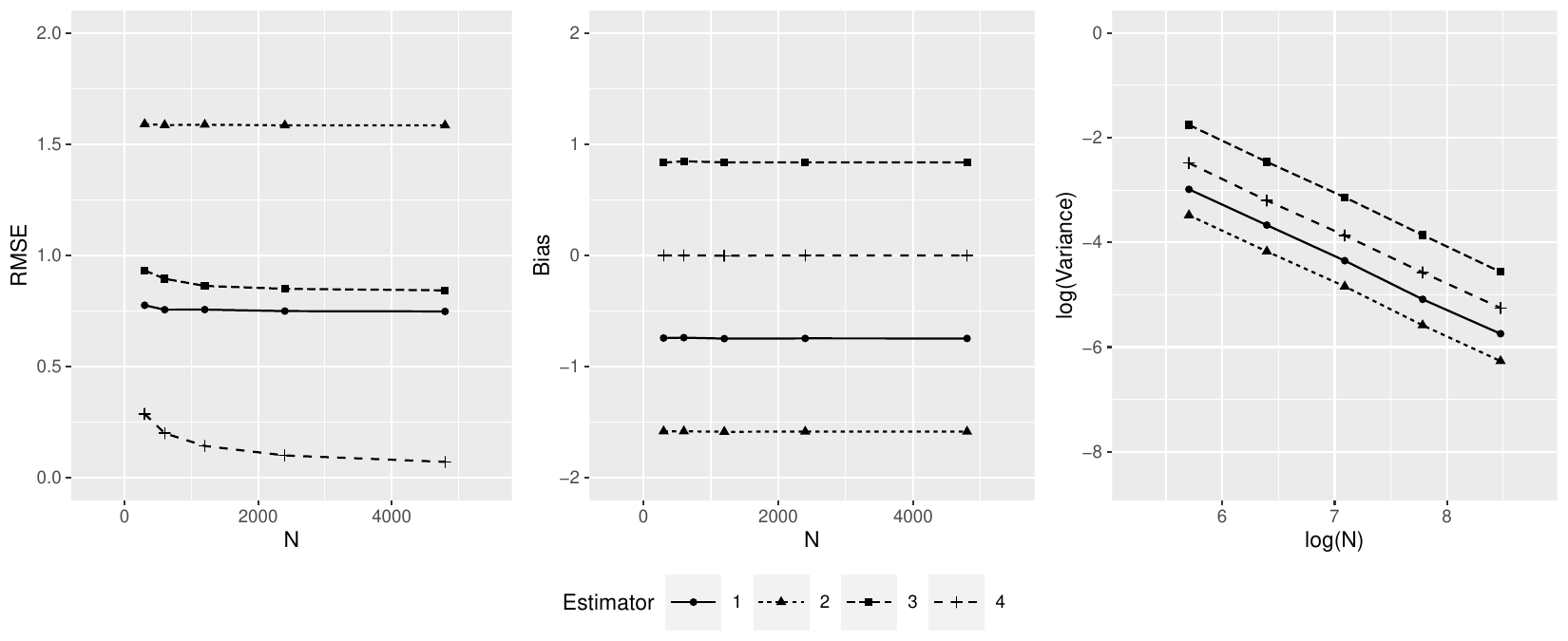}
\caption{}
\label{res:fam}
\end{subfigure}
\vfill
\begin{subfigure}[b]{0.9\textwidth}
\includegraphics[width=0.97\linewidth]{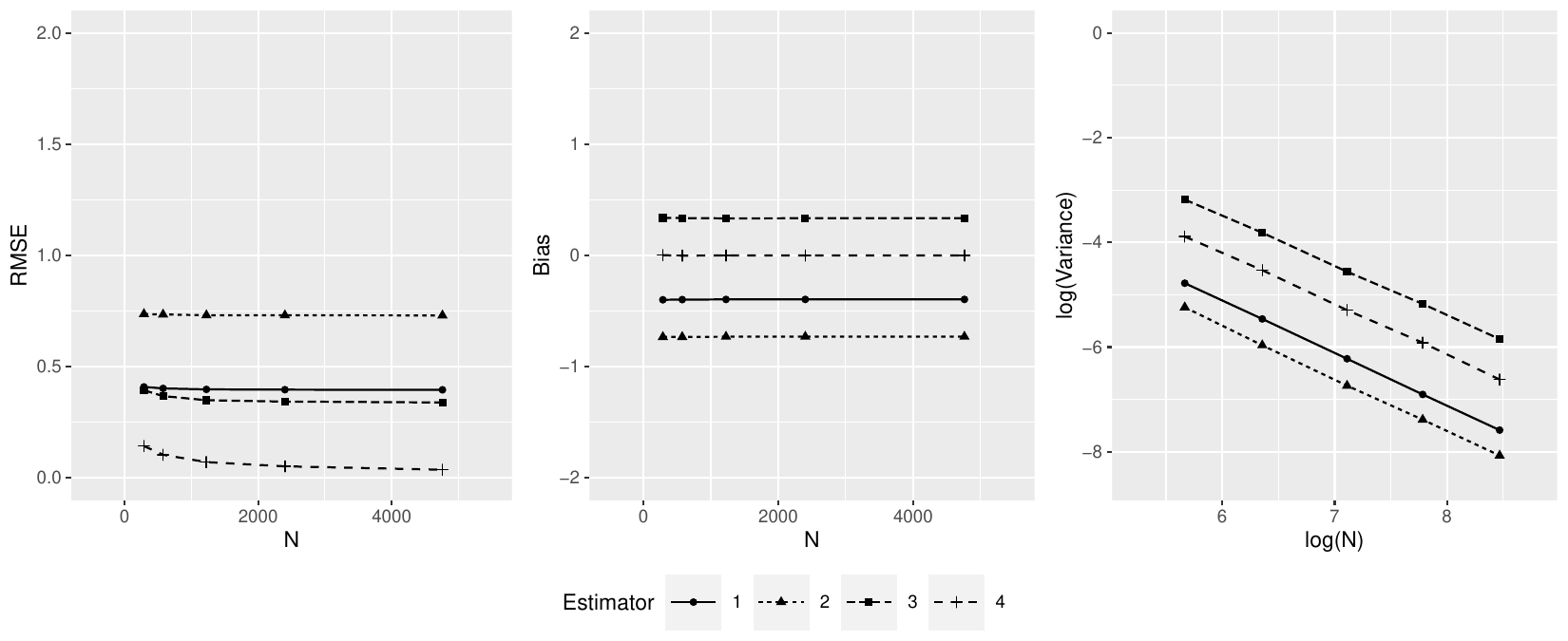}
\caption{}
\label{res:3dlatt}
\end{subfigure}
\caption{Empirical RMSE, bias and log variance plots \subref{res:fam} for the estimation of $\tau_N(1, 0)$ in family networks and \subref{res:3dlatt} for the estimation of $\tau_N(0.5, 0.1)$ in $2$-d lattices using the naive (1), confounding adjusted (2), interference adjusted (3) and fully adjusted estimator (4), respectively}
\end{figure}

\section{Empirical Validation}

\subsection{Further Empirical Results} \label{app:furtherres}

In Figures \ref{res:fam} and \ref{res:3dlatt} we show the results of the simulation study for the family networks and the $2$-d lattices, respectively. They conform to the behavior expected per our theoretical results, that is, $\sqrt{N}$-consistency.

\begin{table}[t]

\centering
\begin{tabular}{r|rrr}
$N$ & $I(N,10/N)$ & Family & 2d-lattice \\ 
  \hline
300 & 10.87 & 39.06 & 0.97 \\ 
600 & 4.86 & 18.58 & 0.54 \\ 
1200 & 2.29 & 9.73 & 0.25 \\ 
2400 & 1.06 & 4.68 & 0.17 \\ 
4800 & 0.60 & 3.39 & 0.07
\end{tabular}
\caption{Scaled to $N$ RMSE of the variance estimator from Lemma \ref{lemma:varest} with respect to the empirical variance of the fully adjusted estimator in our simulation study.}
\label{table:variance rmse}
\end{table}

\subsection{Asymptotic Normality and Asymptotic Variance}\label{app:normality}

Here we aim to assess the convergence of $\sqrt{N}(\hat{\tau}_N(\pi,\eta)-\tau_N(\pi,\eta))$ to a normal distribution for the three examples in which our theory claims asymptotic normality, that is, the family networks, $2$-d lattices, and the Erd{\H{o}}s--R{\'e}nyi networks $I(N,10/N)$. To do so, given an interaction network graph $I^N$, we compute the Shapiro-Wilk Normality test \citep{shapiro1965analysis} for the $\texttt{nrep.data}=100$ scaled estimators $\sqrt{N}\left(\hat{\tau}_N(\pi,\eta)-\tau_N(\pi,\eta)\right)$, giving us a $p$-value for each of the $\texttt{nrep.graph}$ networks $I^N$. Under the null hypothesis that the scaled estimator $\sqrt{N}\left(\hat{\tau}_N(\pi,\eta)-\tau_N(\pi,\eta)\right) $ is normally distributed, the distribution of the $p$-values is Unif$(0,1)$. We plot the empirical distribution functions (ecdfs) of the $\texttt{nrep.data}$ $p$-values in dark gray for the smallest and the largest sample size $N$. In addition, we add the ecdf of $100$ samples of $\texttt{nrep.data}$ draws of a Unif$(0,1)$-distribution in light gray. The results are shown in Figure \ref{ecdf:family} for the family networks, in Figure \ref{ecdf:2dlatt} for the $2$-d lattices, and in Figure \ref{ecdf:const} for the Erd{\H{o}}s--R{\'e}nyi networks $I(N,10/N)$. We observe that the ecdfs of the $p$-values of the normality test seems to converge to the ecdf of a Unif$(0,1)$-distribution as $N$ grows.

In addition to verify the results of Lemma \ref{lemma:varest} we computed the scaled to sample size empirical RMSE of the asymptotic variance estimator from Lemma \ref{lemma:varest} and the empirical variance across all repetitions for each graph-types and sample sizes. We summarize the results in Table \ref{table:variance rmse} and they confirm that the variance estimator from Lemma \ref{lemma:varest} consistently estimates the asymptotic variance of our estimator.

\begin{figure}[H]
\centering

\begin{subfigure}[b]{0.85\textwidth}
\includegraphics[width=1\linewidth]{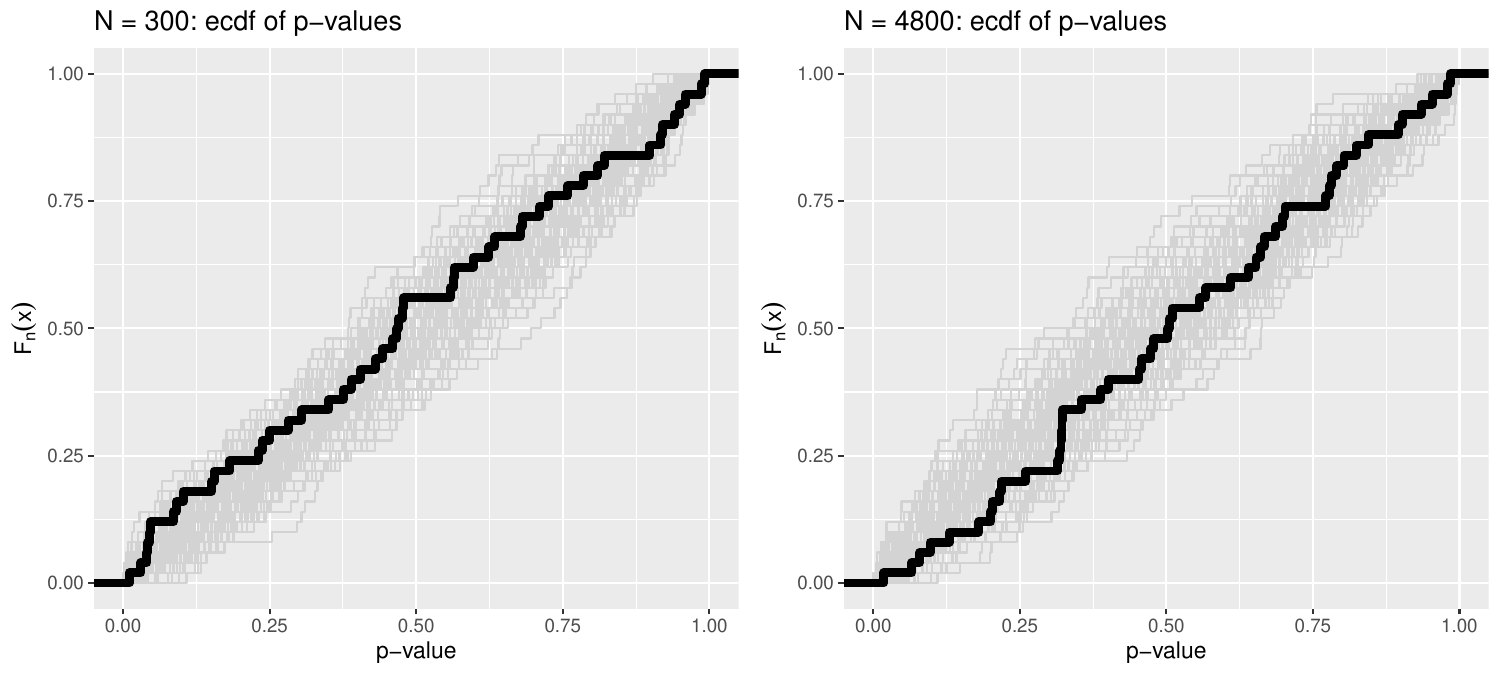}
\caption{}
\label{ecdf:family}
\end{subfigure}
\vfill
\begin{subfigure}[b]{0.85\textwidth}
\includegraphics[width=1\linewidth]{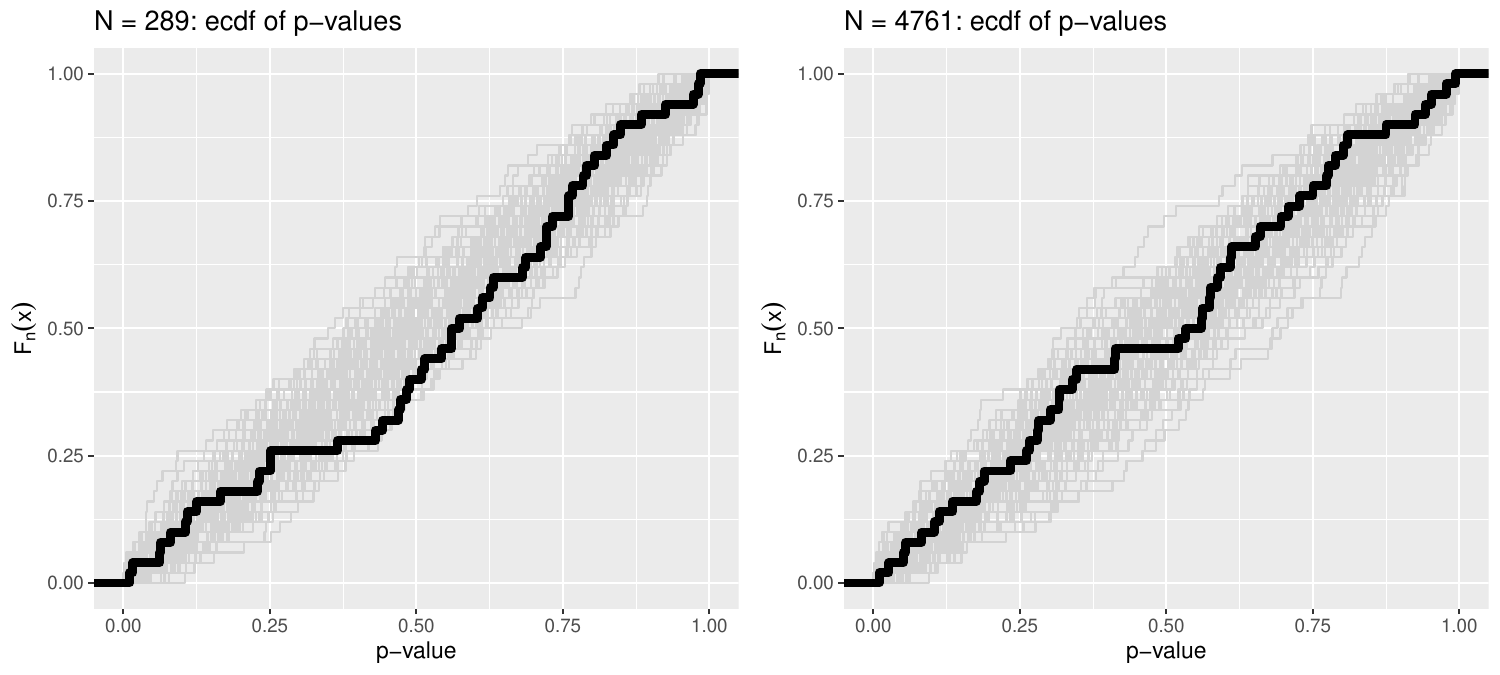}
\caption{}
\label{ecdf:2dlatt}
\end{subfigure}
\vfill
\begin{subfigure}[b]{0.85\textwidth}
\includegraphics[width=1\linewidth]{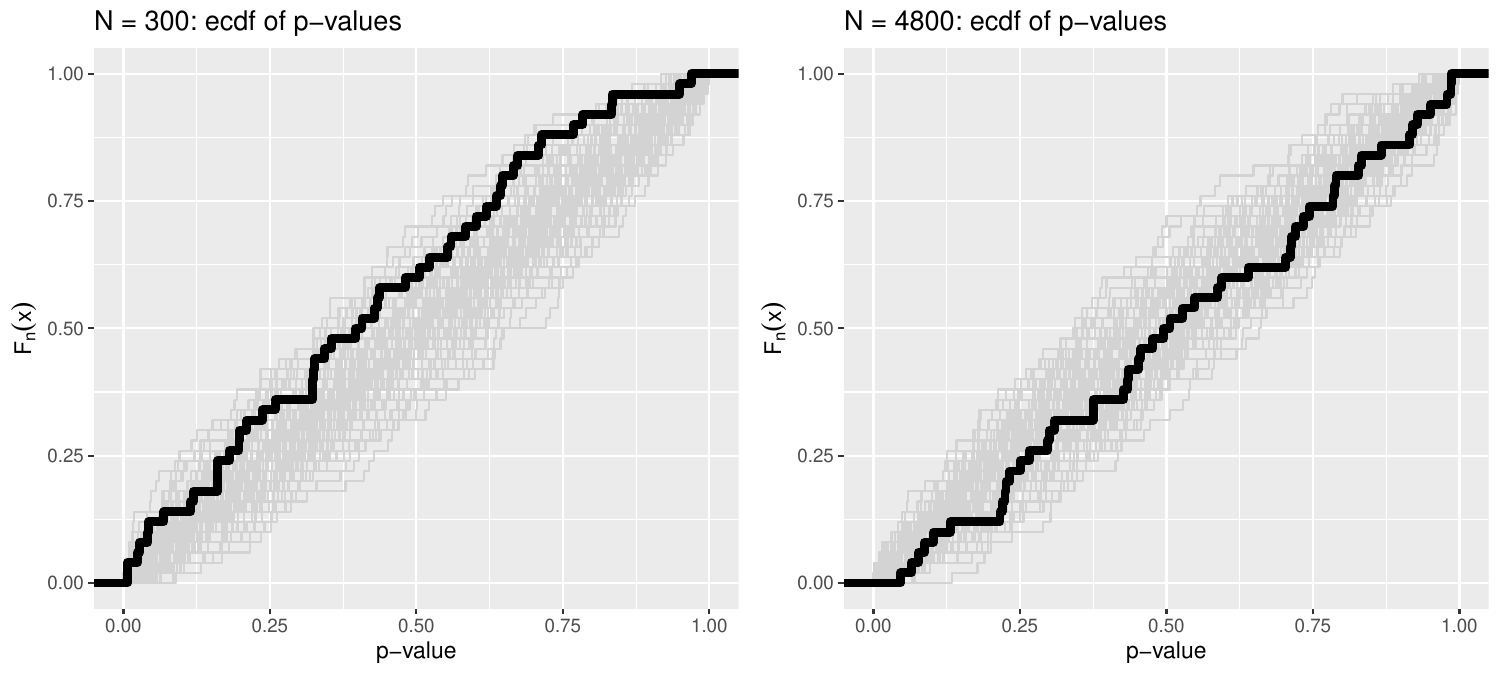}
\caption{}
\label{ecdf:const}
\end{subfigure}
\caption{Empirical distribution functions of the p-values from a Shapiro-Wilk Normality of $\sqrt{N}(\hat{\tau}_N(\pi,\eta)-\tau_N(\pi,\eta))$ in \subref{ecdf:family} family networks, \subref{ecdf:2dlatt} $2$-d lattices and \subref{ecdf:const} Erd{\H{o}}s--R{\'e}nyi networks with parameters $I(N,10/N)$ using the fully adjusted estimator.}
\end{figure}

\end{document}